\documentclass[
  aps,
  prx,
  reprint,
  superscriptaddress,
  amsmath,
  amssymb,
  nofootinbib,
  longbibliography
]{revtex4-2}

\usepackage{graphicx}
\usepackage{bm}
\usepackage{mathtools}
\usepackage{amsthm}
\usepackage{algpseudocode}
\usepackage{hyperref}
\usepackage{physics}
\usepackage{color}
\usepackage{mathrsfs}

\newcounter{algorithm}
\renewcommand{\thealgorithm}{\arabic{algorithm}}

\providecommand{\theHalgorithm}{}
\renewcommand{\theHalgorithm}{main.algorithm.\arabic{algorithm}}

\newenvironment{inlinealgorithm}[2]{%
  \par\medskip
  \refstepcounter{algorithm}%
  \label{#1}%
  \noindent\textbf{Algorithm~\thealgorithm. #2.}%
  \par\smallskip
}{%
  \par\medskip
}

\newtheorem{theorem}{Theorem}
\newtheorem{proposition}[theorem]{Proposition}
\newtheorem{lemma}[theorem]{Lemma}
\newtheorem{corollary}[theorem]{Corollary}

\newtheorem{definition}[theorem]{Definition}

\newtheorem{problem}[theorem]{Problem}

\newtheorem*{lemma*}{Lemma}
\newtheorem*{proposition*}{Proposition}
\newtheorem*{theorem*}{Theorem}
\newtheorem*{corollary*}{Corollary}
\newtheorem*{problem*}{Problem}

\newcommand{\supp}{\operatorname{supp}}
\newcommand{\Span}{\operatorname{span}}

\newcommand{\F}{\mathsf{F}}
\newcommand{\D}{\mathsf{D}}
\newcommand{\Pc}{\mathcal{P}}
\newcommand{\Fc}{\mathbb{F}_2}

\begin{document}

\title{Sample-Efficient Tomography of a Class of Mixed States 
\\
with Extensive Entanglement and Magic}

\author{Pengcheng Liao}
\email{liaopeng@usc.edu}
\affiliation{Ming Hsieh Department of Electrical and Computer Engineering, University of Southern California, Los
Angeles, California 90089, USA}

\author{Quntao Zhuang}
\email{qzhuang@usc.edu}
\affiliation{Ming Hsieh Department of Electrical and Computer Engineering, University of Southern California, Los
Angeles, California 90089, USA}
\affiliation{Department of Physics and Astronomy, University of Southern California, Los
Angeles, California 90089, USA}

\date{\today}

\begin{abstract}

Full tomography of a generic many-qubit quantum state requires exponentially
many copies, while suitable structural constraints can make reconstruction
sample-efficient. Existing approaches exploit, for example, limited
entanglement structure, low magic, or constrained
state-preparation circuits. Here we consider a class of mixed states that can
simultaneously exhibit extensive entanglement and extensive magic.
Specifically, we introduce Clifford-encoded block-product (CEBP) states,
obtained by applying an unknown global Clifford unitary to a tensor product of
arbitrary mixed states supported on unknown blocks of bounded size. We show that, for a fixed block size, CEBP states can be reconstructed using polynomially many copies by exploiting Clifford-preserved Pauli correlations to recover the latent structure and reduce the remaining problem to local tomography. Our result demonstrates that sample-efficient tomography can arise from bounded complexity in a latent frame even when the physical state is highly entangled, highly magical, and mixed, and suggests complexity modulo structured transformations as a broader organizing principle for quantum-state learnability.


\end{abstract}

\maketitle

\section{Introduction}

Quantum state tomography reconstructs a classical description of an unknown
quantum state from measurements on independently prepared copies, enabling
state characterization, device validation, resource quantification, and
prediction of measurement statistics~\cite{Hradil1997,JamesKwiatMunroWhite2001,ParisRehacek2004,AltepeterJamesKwiat2004,BlumeKohout2010,ChristandlRenner2012}.
For unrestricted $n$-qubit states, however, reconstruction to fixed accuracy
requires exponentially many copies~\cite{haah2017sample,ODonnellWright2016}.
Sample-efficient tomography therefore relies on exploitable structure that
restricts the state class.

Existing approaches to tomography exploit different restrictions on the unknown states. Product-state tomography and tensor-network methods exploit limited correlation structure, with matrix product state and matrix product operator approaches using bounded bond dimension~\cite{Cramer2010Efficient,Baumgratz2013Scalable,Lanyon2017Efficient}. Stabilizer-based methods instead exploit restricted magic, enabling efficient learning of stabilizer states and suitable extensions involving few $T$ gates or few local non-Clifford gates~\cite{Montanaro2017,lai2022learning,leone2024,grewal2023efficient,HangleiterGullans2024}. Beyond these restricted-resource regimes, Zhao et al.\ establish sample-efficient learning of pure states with bounded preparation-gate complexity~\cite{ZhaoEtAl2024BoundedGate}. Their result accommodates substantial entanglement and magic but relies on purity. Under the additional restriction of constant-depth local preparation on a finite-dimensional lattice, Landau and Liu obtain polynomial-time reconstruction of pure states~\cite{LandauLiu2025Shallow}.

Beyond these regimes, a central question is: \emph{can highly mixed states
with extensive entanglement and magic admit sample-efficient tomography, without a constant-depth
preparation assumption?}

We provide an affirmative answer for a structured family of
\emph{Clifford-encoded block-product} (CEBP) states,
$\rho=U_{\mathrm c}(\bigotimes_{a=1}^{K}\rho_a)U_{\mathrm c}^{\dagger}$,
where $U_{\mathrm c}$ is an unknown global Clifford unitary and the unknown
states $\rho_a$ are supported on disjoint blocks of at most $d$ qubits.
The blocks may be arbitrary mixed states, and neither the encoder nor the
partition is supplied to the learner. The only structural prior is the
upper bound $d$ on the hidden block size; no restriction is imposed on the
depth of the encoding circuit.
CEBP states can exhibit extensive entanglement and magic: Clifford encoding
can generate entanglement across physical bipartitions while preserving
magic supplied by the block inputs. Allowing mixed inputs also accommodates
extensive entropy and exponentially large global rank. The class therefore
extends beyond low-bond-dimension descriptions in a fixed physical ordering
and the few-non-Clifford regimes discussed above. The relevant restriction
is the bounded block size in a hidden Clifford frame, rather than small
entanglement, limited non-Clifford resources, or purity.

\begin{figure*}[t]
  \centering
    \includegraphics[width=0.99\linewidth]{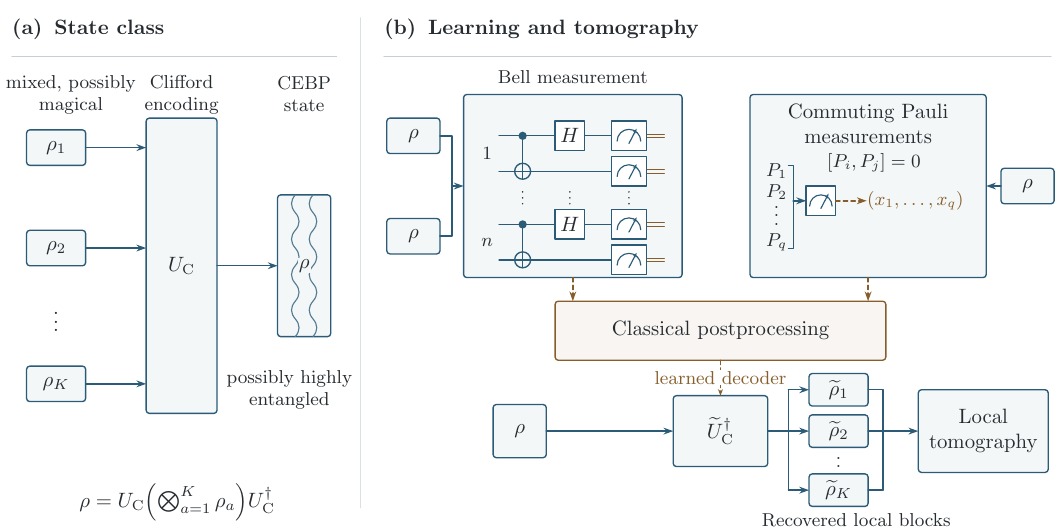}
    \caption{
Schematic of the CEBP state class and learning protocol.
(a) A CEBP state is obtained by applying an unknown Clifford unitary $U_{\mathrm C}$ to a product of local blocks $\rho_a$, each supported on at most $d$ qubits. The local blocks may be mixed and magical, while the Clifford encoding can generate a highly entangled global state.
(b) The learner combines Bell measurements on many pairs of copies with joint measurements of commuting Pauli tuples on fresh single copies. The former are used to estimate $\left|\Tr(\rho P)\right|^2$, while the latter provide the signed moments required for cumulant grouping. Classical postprocessing combines these data to construct a decoding Clifford $\widetilde U_{\mathrm C}^{\dagger}$. Applying the learned decoder to fresh copies exposes the recovered local blocks, which are then reconstructed by local tomography.
}  
    \label{fig:cebp-overview}
\end{figure*}

Our main result establishes tomography of CEBP states with copy complexity
polynomial in the system size and inverse target accuracy for every fixed
maximum block size $d$, with a high-probability trace-norm guarantee.
The reconstruction exploits the algebraic structure of Pauli correlations
preserved by Clifford encoding. Two-copy Bell measurements estimate squared
Pauli expectation values, whose blockwise factorization and commutation
structure enable recovery of suitable Pauli algebras. For multiqubit blocks,
additional single-copy measurements of commuting Pauli observables provide
mixed cumulants that guide the grouping of recovered sectors, with
controlled approximation error. A recovered Clifford then localizes the
relevant algebras onto bounded-size registers for local tomography.
Exact identification of the original encoder and partition is unnecessary.

The remainder of the paper is organized as follows.
Section~\ref{sec: problem} defines the CEBP state class and formulates the associated
tomography problem. Section~\ref{sec: protocol} presents our sample-efficient learning
protocol and establishes its end-to-end sample-complexity guarantee.
Section~\ref{sec: numerical} reports numerical simulations comparing the proposed CEBP
tomography protocol with local-Pauli linear inversion. Finally,
Section~\ref{sec: discussion} discusses the implications of our results and directions for
future work. Technical details and proofs are deferred to the Supplemental
Material.

\section{State Class and Learning Problem}
\label{sec: problem}

A state with complicated correlations in the physical basis may admit a
much simpler description in an appropriate Clifford frame. This perspective
appears in several existing constructions. Clifford encoders map logical
states into stabilizer codes~\cite{Gottesman1997}, while Clifford-enhanced
matrix product states combine global Clifford transformations with
low-bond-dimension cores for quantum state
designs~\cite{LamiHaugDeNardis2025} and many-body
simulation~\cite{QianHuangQin2024}. Clifford dynamics can likewise generate
nonlocal correlations and entanglement from simple initial
states~\cite{SommersHuseGullans2023,RichterLuntPal2023}. A related
continuous-variable construction is provided by Gaussian-entanglable states,
obtained by applying Gaussian processing to separable, potentially
non-Gaussian inputs, whose pure-state subclass admits sample-efficient
tomography~\cite{zhao2025}.

Here we specialize this viewpoint to a core that factorizes into blocks
of bounded size, without assuming pure inputs. Each block may contain
arbitrary correlations and mixedness; only correlations between distinct
blocks are excluded in the latent frame. Applying a global Clifford
unitary to such a core defines the Clifford-encoded block-product (CEBP)
state class.

For a hidden partition
$B_1\sqcup\cdots\sqcup B_K=[n]$, a CEBP state has the form
\begin{equation}
\rho
=
U_{\mathrm c}
\left(
\bigotimes_{a=1}^{K}\rho_a
\right)
U_{\mathrm c}^{\dagger},
\qquad
\rho_a\in \D\!\bigl((\mathbb C^2)^{\otimes |B_a|}\bigr),
\label{eq:clifford-encoded-block-product}
\end{equation}
where $U_{\mathrm c}$ is a Clifford unitary and the latent block states
$\rho_a$ may be arbitrary $|B_a|$-qubit states, as shown in Fig~\ref{fig:cebp-overview}(a).

Here is a simple way to see that CEBP states can simultaneously exhibit large-scale entanglement and extensive magic. For even
$n$, take single-qubit hidden blocks with latent product state
$\ket{T}^{\otimes n}$, where
\begin{equation}
\ket{T}:=T\ket{+}=(\ket 0+e^{i\pi/4}\ket 1)/\sqrt2 .
\end{equation}
Fix the physical ordering $1,\ldots,n$ and the midpoint cut
$[n/2]|[n]\setminus[n/2]$.  Let $U_{\mathrm c}$ be the product of CNOT gates
from qubit $i$ to qubit $n/2+i$ for $i=1,\ldots,n/2$.  Then
$U_{\mathrm c}\ket{T}^{\otimes n}$ is a tensor product of $n/2$ identical
two-qubit entangled states across the midpoint cut.  Hence its entanglement
entropy across this cut is $\Omega(n)$, and any exact MPS representation in
this fixed physical ordering requires bond dimension $\exp(\Omega(n))$.
At the same time, the order-two stabilizer R\'enyi entropy
obeys~\cite{LeoneOlivieroHamma2022SRE}
\begin{equation}
M_2(\ket{T}^{\otimes n})
=
M_2(U_{\mathrm c}\ket{T}^{\otimes n})
=
n\log(4/3),
\end{equation}
where additivity gives the linear scaling and Clifford invariance gives the
equality after encoding.  Thus, the CEBP class can contain states with both
extensive fixed-order entanglement and extensive magic even when the hidden
block size is one.  Consequently, fixed-order bounded-bond-dimension MPS
tomography and few-non-Clifford-gate methods do not directly apply to the full
CEBP class.

Tomography for this class is therefore nontrivial.  This
raises the central question: does the hidden bounded-block structure
nevertheless enable sample-efficient tomography?
\begin{problem}[Tomography of Clifford-encoded block-product states]
\label{prob:cebp-tomography}
Given independent copies of an unknown $n$-qubit state $\rho$, target
accuracy $\varepsilon\in(0,1)$, failure probability
$\delta\in(0,1)$, and a prior $d\in\mathbb N$, assume that
there exist a hidden partition $B_1\sqcup\cdots\sqcup B_K=[n]$ and an
unknown Clifford unitary $U_{\mathrm c}$ such that $\rho$ has the form
in Eq.~\eqref{eq:clifford-encoded-block-product}, with
$\max_{a\in[K]} |B_a|\le d$. 
The task is to output a classical description of an estimator $\widehat\rho$
such that
\begin{equation}
\Pr\!\left\{\|\widehat\rho-\rho\|_1\le\varepsilon\right\}
\ge 1-\delta.
\end{equation}
\end{problem}

We use the standard trace distance
$D_{\mathrm{tr}}(\omega,\sigma):=\frac12\|\omega-\sigma\|_1$.
The analytical parameter $\varepsilon$ specifies trace-norm error, while
numerical reconstruction errors are reported in trace distance. Thus the
trace-norm target $\varepsilon$ corresponds to
$D_{\mathrm{tr}}(\widehat\rho,\rho)\le\varepsilon/2$.

\noindent We emphasize that the learner is not given the hidden partition, the number of
blocks, the block locations, or a generating tableau or description of
$U_{\mathrm c}$.
The only structural prior is the upper bound $d$ on the largest
hidden block. Nevertheless, this limited prior is sufficient for sample-efficient tomography when $d$ is fixed, as summarized by our main result below.

Before we proceed any further, we give an informal result here. 

\begin{theorem}[Informal]
\label{thm:informal-cebp-tomography}
Problem~\ref{prob:cebp-tomography} can be solved with $O(n^{O(d)} 2^{O(d\log d)})$ copies of  $\rho$. 
\end{theorem}

Theorem~\ref{thm:informal-cebp-tomography} establishes that the hidden block-product structure remains learnable even after an unknown global Clifford transformation has obscured the underlying tensor-product decomposition. In particular, the copy complexity remains polynomial in $n$ for every fixed $d$, despite the fact that the resulting physical states can simultaneously possess extensive entanglement and magic.

Figure~\ref{fig:cebp-overview}(b) summarizes our tomography protocol.  To uncover this hidden structure, the learner combines Bell measurements on many copy pairs with commuting-Pauli measurements on fresh single copies, the latter supplying the signed moments needed for cumulant grouping.  Classical postprocessing uses these data to construct an effective decoding Clifford $\widetilde U_{\mathrm c}^{\dagger}$; applying it to fresh copies exposes the recovered local blocks, which can then be estimated by local tomography.  The next section develops these steps and their sample complexity in detail.

\section{Sample-Efficient Learning Protocol}
\label{sec: protocol}
We now give a constructive learning strategy for CEBP states and write
$d$ as the learner-known valid hidden-block-size cap.
There are two ingredients underlying the protocol: squared Pauli scores factor across the
unknown hidden blocks, and a transversal two-copy Bell measurement estimates
all such scores simultaneously.  From these foundations, the operational
stages are certified stabilizer peeling, rank-guided Pauli-algebra recovery,
bounded-order cumulant grouping, Clifford localization, and local tomography.
Detailed certification conditions, calibrated parameters, copy ledgers, and
proofs are deferred to Supplemental Material (SM),
Secs.~\ref{sec:supp-notation}--\ref{sec:supp-end-to-end}.

\subsection{Score factorization in a hidden Pauli frame}

Let $\Pc_n=\{I,X,Y,Z\}^{\otimes n}$ denotes the phase-free Pauli $n$-qubit Pauli group. For $P\in\Pc_n$, define its squared Pauli score by
\begin{equation}
s_\rho(P):=\left|\Tr(\rho P)\right|^2.
\end{equation}
For the CEBP representation in
Eq.~\eqref{eq:clifford-encoded-block-product}, write
\begin{equation}
U_{\mathrm c}^\dagger P U_{\mathrm c}
=
\chi(P)\bigotimes_{a=1}^K P_a,
\qquad
P_a\in\Pc_{B_a},
\qquad
\chi(P)\in\{\pm1\}.
\end{equation}
The latent product across blocks then gives
\begin{equation}
s_\rho(P)
=
\prod_{a=1}^K s_{\rho_a}(P_a),
\qquad
s_{\rho_a}(P_a):=\left|\Tr(\rho_aP_a)\right|^2.
\label{eq:main-transformed-pauli-factorization-intuition}
\end{equation}
Only cross-block factorization is asserted: each $\rho_a$ may contain
arbitrary correlations and entanglement within $B_a$. Because each block score lies in $[0,1]$, the score of a nonidentity block component is at least as large as the score of any composite Pauli containing that component. This monotonicity provides the basic signal for recovery.

To see the signal in its simplest form, temporarily take single-qubit hidden
blocks, so the latent state is $\bigotimes_{i=1}^n\rho_i$.  The unknown
Clifford transports the latent Pauli directions to physical Pauli strings
\begin{equation}
Q_i^\alpha
:=
U_{\mathrm c}\sigma_i^\alpha U_{\mathrm c}^\dagger,
\qquad
\alpha\in\{x,y,z\}.
\end{equation}
Every $P\in\Pc_n$ can be written, up to sign, as
$\prod_iQ_i^{\alpha_i}$ with $\alpha_i\in\{0,x,y,z\}$ and $Q_i^0=I$.
Equation~\eqref{eq:main-transformed-pauli-factorization-intuition} becomes
\begin{equation}
s_\rho(P)
=
\prod_{i=1}^n
\left|\Tr\!\left(\rho_i\sigma_i^{\alpha_i}\right)\right|^2,
\qquad
\sigma_i^0:=I_i,
\end{equation}
so high-scoring Paulis form a natural candidate pool for the transported
single-qubit frame.

Scores alone do not identify those latent Pauli directions.  A composite Pauli may
outrank an unrelated single-sector Pauli even though it cannot outrank its own
nonidentity components.  The protocol therefore also uses the preserved
symplectic structure: the three Pauli directions of one latent qubit anticommute pairwise,
whereas Pauli directions associated with distinct latent qubits commute.  
Rank-guided compatibility tests combine this commutation structure with score ordering to distinguish Pauli directions supported within a single latent block from composite directions spanning multiple blocks.

A second ambiguity comes from exact stabilizers.  If $S$ has unit score and
$\odot$ denotes phase-free Pauli multiplication, then stabilizer dressing
leaves all scores invariant, $s_\rho(S\odot P)=s_\rho(P)$.  
A Pauli direction with empirical score close to one need not be an exact stabilizer, yet it can still create an approximate degeneracy in the score landscape and interfere with subsequent recovery. The protocol therefore peels off such near-stabilizer directions while explicitly controlling the approximation error introduced by this removal.

\subsection{Two-copy Bell-sampling score estimation}

\begin{figure*}[!t]
  \centering
    \includegraphics[width=0.99\linewidth]{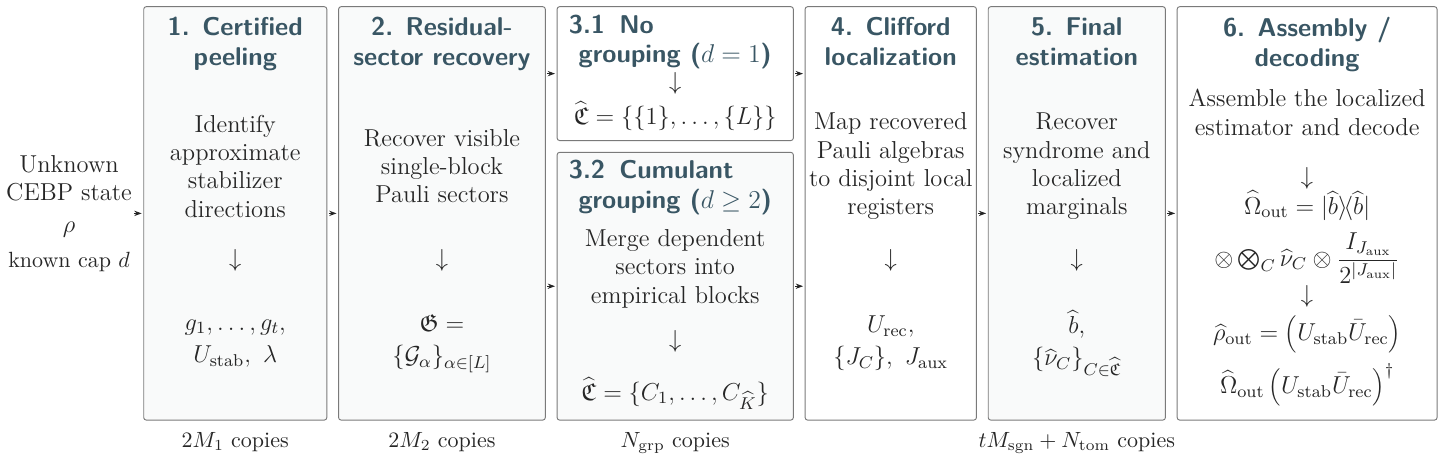}
    \caption{Schematic of the CEBP tomography protocol. Starting from an unknown CEBP state
$\rho$ and a known upper bound $d$ on the hidden block size, the learner first
peels approximate stabilizer directions, recovering
$g_1,\ldots,g_t$, $U_{\mathrm{stab}}$, and $\lambda$, and then performs
rank-guided recovery of the visible single-block Pauli sectors
$\mathfrak G$. For $d=1$, each recovered sector forms its own empirical block,
whereas for $d\ge2$ cumulant grouping merges dependent sectors to obtain the
empirical partition $\widehat{\mathfrak C}$. A simultaneous Clifford
localization maps the recovered Pauli algebras to disjoint registers
$\{J_C\}$ and $J_{\mathrm{aux}}$. Final estimation recovers the syndrome
$\widehat b$ and the localized marginals $\{\widehat\nu_C\}$, which are
assembled into $\widehat\Omega_{\mathrm{out}}$ and decoded through
$U_{\mathrm{stab}}\bar U_{\mathrm{rec}}$ to obtain
$\widehat\rho_{\mathrm{out}}$. The fresh-copy budgets used by the measurement
stages are indicated below the corresponding boxes.}  
    \label{fig:cebp-protocol-schematic}
\end{figure*}

In one Bell round, take two independent copies $\rho\otimes\rho$ and perform a
Bell-basis measurement transversally on corresponding qubit pairs.  Equivalently,
for each $i\in[n]$, jointly measure the commuting observables
\begin{equation}
\sigma_i^{x,(1)}\sigma_i^{x,(2)},\qquad
\sigma_i^{y,(1)}\sigma_i^{y,(2)},\qquad
\sigma_i^{z,(1)}\sigma_i^{z,(2)}.
\end{equation}
This measurement requires neither access to $\rho^*$ nor knowledge of the
encoding Clifford $U_{\mathrm c}$.  Its use as a uniform score-estimation primitive
builds on Bell sampling for stabilizer learning~\cite{Montanaro2017}; its
simultaneous-query role is also related to broader observable-prediction
methods~\cite{HuangKuengPreskill2021}.

Let
\begin{equation}
B_{i,j}^{\alpha}\in\{-1,+1\},
\quad
i\in[n],\quad \alpha\in\{x,y,z\},\quad j\in[M],
\end{equation}
be the recorded eigenvalues from $M$ rounds.  For
$P=\bigotimes_iP_i\in\Pc_n$, use the labelled support
$\supp(P):=\{(i,\alpha):P_i=\sigma_i^\alpha\}$ and define
\begin{equation}
b_j(P)
:=
\prod_{(i,\alpha)\in\supp(P)}B_{i,j}^{\alpha},
\qquad
\hat s(P)
:=
\frac1M\sum_{j=1}^M b_j(P).
\end{equation}
First, this estimator is unbiased. 
\begin{lemma}[Expectation of the Bell-sampling score product]
\label{lem:main-bell-product-expectation}
For every $P\in\Pc_n$,
\begin{equation}
\mathbb E[\hat s(P)]=s_\rho(P).
\end{equation}
\end{lemma}
\noindent Second, one can estimate the score of all $P\in\Pc_n$ with polynomial samples. 
\begin{proposition}[Uniform Bell-sampling score concentration]
\label{prop:main-uniform-bell-score-concentration}
For $M$ Bell rounds and $\zeta\in(0,1)$, with probability at least
$1-\zeta$,
\begin{equation}
\begin{aligned}
\left|\hat s(P)-s_\rho(P)\right|
&\le\tau(M,n,\zeta),
&&\forall P\in\Pc_n,\\
\tau(M,n,\zeta)
&:=\sqrt{\frac{2\log(2\cdot4^n/\zeta)}{M}}.
\end{aligned}
\label{eq:main-bell-uniform-guarantee}
\end{equation}
\end{proposition}
Thus one stored record supports every later fixed or adaptively selected
score query on the same state, with scores evaluated on demand.
Additive accuracy $\eta$ needs
$M=O((n+\log(1/\zeta))/\eta^2)$ Bell rounds.  The proof is in the SM,
Sec.~\ref{sec:supp-bell-estimator}.

Measurement records are reused only for queries concerning the same fixed state. Whenever the protocol changes the state or begins a stage whose analysis requires independent data, it uses fresh copies.
For example, we have $2M_1$ copies reserved for peeling.  After the
peeling Clifford is fixed, fresh copies are transformed by
$U_{\mathrm{stab}}^\dagger$, and the full peeled state is retained with its
syndrome-prefix and residual coordinates identified.  An independent $M_2$
recovery Bell record is collected on that full state.  Cumulants, syndrome
signs, and empirical-register tomography use separate fresh-copy pools.  A
Bell round consumes two state copies, so $M_1$ and $M_2$ rounds contribute
$2M_1$ and $2M_2$ copies, respectively.

\subsection{Sketch of the algorithm}

The protocol has five operational learning stages: certified peeling,
residual-sector recovery, cumulant grouping, Clifford localization, and final estimation. 
Figure~\ref{fig:cebp-protocol-schematic} summarizes the protocol flow,
the learner-visible outputs at each stage, and the corresponding fresh-copy
budgets.
Measurement stages use disjoint fresh-copy pools and never reuse
data when their guarantees require independence.  Deterministic assembly and
Clifford decoding follow those stages and consume no further copies.  The recovered
Cliffords, sectors, empirical groups, registers, syndrome estimate, and local
estimates are learner-visible; hidden assignments, aggregate true-block
registers, and comparison states are proof-only.

The $M_1$ Bell record estimates the Pauli scores of the original state. From
the near-unit-score directions, certified peeling selects a commuting basis
$g_1,\ldots,g_t$ satisfying $s_\rho(g_j)\ge 1-\lambda$. These operators are
approximate stabilizers of $\rho$. The learner then constructs a Clifford
$U_{\mathrm{stab}}$ that maps this basis to standard $Z$ operators, thereby
separating the syndrome and residual coordinates. Writing $m:=n-t$, we have
\begin{equation}
U_{\mathrm{stab}}^\dagger g_jU_{\mathrm{stab}}=Z_j
\quad (j\in[t]),
\qquad
\widetilde\rho
:=
U_{\mathrm{stab}}^\dagger\rho U_{\mathrm{stab}}.
\end{equation}

On the peeling event, there exist a syndrome $b\in\{0,1\}^t$ and a residual
state $\rho_{\mathrm{res}}$ on the remaining $m$ qubits such that, with
$\varepsilon_{\mathrm{peel}}:=t\lambda/2$,
\begin{equation}
\left\|
\widetilde\rho
-
|b\rangle\!\langle b|\otimes\rho_{\mathrm{res}}
\right\|_1
\le
2\sqrt{\varepsilon_{\mathrm{peel}}}
=
\sqrt{2t\lambda}.
\end{equation}
Thus the quality of the syndrome--residual decomposition is controlled by the
peeling rank $t$ and the accepted score deficit $\lambda$. 
Moreover, every Pauli operator $R$  on the remaining $m$ qubits satisfies
$s_{\widetilde\rho}(R)\le 1-\lambda$. Consequently, no 
residual Pauli direction has a near-unit score, removing the stabilizer-score
degeneracy that would otherwise obstruct rank-guided residual recovery.

An independent $M_2$ Bell record on $\widetilde\rho$ drives rank-guided
recovery.  The learner ranks all operators of form $Z^c\otimes P$ based on their estimated score, where
$P$ is the residual action on $m$ remaining qubits and $Z^c$ is its syndrome-prefix lift in the first $t$ qubits. 
The learner then performs compatibility tests that preserve the required commutation relations and remove already represented components.
When the score estimation is accurate enough,  each transformed single-block
component precedes any composite containing it because of the hidden block-product structure.

As a result, the recovery stage outputs a sector family
$\mathfrak G=\{\mathcal G_\alpha\}_{\alpha\in[L]}$, where each sector is either a
Pauli triple
$\mathcal G_\alpha=\{\widehat R_\alpha^x,\widehat R_\alpha^y,
\widehat R_\alpha^z\}$ or a singleton
$\mathcal G_\alpha=\{\widehat R_\alpha^x\}$. Operators belonging to distinct
sectors commute, whereas the three operators within each Pauli triple sector
pairwise anticommute. On the successful recovery event, the sectors are
block-pure: all operators in a given sector are transformed Pauli operators
originating from the same hidden block. The recovered family is also
thresholded-complete, in the sense that every nonidentity residual Pauli
direction whose score exceeds the lower recovery threshold lies in the span
generated by the recovered sectors.

The recovered sectors identify individual latent Pauli directions but not yet which sectors belong to the same hidden block. 
For $d\ge2$, the learner must determine which recovered sectors originate
from the same hidden block. We use mixed cumulants of commuting Pauli
observables as the grouping signal
\cite{doubilet1972foundations,smith2020multivariateKstatistics}. A mixed
cumulant vanishes when the tested observables factor across independent
components, whereas a nonzero cumulant certifies residual dependence among
the corresponding sectors. Bell sampling cannot estimate these cumulants
directly because it reveals only squared Pauli expectations and therefore
discards the signs required by the moment--cumulant relation. Instead, the
learner uses a fresh $N_{\mathrm{grp}}$ pool and jointly measures each queried
commuting Pauli tuple on ordinary single copies of the peeled state. The
grouping routine merges clusters connected by a witnessed cumulant and
rescans all admissible orders after every merge.

The routine outputs a learner-visible partition of the recovered sector family,
\begin{equation}
\widehat{\mathfrak C}
=
\{C_1,\ldots,C_{\widehat K}\},
\quad
[L]
=
\bigsqcup C_i.
\end{equation}
Each cluster $C_i\in\widehat{\mathfrak C}$ consists of recovered sector indices
originating from a single hidden block, so the resulting partition contains no
cross-block merges. The partition may nevertheless over-refine a hidden block.
In that case, the admissible mixed cumulants across distinct recovered clusters
are bounded by the calibrated grouping scale, which quantitatively controls any
within-block dependence left unmerged.
The exact construction and calibrations are in the SM,
Sec.~\ref{sec:supp-cumulant-grouping}.
For $d=1$, grouping is unnecessary: each recovered sector forms its own empirical cluster.

Once the empirical partition $\widehat{\mathfrak C}$ is determined, each empirical cluster represents one candidate localized block. The learner next maps its recovered Pauli algebra onto a physical register $J_C$.
Each Pauli triple sector already supplies a symplectic pair, while each singleton
sector is augmented by a compatible conjugate partner. These partners are
chosen simultaneously so that commutation between distinct empirical clusters
is preserved. Since the recovered sector family is only
thresholded-complete, its generated Pauli groups need not span the entire
residual Pauli space. The remaining symplectic coordinates are therefore
placed in an auxiliary register $J_{\mathrm{aux}}$, giving the disjoint
decomposition
\begin{equation}
[m]
=
\left(
\bigsqcup_{C\in\widehat{\mathfrak C}}J_C
\right)
\sqcup J_{\mathrm{aux}},
\qquad
|J_C|\le d.
\end{equation}

For each empirical cluster $C\in\widehat{\mathfrak C}$, define the phase-free
Pauli group generated by its recovered sectors as
\begin{equation}
\mathcal P(C)
:=
\left\langle
\bigcup_{\alpha\in C}\mathcal G_\alpha
\right\rangle.
\end{equation}
The simultaneous symplectic completion yields a residual Clifford
$U_{\mathrm{rec}}$ such that
\begin{equation}
U_{\mathrm{rec}}^\dagger
\mathcal P(C)
U_{\mathrm{rec}}
\subseteq
\Pc(J_C)
\qquad
\text{for every }C\in\widehat{\mathfrak C}.
\end{equation}
Moreover, the added conjugate partners complete the recovered directions to
the full Pauli group on $J_C$. Thus a single global Clifford localizes all
empirical clusters onto pairwise disjoint physical registers.

After the recovered sectors are converted into small and pairwise disjoint physical registers, local tomography can be performed.
Let us define
\begin{equation}
\bar U_{\mathrm{rec}}
:=
I^{\otimes t}\otimes U_{\mathrm{rec}},
\qquad
\rho_{\mathrm{loc}}
:=
\bar U_{\mathrm{rec}}^\dagger
U_{\mathrm{stab}}^\dagger
\rho
U_{\mathrm{stab}}
\bar U_{\mathrm{rec}}.
\label{eq:mt-computable-localized-state}
\end{equation}
The combined action of $U_{\mathrm{stab}}$ and $U_{\mathrm{rec}}$ transforms $\rho$ into an approximate block-product state
\begin{equation}
\rho_{\mathrm{loc}} \approx
\Omega_{\mathrm{emp}}
:=
|b\rangle\!\langle b|
\otimes
\left(
\bigotimes_{C\in\widehat{\mathfrak C}}\omega_C
\right)
\otimes
\frac{I_{J_{\mathrm{aux}}}}{2^{|J_{\mathrm{aux}}|}}, 
\label{eq:mt-localized-emp-block-model}
\end{equation}
where
$\omega_C\in\D\!\left((\mathbb C^2)^{\otimes |J_C|}\right)$.
The learner then performs local tomography on the actual
marginals
\begin{equation}
\nu_C
:=
\Tr_{\overline{J_C}}\!\left(\rho_{\mathrm{loc}}\right),
\qquad
C\in\widehat{\mathfrak C},
\end{equation}
and obtains physical estimates $\widehat\nu_C$. An independent pool of signed
measurements similarly recovers the syndrome estimate $\widehat b$. The
localized estimator is then
\begin{equation}
\widehat\Omega_{\mathrm{out}}
:=
|\widehat b\rangle\!\langle\widehat b|
\otimes
\left(
\bigotimes_{C\in\widehat{\mathfrak C}}
\widehat\nu_C
\right)
\otimes
\frac{I_{J_{\mathrm{aux}}}}{2^{|J_{\mathrm{aux}}|}}.
\end{equation}
Because the registers $J_C$ are pairwise disjoint, one can measure one chosen local Pauli observable on every register
simultaneously in each round. Consequently, the local tomography cost is governed by the
largest register size $\max_C|J_C|\le d$, rather than by the sum of all
register sizes.
Finally, the learner decodes the localized estimator through the recovered
Cliffords:
\begin{equation}
\widehat\rho_{\mathrm{out}}
:=
\left(
U_{\mathrm{stab}}\bar U_{\mathrm{rec}}
\right)
\widehat\Omega_{\mathrm{out}}
\left(
U_{\mathrm{stab}}\bar U_{\mathrm{rec}}
\right)^\dagger.
\end{equation}
The complete procedure is summarized in Alg.~\ref{alg:block-product-main-tomography}.
\begin{inlinealgorithm}{alg:block-product-main-tomography}{Tomography of Clifford-encoded block-product states}
\begin{algorithmic}[1]
\Require Independent-copy access to $\rho$; $n$, known valid cap $d$, target
trace-norm accuracy $\varepsilon$, and failure probability $\delta$.
\Ensure Either \textsc{failure} or a compact physical
estimator $\widehat\rho_{\mathrm{out}}$. 
\State Calibrate all thresholds and budgets from $(n,d,\varepsilon,\delta)$.
Reserve the common $2M_1,2M_2$ Bell pools and $N_{\mathrm{tom}}$ for local
tomography, where $N_{\mathrm{tom}}:=N_{1\mathrm q}$ for $d=1$ and
$N_{\mathrm{tom}}:=N_{\mathrm{bp}}$ for $d\ge2$.
Set $N_{\mathrm{grp}}=0$ for $d=1$; for $d\ge2$, reserve
$N_{\mathrm{grp}}$ fresh copies for grouping.
\State Use $M_1$ Bell rounds on $\rho$ to certify peeling and construct
$g_1,\ldots,g_t,U_{\mathrm{stab}}$; fail on a visible certification failure.
\State Reserve the realized $tM_{\mathrm{sgn}}$ sign
pool and set $m=n-t$.
\State Apply $U_{\mathrm{stab}}^\dagger$ to the second Bell pool and use its
$M_2$ rounds to recover a family $\mathfrak G$ of $L$ sectors.
\If{$d=1$}
  \State No grouping is needed; $\widehat{\mathfrak C}\gets \{\{1\},\dots,\{L\}\}$.
\Else
  \State Use $N_{\mathrm{grp}}$ fresh copies for
  calibrated cumulant merging and rescanning to obtain
  $\widehat{\mathfrak C}$.
\EndIf
\State Complete all groups simultaneously, construct $U_{\mathrm{rec}}$, and
compute $J_C,J_{\mathrm{aux}}$.

\State Use the reserved $tM_{\mathrm{sgn}}$ fresh copies to recover
$\widehat b$.
\State Estimate every marginal $\nu_C$  to obtain $\widehat\nu_C$  using $N_{\mathrm{tom}}$ fresh copies. 
\State Return 
$\widehat\rho_{\mathrm{out}}$
\end{algorithmic}
\end{inlinealgorithm}

\subsection{End-to-End Guarantee and Sample Complexity}

\begin{figure*}[!t]
    \centering
    \includegraphics[width=0.99\linewidth]{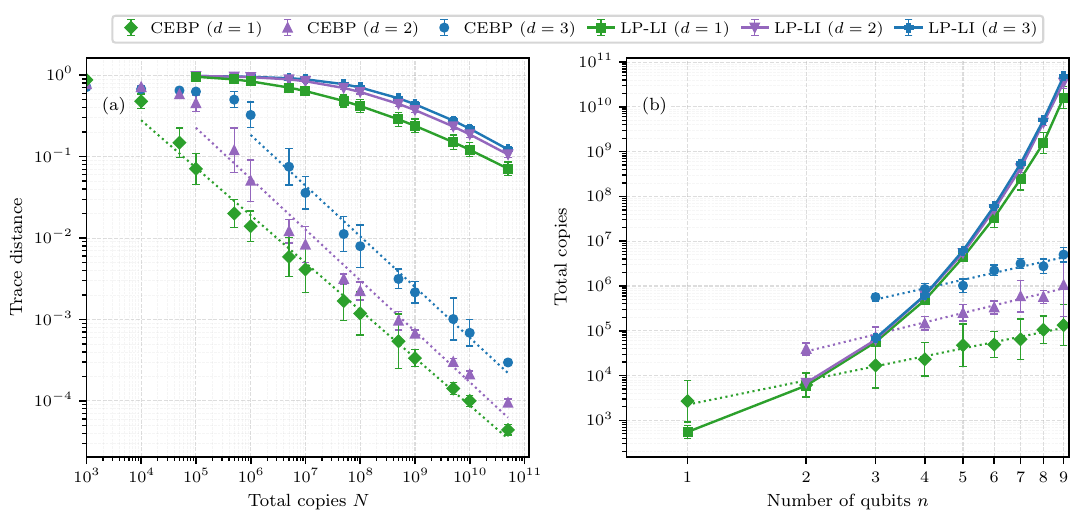}
\caption{
Sample complexity of CEBP tomography versus local-Pauli linear inversion (LP-LI).
For each $(n,d)$, we generate 20 independent CEBP target states.
The $n$ qubits are partitioned into blocks of maximum size $d$, with
$d=1$ corresponding to $(1,\ldots,1)$,
$d=2$ to $(2,\ldots,2,1/2)$ with a final one-qubit block when $n$ is odd,
and $d=3$ to $(3,\ldots,3,1/2/3)$ with a final block of size one or two when required.
For each block, the density matrix is sampled independently from the
Hilbert--Schmidt ensemble using a normalized complex Ginibre matrix~\cite{ZyczkowskiSommers2001}, and the
resulting block-product state is encoded by a random Clifford circuit generated
by a nonuniform $5n$-step random walk over Hadamard, phase, and CNOT gates.
(a) Trace-distance reconstruction error
$D_{\mathrm{tr}}(\widehat\rho,\rho):=\frac12\|\widehat\rho-\rho\|_1$ as a function of the
total number of copies $N$ for $n=10$ qubits and $d=1,2,3$.
CEBP results are shown as unconnected markers with error bars, while LP-LI
results are shown as solid connected curves.
Points show the geometric mean over completed runs, comprising 20 independently
generated target states with 17 independent holdout measurement seeds per state
for CEBP and 6 independent measurement replicates per state for LP-LI.
Error bars indicate one standard deviation in log space.
Dotted lines show power-law fits $D_{\mathrm{tr}}\propto N^{p}$ to the CEBP
data, using $N\geq 10^4$, $10^5$, and $10^6$ for $d=1,2,3$, respectively.
The fitted exponents and coefficients of determination
$(p,R^2_{\log_{10}})$ are
$(-0.584,0.992)$, $(-0.625,0.980)$, and $(-0.622,0.983)$
for $d=1,2,3$, respectively.
(b) Total number of copies required to achieve
$D_{\mathrm{tr}}\leq 0.05$ as a function of the number of qubits $n$, for
$d=1,2,3$.
Points show the geometric mean over 20 independently generated target states,
with error bars indicating one standard deviation in log space.
For the CEBP data, empirical power-law fits $N\propto n^{p}$ using all
available data points give
$(p,R^2_{\log_{10}})=(1.795,0.985)$, $(2.155,0.986)$, and $(1.974,0.942)$
for $d=1,2,3$, respectively.
All reported power-law regressions and $R^2$ values are evaluated in
$\log_{10}$ space.
Across the simulated regime, CEBP tomography requires substantially fewer
copies than LP-LI, while the CEBP sample complexity increases with the block
size $d$ at fixed system size.
}
\label{fig:cebp_vs_lpli}
\end{figure*}

By carefully choosing suitable parameters and thresholds, one can show  Alg.~\ref{alg:block-product-main-tomography} performs tomography for a CEBP state using a polynomial number of samples. 

\begin{theorem}[CEBP tomography]
\label{thm:main-calibrated-cebp-tomography}

Let $n,d\ge1$ and $0<\varepsilon,\delta<1$. Let $\rho$ be an $n$-qubit CEBP
state whose hidden blocks all have size at most $d$.

There exists a calibrated procedure implementing
Alg.~\ref{alg:block-product-main-tomography} such that, with probability at
least $1-\delta$, it returns an estimator $\widehat\rho_{\mathrm{out}}$
satisfying the trace-norm guarantee
\begin{equation}
\left\|
\widehat\rho_{\mathrm{out}}-\rho
\right\|_1
\le
\varepsilon.
\end{equation}

When $d=1$, the required number of copies of $\rho$ is
\begin{equation}
N_{\mathrm{tot}}^{(d=1)}
=
\widetilde O\!\left(
\frac{n^9}{\varepsilon^8}
\right).
\label{eq:main-calibrated-d-one-copies}
\end{equation}
For $d\ge2$, the required number of copies is
\begin{equation}
N_{\mathrm{tot}}^{(d\ge2)}
=
\widetilde O\!\left[
2^{O(d\log d)}
\left(
\frac{n^{19}}{\varepsilon^8}
+
\frac{n^{d+5}}{\varepsilon^2}
\right)
\right].
\label{eq:main-calibrated-d-ge-two-copies}
\end{equation}
Here $\widetilde O$ suppresses logarithmic factors in the displayed parameters
and in $\delta^{-1}$.

\end{theorem}

The detailed proof of the parameter choices is given in
Sec.~\ref{sec:supp-end-to-end}. The copy complexity is polynomial in $n$ for
every fixed $d$, with an explicit exponential dependence on the block-size cap
$d$. For $d=1$, the grouping stage and the associated splitting error are absent,
which leads to a lower copy complexity. For $d\ge2$, direct
bounded-order cumulant enumeration scales as $n^{O(d)}$ and is therefore
polynomial in $n$ for fixed $d$, although  it grows rapidly with the
block-size cap. We do not optimize the classical runtime here. We next examine
the protocol numerically through simulation.

\section{Numerical Simulation}
\label{sec: numerical}

In this section, we perform numerical simulations to compare the sample
complexity of our CEBP tomography protocol with local-Pauli linear inversion
(LP-LI) on randomly generated mixed CEBP states, as summarized in
Fig.~\ref{fig:cebp_vs_lpli}. Across the simulated parameter regime, CEBP
tomography requires substantially fewer copies than LP-LI to achieve the same
reconstruction accuracy. The numerical results also show a clear dependence on
the block size $d$: for fixed system size, larger $d$ generally requires more
copies to reach a given trace-distance error, which is consistent with our
analytical sample-complexity bounds.
Here we compare the two methods only in terms of the number of physical copies of the unknown state; other resources such as circuit depth, entangling-gate count, measurement-setting complexity, classical runtime, and memory are not normalized.

Figure~\ref{fig:cebp_vs_lpli}(a) examines the copy dependence of the
reconstruction accuracy at fixed system size $n=10$. For all three block sizes,
the trace-distance error decreases systematically as the total number of copies
$N$ increases, with CEBP maintaining a substantial advantage over LP-LI in the
low-error regime. For fixed $n$ and $d$, our rigorous worst-case
sample-complexity guarantee scales as
$N=\widetilde{O}(\varepsilon^{-8})$ in the trace-norm target $\varepsilon$.
Since $\varepsilon=2D_{\mathrm{tr}}$ for corresponding targets, this is
equivalently $N=\widetilde{O}(D_{\mathrm{tr}}^{-8})$ up to a constant factor.
To characterize the numerical behavior, we fit the CEBP data to
$D_{\mathrm{tr}}\propto N^{p}$ in the low-error regime, using
$N\geq 10^4$, $10^5$, and $10^6$ for $d=1,2,3$, respectively.
The resulting exponents are
$p=-0.584$, $-0.625$, and $-0.622$, with corresponding
$R^2_{\log_{10}}=0.992$, $0.980$, and $0.983$.
Equivalently, these fits correspond to
$N\propto D_{\mathrm{tr}}^{-p_D}$ with
$p_D\simeq 1.71$, $1.60$, and $1.61$ for $d=1,2,3$,
respectively. Thus, over the simulated accuracy range, the empirical dependence
on the target error is substantially milder than the worst-case
$D_{\mathrm{tr}}^{-8}$ analytical bound and is of comparable order to the
sampling-limited $D_{\mathrm{tr}}^{-2}$ scaling. This difference does not contradict
the rigorous guarantee, which is designed to hold uniformly in the worst case;
rather, it indicates that the bound is rather conservative for the random CEBP
instances considered here.

Figure~\ref{fig:cebp_vs_lpli}(b) instead fixes the target accuracy at
$D_{\mathrm{tr}}\leq 0.05$ and studies the copy dependence on the number of
qubits. This threshold corresponds to $\|\widehat\rho-\rho\|_1\le0.10$,
or the analytical trace-norm target $\varepsilon=0.10$.
Empirical power-law fits of the CEBP data to
$N_{\min}\propto n^{p_n}$, using all available system sizes for each $d$, yield
$p_n\simeq 1.80$, $2.16$, and $1.97$ for $d=1,2,3$, respectively, with
$R^2_{\log_{10}}=0.985$, $0.986$, and $0.942$.
These results indicate an approximately quadratic growth of the empirical CEBP
sample complexity over the accessible system sizes. This scaling is again
substantially milder than the worst-case analytical dependence, which scales as
$n^9$ for $d=1$ and is dominated by $n^{19}$ for fixed $d\geq 2$.
The similar fitted exponents for $d=1,2,3$ further suggest that, over the
simulated range, increasing the block size primarily increases the overall
sample requirement rather than strongly changing its observed scaling with
$n$. In contrast, LP-LI grows much more rapidly with system size over the
simulated range, leading to an increasing separation between the two methods.
These results demonstrate that exploiting the encoded block-product structure
can provide a substantial practical reduction in sample complexity.

\section{Discussion}
\label{sec: discussion} 

Our CEBP construction suggests a notion of \emph{genuine magical entanglement} that characterizes multipartite nonstabilizer correlations which cannot be removed by Clifford transformations.  For a pure state, we define the genuine magical-entanglement depth $d_{\mathrm{GME}}$ as  the smallest maximum block size among all possible CEBP representations of the state.   No Clifford transformation can further reduce the required maximum block size below $d_{\mathrm{GME}}$; the correlations remaining within the irreducible blocks are therefore not merely entanglement generated by Clifford processing, but involve intrinsically nonstabilizer structure.  In this sense, they may be viewed as a form of  magical entanglement.  In particular, $d_{\mathrm{GME}}=1$ means that all nonstabilizer resources can be localized onto individual qubits by a Clifford transformation.  For $d_{\mathrm{GME}}>1$, genuinely multiqubit nonstabilizer structure remains after the optimal Clifford reduction.  How to characterize and quantify this residual magical entanglement beyond the block-size diagnostic introduced here remains an open question.

This notion is related to, but distinct from, recently introduced notions of \emph{non-local magic}~\cite{ViscardiLeoneHamma2026NonlocalMagic,TorreFranchiniGiampaolo2026NonlocalMagic}.  Non-local magic quantifies the nonstabilizerness that cannot be removed by local unitaries with respect to a chosen bipartition, whereas $d_{\mathrm{GME}}$ asks whether the full state can be reduced to bounded-size blocks by a global Clifford transformation.  The distinction is already visible in our construction: a state such as $\mathrm{CNOT}\ket{T}\ket{T}$ has $d_{\mathrm{GME}}=1$, because the CNOT can be undone by a global Clifford, while it can nevertheless possess nonzero non-local magic across the two-qubit bipartition.  Thus the two quantities probe different aspects of irreducible nonstabilizer structure, and understanding their relation more systematically may provide useful connections between tomography complexity and resource-theoretic characterizations of magic.

A broader open direction is to extend the present approach beyond CEBP states toward more general Clifford-entanglable (CE) states, where block-product inputs may be processed by Clifford channels and combined through classical mixtures.  The present protocol relies strongly on the existence of a single invertible Clifford frame in which the state becomes block-product.  General Clifford channels can discard or dephase Pauli information, while mixtures of different latent Clifford structures need not admit a common frame and can introduce cross terms in the two-copy Bell-sampling statistics.  A natural progression would therefore be to first consider a single unknown Clifford channel acting on a block-product input, then structured mixtures sharing a common Clifford frame or partition, before addressing the full CE class.  The known structural characterizations of Clifford channels~\cite{yashin2025characterization} may provide useful tools for these extensions.

Finally, the present work focuses primarily on establishing \emph{sample efficiency}, and the classical time complexity of the reconstruction procedure has not been optimized.  Several parts of the current algorithm use conservative enumeration and recovery procedures over bounded-size blocks, which are sufficient but are unlikely to be optimal.  Improving the classical postprocessing, particularly its dependence on the block size $d$, is therefore an important direction for future work.  More generally, obtaining tighter upper bounds and complementary lower bounds for both sample and computational complexity would help clarify how directly the latent CEBP block structure captures the intrinsic difficulty of quantum tomography.

\textit{Acknowledgment}
P.L. and Q.Z. acknowledge support from NSF (2240641, 2350153). Q.Z. also acknowledges support from AFOSR MURI FA9550-24-1-0349 and ONR MURI N000142612102. This project is funded in part by the Gordon and Betty Moore Foundation.

\textit{Data and code availability}. The data and scripts used to generate the numerical results are
available on \href{https://github.com/peng-cheng-liao/Sample-Efficient-Tomography-of-Clifford-Encoded-Block-Product-States}{Github}

\bibliography{refs}

\clearpage
\onecolumngrid
\appendix

\makeatletter
\@removefromreset{equation}{section}
\makeatother
\setcounter{section}{0}
\setcounter{equation}{0}
\setcounter{theorem}{0}
\setcounter{algorithm}{0}
\setcounter{figure}{0}
\setcounter{table}{0}
\renewcommand{\thesection}{S\arabic{section}}
\renewcommand{\thesubsection}{\thesection.\arabic{subsection}}
\makeatletter
\renewcommand{\p@subsection}{}
\makeatother
\renewcommand{\theequation}{S\arabic{equation}}
\renewcommand{\thetheorem}{S.\arabic{theorem}}
\renewcommand{\thealgorithm}{S\arabic{algorithm}}
\renewcommand{\thefigure}{S\arabic{figure}}
\renewcommand{\thetable}{S\arabic{table}}
\providecommand{\theHsection}{}
\providecommand{\theHsubsection}{}
\providecommand{\theHequation}{}
\providecommand{\theHtheorem}{}
\providecommand{\theHfigure}{}
\providecommand{\theHtable}{}
\renewcommand{\theHsection}{supp.section.\arabic{section}}
\renewcommand{\theHsubsection}{supp.subsection.\arabic{section}.\arabic{subsection}}
\renewcommand{\theHequation}{supp.equation.\arabic{equation}}
\renewcommand{\theHtheorem}{supp.theorem.\arabic{theorem}}
\renewcommand{\theHalgorithm}{supp.algorithm.\arabic{algorithm}}
\renewcommand{\theHfigure}{supp.figure.\arabic{figure}}
\renewcommand{\theHtable}{supp.table.\arabic{table}}

\section*{Supplemental Material}
\makeatletter
\newcommand{\smtocentrytext}[2]{#1\quad #2}
\newcommand{\l@smtocsection}[2]{%
  \par\addvspace{0.32em}%
  \noindent #1\dotfill #2\par
}
\newcommand{\l@smtocsubsection}[2]{%
  \par\addvspace{0.10em}%
  \noindent\hspace*{1.5em}#1\dotfill #2\par
}
\newcommand{\smtableofcontents}{%
  \begingroup
  \small
  \noindent\textbf{Contents}\par\smallskip
  \@starttoc{smtoc}%
  \endgroup
}
\makeatother
\smtableofcontents

\let\smtocoriginalsection\section
\let\smtocoriginalsubsection\subsection
\RenewDocumentCommand{\section}{s o m}{%
  \IfBooleanTF{#1}{%
    \smtocoriginalsection*{#3}%
  }{%
    \IfNoValueTF{#2}{%
      \smtocoriginalsection{#3}%
      \addcontentsline{smtoc}{smtocsection}{%
        \protect\smtocentrytext{\thesection}{#3}}%
    }{%
      \smtocoriginalsection[#2]{#3}%
      \addcontentsline{smtoc}{smtocsection}{%
        \protect\smtocentrytext{\thesection}{#2}}%
    }%
  }%
}
\RenewDocumentCommand{\subsection}{s o m}{%
  \IfBooleanTF{#1}{%
    \smtocoriginalsubsection*{#3}%
  }{%
    \IfNoValueTF{#2}{%
      \smtocoriginalsubsection{#3}%
      \addcontentsline{smtoc}{smtocsubsection}{%
        \protect\smtocentrytext{\thesubsection}{#3}}%
    }{%
      \smtocoriginalsubsection[#2]{#3}%
      \addcontentsline{smtoc}{smtocsubsection}{%
        \protect\smtocentrytext{\thesubsection}{#2}}%
    }%
  }%
}
\addcontentsline{toc}{section}{Supplemental Material}

This SM contains the proofs and technical details
supporting the results in the main text. It is organized to be readable
independently, while using the same notation as the main text.

\section{State Class and Learning Model}
\label{sec:supp-state-class}

\subsection{Standing model and notation}

Let
\begin{equation}
\Pc_n:=\{I,\sigma^x,\sigma^y,\sigma^z\}^{\otimes n}
\end{equation}
denote the phase-free $n$-qubit Pauli family.  Unless signed operators are
displayed explicitly, Pauli products, subgroups, and binary spans throughout
the SM are taken modulo global phase; canonical Hermitian representatives and
symplectic labels are fixed in Sec.~\ref{sec:supp-pauli-notation}.  Fix a
hidden partition
\begin{equation}
B_1\sqcup \cdots \sqcup B_K=[n],
\qquad
L_a:=|B_a|,
\qquad
\max_{a\in[K]}L_a\le d.
\label{eq:supp-hidden-partition}
\end{equation}

Clifford gates are generated by $\{S,H,\mathrm{CNOT}\}$ and map Pauli
operators to Pauli operators under conjugation.  
A state is block-product with respect to this partition if it has the form
\begin{equation}
\sigma=\bigotimes_{a=1}^K \rho_a,
\qquad
\rho_a\in \D\!\bigl((\mathbb C^2)^{\otimes |B_a|}\bigr).
\label{eq:supp-block-product-state}
\end{equation}
Each block state $\rho_a$ may be mixed and may contain arbitrary
entanglement within $B_a$; the product structure is only across distinct
blocks.
Clifford-encoded block-product (CEBP) states are defined as 
\begin{equation}
\rho
=
U_{\mathrm c}
\left(
\bigotimes_{a=1}^{K}\rho_a
\right)
U_{\mathrm c}^{\dagger},
\qquad
\rho_a\in \D\!\bigl((\mathbb C^2)^{\otimes L_a}\bigr),
\label{eq:supp-clifford-encoded-block-product}
\end{equation}
where $U_{\mathrm c}$ is an $n$-qubit Clifford unitary and the block states
$\rho_a$ may be arbitrary mixed states.  

The unknown target $\rho$ is a CEBP state of the form in
Eq.~\eqref{eq:supp-clifford-encoded-block-product}.  The learner is given
independent copies of $\rho$, the system size $n$, a
maximum-block-size prior $d$, an accuracy
$\varepsilon\in(0,1)$, and a failure probability $\delta\in(0,1)$.  The
number of blocks $K$, the partition $\{B_a\}_{a=1}^K$, the block states
$\{\rho_a\}_{a=1}^K$, and the Clifford unitary $U_{\mathrm c}$ are unknown.
The task is to output a classical description of an estimator $\widehat\rho$
satisfying the trace-norm error guarantee
\begin{equation}
\Pr\!\left\{\|\widehat\rho-\rho\|_1\le\varepsilon\right\}
\ge 1-\delta.
\label{eq:supp-learning-guarantee}
\end{equation}
Throughout, hidden blocks and labels, hidden lifts and aggregate true-block
registers, and comparison states are proof-only.  Returned Paulis, sectors,
empirical clusters, Cliffords, and registers are learner-computable from the
inputs and measurement transcript; in particular, the learner never
reconstructs $U_{\mathrm c}$ or the hidden partition.  Every probabilistic
stage uses a fresh independent copy pool and is analyzed conditionally on the
preceding visible transcript.  Here and below, $M$ Bell rounds consume $2M$
physical copies, while simultaneous tomography on disjoint registers uses the
maximum local schedule rather than their sum.  Empty products and spans have
their usual identities; branches with $t=0$, $t=n$, $L\le1$, or
$\widehat K=0$ are interpreted by the explicit vacuous conventions stated
where they arise.

\subsection{Extensive entanglement and magic in CEBP states}
\label{sec:supp-cebp-entanglement-magic}

For even $n$, take single-qubit hidden blocks with latent product state
$\ket{T}^{\otimes n}$, where
\begin{equation}
\ket{T}:=T\ket{+}=(\ket 0+e^{i\pi/4}\ket 1)/\sqrt2 .
\end{equation}
Fix the physical ordering $1,\ldots,n$ and the midpoint cut
$[n/2]|[n]\setminus[n/2]$.  Let $U_{\mathrm c}$ be the product of CNOT gates
from qubit $i$ to qubit $n/2+i$ for $i=1,\ldots,n/2$.  Then
$U_{\mathrm c}\ket{T}^{\otimes n}$ is a tensor product of $n/2$ identical
two-qubit states across the midpoint cut.  One such state is
\begin{equation}
\ket{\eta}
=\mathrm{CNOT}\ket{T}\ket{T}
=\frac{1}{2}\left(
\ket{00}+e^{i\pi/4}\ket{01}+i\ket{10}
+e^{i\pi/4}\ket{11}
\right).
\end{equation}
Its two Schmidt probabilities are
\begin{equation}
\lambda_{\pm}=\frac{1}{2}\left(1\pm\frac{1}{\sqrt2}\right),
\end{equation}
so its entanglement entropy
$s_T:=-\sum_{\nu\in\{+,-\}}\lambda_\nu\log\lambda_\nu \approx0.42$ is strictly
positive.  Additivity across the $n/2$ pairs therefore gives entanglement
entropy $(n/2)s_T=\Theta(n)$ across the midpoint cut.  Both Schmidt
probabilities are nonzero, so each crossing pair has Schmidt rank two and
the full state has Schmidt rank $2^{n/2}$ across this cut.  Consequently,
any exact MPS representation in this fixed physical ordering requires bond
dimension at least $2^{n/2}$.

For a pure $m$-qubit state $\ket{\phi}$, the order-two stabilizer R\'enyi
entropy is~\cite{LeoneOlivieroHamma2022SRE}
\begin{equation}
M_2(\ket{\phi})
:=-\log\!\left(
2^{-m}\sum_{P\in\Pc_m}|\bra{\phi}P\ket{\phi}|^4
\right).
\label{eq:supp-stabilizer-renyi-two}
\end{equation}
For $\ket{T}$, the Pauli expectation values are
\begin{equation}
\bra{T}I\ket{T}=1,
\qquad
\bra{T}X\ket{T}=\bra{T}Y\ket{T}=\frac{1}{\sqrt2},
\qquad
\bra{T}Z\ket{T}=0.
\end{equation}
Consequently,
\begin{equation}
M_2(\ket{T})
=-\log\!\left[\frac{1}{2}\left(1+\frac14+\frac14\right)\right]
=\log(4/3).
\end{equation}
Additivity and Clifford invariance now give
\begin{equation}
M_2(\ket{T}^{\otimes n})
=M_2(U_{\mathrm c}\ket{T}^{\otimes n})
=n\log(4/3).
\end{equation}
Thus, the CEBP class contains states with both linear bipartite entanglement
and extensive magic even when the hidden block size is one.

\section{Pauli Preliminaries and Bell Sampling}
\label{sec:supp-notation}

We next establish the Pauli conventions, hidden block structure, and common
Bell-sampling estimator used throughout the learning protocol.

\subsection{Phase-free Pauli and symplectic notation}
\label{sec:supp-pauli-notation}

Recall the phase-free Pauli family $\Pc_n$ and group/span convention from
Sec.~\ref{sec:supp-state-class}.  Whenever an operator rather than a
phase-free class is required, we use the canonical Hermitian representative:
for $v=(x,z)\in\Fc^{2m}$, set
\begin{equation}
P(v)
:=
i^{\sum_{j=1}^m x_jz_j}
(\sigma_1^x)^{x_1}\cdots(\sigma_m^x)^{x_m}
(\sigma_1^z)^{z_1}\cdots(\sigma_m^z)^{z_m}.
\end{equation}
The sum in the exponent is evaluated over the integers, rather than modulo
two, so this expression agrees with the tensor product of its local
$I,\sigma^x,\sigma^y,\sigma^z$ factors.  This convention fixes one Hermitian
representative of each phase-free Pauli class; a Clifford conjugation may
still return its negative, which never affects the scores defined below.

Thus every phase-free Pauli on $m$ qubits has a unique binary label
$v=(x,z)\in\Fc^{2m}$, with $(x_i,z_i)=(0,0),(1,0),(1,1),(0,1)$ corresponding,
respectively, to the local factors $I,\sigma^x,\sigma^y,\sigma^z$.  We use
binary labels and Pauli operators interchangeably when the number and
ordering of the qubits are clear.

The symplectic inner product is
\begin{equation}
[u,v]
:=
u^\mathsf{T}
\begin{pmatrix}
0 & I_m \\
I_m & 0
\end{pmatrix}
v
\pmod 2.
\end{equation}
It is bilinear and alternating over $\Fc$, and it is the symplectic form used
throughout the recovery arguments.

In particular, the Hermitian representatives obey
\begin{equation}
P(u)P(v)=(-1)^{[u,v]}P(v)P(u).
\end{equation}
Hence two phase-free Paulis commute if and only if $[u,v]=0$, and they
anticommute if and only if $[u,v]=1$
\cite{Gottesman1997,DehaeneDeMoor2003,AaronsonGottesman2004}.

For a Pauli string $P=\bigotimes_{i=1}^m P_i$, its labelled support is
\begin{equation}
\supp(P)
:=
\bigl\{(i,\alpha):P_i=\sigma_i^\alpha,
\alpha\in\{x,y,z\}\bigr\},
\end{equation}
where $\sigma_i^\alpha$ denotes the Pauli $\sigma^\alpha$ acting on qubit
$i$.  Retaining the label $\alpha$, rather than only the occupied qubit $i$,
makes the Bell parities below unambiguous.

Finally, for any state $\omega$ and any $P\in\Pc_n$, we define its squared
Pauli expectation score by
\begin{equation}
s_\omega(P):=|\Tr(\omega P)|^2.
\label{eq:supp-score-def}
\end{equation}
Because the representative $P$ is Hermitian, $0\le s_\omega(P)\le1$ and the
score is independent of the representative's sign.  When the state is clear
from context, we simply write $s(P)$.

\subsection{Transformed block-Pauli structure and score factorization}
\label{sec:supp-block-factorization}

With the phase-free Pauli conventions fixed, we now express the hidden
block-product structure in physical Pauli coordinates.  For each hidden
block $B_a$, $a\in[K]$, let $L_a:=|B_a|$ denote its block size.  We enumerate
the nonidentity Pauli operators supported on $B_a$ as
\begin{equation}
\Pc_{B_a}\setminus\{I_{B_a}\}
=
\{\Lambda_a^1,\dots,\Lambda_a^{4^{L_a}-1}\},
\end{equation}
where $I_{B_a}$ denotes the identity operator on the qubits in $B_a$.
For every $a\in[K]$ and every
$\ell\in[4^{L_a}-1]$, define the transformed single-block Pauli
operator
\begin{equation}
Q_a^\ell
:=
U_{\mathrm c}
\bigl(
\Lambda_a^\ell \otimes I_{B_a^c}
\bigr)
U_{\mathrm c}^\dagger ,
\label{eq:supp-transformed-block-pauli-indexed}
\end{equation}
where $\Lambda_a^\ell\otimes I_{B_a^c}$ denotes the
$n$-qubit Pauli operator that acts as $\Lambda_a^\ell$ on block
$B_a$ and as identity on all other blocks.  The collection
\begin{equation}
\mathcal Q
:=
\Bigl\{
Q_a^\ell:\ a\in[K],\ \ell\in[4^{L_a}-1]
\Bigr\}
\end{equation}
is the transformed block-Pauli family associated with the hidden
block-product structure.

\begin{lemma}[Hidden block-Pauli decomposition]
\label{lem:supp-hidden-block-pauli-decomposition}
Every $P\in\Pc_n$ has a unique transformed block support
\begin{equation}
\supp_Q(P)
\subseteq
\bigcup_{a=1}^K\bigl(\{a\}\times[4^{L_a}-1]\bigr),
\end{equation}
containing at most one pair with each block index, such that
\begin{equation}
P
=
\chi(P)\prod_{(a,\ell)\in\supp_Q(P)}Q_a^\ell,
\qquad
\chi(P)\in\{\pm1\}.
\label{eq:supp-block-support-decomp}
\end{equation}
\end{lemma}

\begin{proof}
Conjugating $P$ into the hidden coordinates gives
\begin{equation}
\widetilde P:=U_{\mathrm c}^\dagger P U_{\mathrm c},
\end{equation}
Since $U_{\mathrm c}$ is Clifford, $\widetilde P$ is a Hermitian Pauli
operator and therefore equals either sign of the canonical representative of
its phase-free class.

The hidden partition gives the unique phase-free tensor decomposition
\begin{equation}
\widetilde P
=
\chi(P)\bigotimes_{a=1}^K\Lambda_a(P),
\qquad
\Lambda_a(P)\in\Pc_{B_a},
\end{equation}
because a Pauli string has a unique local Pauli factor on every qubit and
hence on every block.

Define $\supp_Q(P)$ by including $(a,\ell)$ precisely when
$\Lambda_a(P)=\Lambda_a^\ell\ne I_{B_a}$.  Each nonidentity block factor has
a unique index $\ell$, so this support contains at most one labelled factor
from each block and is unique.  Because the selected factors originate from
distinct hidden blocks, they commute, so the product is independent of their
ordering.

Finally, conjugating the block decomposition by $U_{\mathrm c}$ and using
Eq.~\eqref{eq:supp-transformed-block-pauli-indexed} gives
Eq.~\eqref{eq:supp-block-support-decomp}, with $\chi(P)$ accounting for the
sign relative to the canonical Hermitian representative.
\end{proof}

We call $P$ \emph{composite in the transformed block-Pauli family} when
$|\supp_Q(P)|\ge2$; otherwise a nonidentity $P$ with one support element is a
transformed single-block Pauli.

\begin{proposition}[Transformed score factorization]
\label{prop:supp-transformed-score-factorization}
For every $P\in\Pc_n$ with the decomposition in
Lemma~\ref{lem:supp-hidden-block-pauli-decomposition},
\begin{equation}
\Tr(\rho P)
=
\chi(P)
\prod_{(a,\ell)\in\supp_Q(P)}
\Tr(\rho_a \Lambda_a^\ell).
\label{eq:supp-block-factorization-exp}
\end{equation}
Consequently,
\begin{equation}
s_\rho(P)
=
\prod_{(a,\ell)\in\supp_Q(P)}
\bigl|\Tr(\rho_a \Lambda_a^\ell)\bigr|^2 .
\label{eq:supp-block-factorization-score}
\end{equation}
\end{proposition}

\begin{proof}
Using the CEBP form in
Eq.~\eqref{eq:supp-clifford-encoded-block-product} and cyclicity of the
trace, the expectation of $P$ equals the expectation of
$U_{\mathrm c}^\dagger P U_{\mathrm c}$ in
$\bigotimes_a\rho_a$.  The unique block decomposition from the preceding
lemma therefore reduces the observable to distinct hidden-block factors.

The trace of a tensor product factorizes across those blocks, while identity
factors contribute $\Tr(\rho_a)=1$.  This gives the product in
Eq.~\eqref{eq:supp-block-factorization-exp}.

The factor $\chi(P)$ records the possible sign between the canonical
Hermitian representative of $P$ and the product of the transported
representatives.  It is the only phase information that can enter the real
expectation value.

Taking absolute squares removes $\chi(P)$ and yields the score product in
Eq.~\eqref{eq:supp-block-factorization-score}.
\end{proof}

In particular, a transformed single-block Pauli has score
\begin{equation}
s_\rho(Q_a^\ell)
=
\bigl|\Tr(\rho_a\Lambda_a^\ell)\bigr|^2.
\label{eq:supp-single-block-score}
\end{equation}

This factorization is only across distinct hidden blocks: a multiqubit
$\Lambda_a^\ell$ need not factorize within the possibly entangled state
$\rho_a$.  The transported Paulis also inherit the hidden symplectic
commutation relations.  If $v_a^\ell$ labels $\Lambda_a^\ell$, then operators
from distinct blocks commute and
\begin{equation}
Q_a^\ell Q_a^{\ell'}
=
(-1)^{[v_a^\ell,v_a^{\ell'}]}
Q_a^{\ell'}Q_a^\ell,
\end{equation}
because Clifford conjugation preserves commutators.  In particular, for a
one-qubit hidden block the transported $X,Y,Z$ Pauli directions form its Pauli frame:
distinct Pauli directions anticommute, frames on different qubits commute, and
Eq.~\eqref{eq:supp-block-factorization-score} factors every score over the
participating frame elements.

Since every factor in Eq.~\eqref{eq:supp-block-factorization-score} lies in
$[0,1]$, each block component dominates the product:
\begin{equation}
(a,\ell)\in\supp_Q(P)
\quad\Longrightarrow\quad
s_\rho(P)\le s_\rho(Q_a^\ell).
\end{equation}
For $s_\rho(P)>0$, equality holds exactly when every other participating
factor has unit score.  Unit-score factors also explain stabilizer dressing:
if $S\odot P$ is the phase-free product and $s_\rho(S)=1$, then
\begin{equation}
s_\rho(S\odot P)=s_\rho(P).
\end{equation}
Indeed, the canonical Hermitian representative $S_{\mathrm H}$ acts as a
fixed sign on the support of $\rho$, while
$S_{\mathrm H}P_{\mathrm H}$ differs from $(S\odot P)_{\mathrm H}$ only by a
Pauli phase; squared absolute expectations remove both factors.

Thus product scores give only a within-product comparison, not a global block
ordering: unrelated components may appear in any order, composites may tie a
component through unit-score factors, and stabilizer dressing may create an
entire degenerate family.  Certified peeling removes the last ambiguity,
symplectic compatibility supports residual recovery, and mixed cumulants are
still needed to identify the hidden-block grouping.

\subsection{Two-copy Bell-sampling score estimator}
\label{sec:supp-bell-estimator}

All fixed or adaptive structural score queries use one common record of
two-copy Bell measurements.

In one round, take two copies $\rho\otimes\rho$ and perform a Bell-basis
measurement on every pair of corresponding qubits.  Thus $M$ Bell rounds
always means $M$ independent copy pairs and $2M$ physical copies; all later
resource ledgers use this convention.  Equivalently, for each $i\in[n]$
jointly measure the commuting observables
\begin{equation}
\sigma_i^{x,(1)}\sigma_i^{x,(2)},
\qquad
\sigma_i^{y,(1)}\sigma_i^{y,(2)},
\qquad
\sigma_i^{z,(1)}\sigma_i^{z,(2)}.
\label{eq:supp-bell-observables}
\end{equation}

Write the recorded eigenvalues as
\begin{equation}
B_{i,j}^\alpha\in\{-1,+1\},
\qquad
i\in[n],\quad \alpha\in\{x,y,z\},\quad j\in[M].
\end{equation}
All three constrained correlator eigenvalues are recorded so a Pauli support
selects its factors directly.

For any $P\in\Pc_n$, define its round-$j$ parity and empirical score by
\begin{equation}
b_j(P)
:=
\prod_{(i,\alpha)\in\supp(P)}B_{i,j}^\alpha,
\qquad
\hat s(P)
:=
\frac1M\sum_{j=1}^M b_j(P).
\label{eq:supp-score-estimator}
\end{equation}
The estimator may lie in $[-1,1]$ although $s_\rho(P)\in[0,1]$; also,
the empty product gives $b_j(I)=1$.

The estimator is the mean of the parity function $b_j(P)$, not the
frequency of one Bell-outcome string.  Many distinct Bell outcomes can
contribute the same parity, and it is this signed moment of the outcome
distribution that equals the squared Pauli expectation.

Related Bell and Pauli methods recover stabilizers, predict observables, or
learn structured nonstabilizer families
\cite{Montanaro2017,KingGossetKothariBabbush2025,GrossNezamiWalter2021,GrewalIyerKretschmerLiang2024,HuangKuengPreskill2021,HangleiterGullans2024,lai2022learning,leone2024};
here Bell sampling is only a uniform score-estimation primitive for arbitrary
mixed states.

\begin{lemma}[Bell-estimator unbiasedness]
\label{lem:supp-bell-product-expectation}
For every $P\in\Pc_n$, the estimator in
Eq.~\eqref{eq:supp-score-estimator} is unbiased:
\begin{equation}
\mathbb E[\hat s(P)]=s_\rho(P).
\end{equation}
\end{lemma}

\begin{proof}
The measured correlators commute, so the parity is the outcome of their
product.  Grouping its factors by copy gives
\begin{equation}
\mathbb E[b_j(P)]
=
\Tr\!\left[
(\rho\otimes\rho)
\prod_{(i,\alpha)\in\supp(P)}
\sigma_i^{\alpha,(1)}\sigma_i^{\alpha,(2)}
\right].
\label{eq:supp-bell-joint-meas-identity}
\end{equation}
The product observable is
\begin{equation}
P^{(1)}P^{(2)},
\label{eq:supp-bell-product-observable}
\end{equation}
and therefore
\begin{equation}
\mathbb E[b_j(P)]
=
\Tr(\rho P)^2
=
|\Tr(\rho P)|^2
=
s_\rho(P),
\end{equation}
where Hermiticity makes $\Tr(\rho P)$ real.  Averaging proves the lemma.
\end{proof}

\begin{proposition}[Uniform Bell-sampling score concentration]
\label{prop:supp-uniform-concentration}
For $\zeta_{\mathrm{BS}}\in(0,1)$, the estimator in
Eq.~\eqref{eq:supp-score-estimator} satisfies
\begin{equation}
|\hat s(P)-s_\rho(P)|\le\tau(M,n,\zeta_{\mathrm{BS}}),
\qquad
\forall\,P\in\Pc_n,
\label{eq:supp-uniform-event}
\end{equation}
with probability at least $1-\zeta_{\mathrm{BS}}$, where
\begin{equation}
\tau(M,n,\zeta_{\mathrm{BS}})
:=
\sqrt{
\frac{2\log(2\cdot4^n/\zeta_{\mathrm{BS}})}{M}
}.
\label{eq:supp-tauM}
\end{equation}
When the parameters are clear from context, we write this quantity as
$\tau$.
\end{proposition}

\begin{proof}
For fixed $P$, Lemma~\ref{lem:supp-bell-product-expectation} and Hoeffding's
inequality~\cite{Hoeffding1963} give
$\Pr[|\hat s(P)-s_\rho(P)|\ge t]\le2e^{-Mt^2/2}$.  Union bounding over the
$4^n$ phase-free Paulis gives
\begin{equation}
\Pr\!\left[
\exists\,P\in\Pc_n:
|\hat s(P)-s_\rho(P)|\ge t
\right]
\le2\cdot4^n e^{-Mt^2/2}.
\end{equation}
Choosing $t=\tau(M,n,\zeta_{\mathrm{BS}})$ from
Eq.~\eqref{eq:supp-tauM} makes this $\zeta_{\mathrm{BS}}$.
\end{proof}

To obtain simultaneous additive score error at most $\eta>0$, it suffices
to choose
\begin{equation}
M
\ge
\frac{2}{\eta^2}
\log\!\left(\frac{2\cdot4^n}{\zeta_{\mathrm{BS}}}\right).
\end{equation}
By Eq.~\eqref{eq:supp-tauM}, the common record costs $O(Mn)$ storage and
supports every fixed or adaptively selected score query on the same event
\eqref{eq:supp-uniform-event}; exhaustive candidate search, when used below,
is a separate classical cost.  Thus the uniform score error uses
$M=O((n+\log(1/\zeta_{\mathrm{BS}}))/\eta^2)$ Bell rounds and, by the
Bell-round convention, $2M$ physical copies.

\section{Certified Stabilizer Peeling}
\label{sec:supp-stabilizer-peeling}

This section turns high Pauli scores into a certified reduction of the
learning problem.  Multiplication by near-stabilizer directions creates
score-degenerate or nearly degenerate physical representatives, which
motivates peeling those directions before recovering the residual frame.  We
first derive a deterministic approximate syndrome decomposition from
independent commuting high-score directions.  We then use noisy Bell-sampling
scores to certify a commuting high-score span, derive certified generator
deficits and a data-dependent peeling bound, and obtain a syndrome-product
fidelity guarantee with a dominant syndrome.

\subsection{Deterministic approximate stabilizer decomposition}
\label{sec:supp-deterministic-approximate-stabilizer-decomposition}

Independent commuting Pauli directions with scores close to one identify a
sign-consistent joint syndrome sector whose leakage is controlled by the
aggregate score deficits.  Diagonalizing those directions by a Clifford
therefore isolates the dominant syndrome sector and its associated normalized
conditioned residual state.

Throughout this section, \emph{fidelity} means root fidelity,
\begin{equation}
\F(\omega,\sigma)
:=
\left\|\sqrt{\omega}\sqrt{\sigma}\right\|_1,
\label{eq:supp-root-fidelity-convention}
\end{equation}
where
$\|A\|_1:=\Tr\sqrt{A^\dagger A}$ denotes the trace norm.  
Thus the Fuchs--van de Graaf inequalities~\cite{FuchsVanDeGraaf1999} take the form
\begin{equation}
1-\F(\omega,\sigma)
\le
\frac12\|\omega-\sigma\|_1
\le
\sqrt{1-\F(\omega,\sigma)^2}.
\label{eq:supp-root-fidelity-fvdg}
\end{equation}

For a rank-$t$ commuting family, define the syndrome and residual registers
by
\begin{equation}
\mathcal H_{\mathrm{syn}}
:=
(\mathbb C^2)^{\otimes t},
\qquad
\mathcal H_{\mathrm{res}}
:=
(\mathbb C^2)^{\otimes(n-t)}.
\label{eq:supp-syndrome-residual-registers}
\end{equation}
Thus the first $t$ diagonalized qubits form the syndrome register and the
remaining $n-t$ qubits form the residual register.
The integer $t$ is the peeling rank.
For $x\in\{0,1\}^t$, let
\begin{equation}
\Pi_x
:=
|x\rangle\!\langle x|\otimes I_{\mathrm{res}}
\label{eq:supp-joint-syndrome-projector}
\end{equation}
denote the corresponding joint syndrome-sector projector, where
$I_{\mathrm{res}}$ is the identity on $\mathcal H_{\mathrm{res}}$.  The case $t=n$
is included, with the residual register interpreted as the one-dimensional
zero-qubit system.  We use the convention $\Pc_0:=\{I\}$. 
Accordingly, every conclusion quantified over nonidentity residual Paulis is
vacuous when $t=n$.

\begin{proposition}[Deterministic approximate stabilizer decomposition]
\label{prop:supp-approximate-stabilizer-decomposition}
Let $g_1,\dots,g_t\in \Pc_n$ be independent commuting Pauli operators
and let $0\le \varepsilon_j^{\mathrm{stab}}\le 1$ satisfy
\begin{equation}
s_\rho(g_j)\ge 1-\varepsilon_j^{\mathrm{stab}},
\qquad
j=1,\dots,t.
\end{equation}
Then there exists an $n$-qubit Clifford unitary $U_{\mathrm{stab}}$ satisfying
\begin{equation}
U_{\mathrm{stab}}^\dagger g_jU_{\mathrm{stab}}=Z_j,
\qquad
j=1,\dots,t.
\label{eq:supp-stabilizer-diagonalization}
\end{equation}
For every such Clifford, there exist a bit string $b\in\{0,1\}^t$ and a
residual state $\rho_{\mathrm{res}}$ on the remaining $n-t$ qubits such that,
for
\begin{equation}
\widetilde\rho
:=
U_{\mathrm{stab}}^\dagger \rho\, U_{\mathrm{stab}},
\end{equation}
\begin{equation}
\F\!\left(
\widetilde\rho,\,
|b\rangle\!\langle b|\otimes \rho_{\mathrm{res}}
\right)
\ge
\sqrt{
\max\left\{0,1-\frac12\sum_{j=1}^t \varepsilon_j^{\mathrm{stab}}\right\}
}.
\label{eq:supp-stabilizer-fidelity}
\end{equation}
Here $b_j=0$ when $\Tr(\rho g_j)\ge0$ and $b_j=1$ otherwise; in particular,
a zero expectation is assigned the positive-sign bit $b_j=0$ by convention.
\end{proposition}

\begin{proof}
Because the $g_j$ are independent and pairwise commuting, symplectic
completion gives a Clifford that maps their phase-free labels to the first
$t$ computational $Z$ Pauli directions.  Its conjugated Hermitian representatives may
initially be $\pm Z_j$; composing with the appropriate single-qubit Pauli
corrections fixes these signs and gives
Eq.~\eqref{eq:supp-stabilizer-diagonalization}.  Fix such a Clifford, set
$\mu_j:=\Tr(\rho g_j)=\Tr(\widetilde\rho Z_j)$, and define
\begin{equation}
b_j
:=
\begin{cases}
0,&\mu_j\ge0,\\
1,&\mu_j<0
\end{cases}.
\label{eq:supp-deterministic-syndrome-bit}
\end{equation}
Then $(-1)^{b_j}\mu_j=|\mu_j|$, and the squared-score hypothesis gives
\begin{equation}
|\mu_j|
=
\sqrt{s_\rho(g_j)}
\ge
\sqrt{1-\varepsilon_j^{\mathrm{stab}}}
\ge
1-\varepsilon_j^{\mathrm{stab}}.
\label{eq:supp-generator-expectation-magnitude}
\end{equation}
The wrong-syndrome-bit probability is therefore
\begin{equation}
p_j^{\mathrm{err}}
=
\frac{1-|\mu_j|}{2}
\le
\frac{\varepsilon_j^{\mathrm{stab}}}{2}.
\label{eq:supp-incorrect-syndrome-bit}
\end{equation}
Writing $p_b:=\Tr(\widetilde\rho\Pi_b)$, a union bound gives
\begin{equation}
p_b
\ge
1-\sum_{j=1}^t p_j^{\mathrm{err}}
\ge
1-\frac12\sum_{j=1}^t\varepsilon_j^{\mathrm{stab}}.
\label{eq:supp-dominant-syndrome-weight}
\end{equation}
If $p_b>0$, define the normalized conditioned residual state associated with
the dominant syndrome sector by
\begin{equation}
\rho_{\mathrm{res}}
:=
\frac{\langle b|\widetilde\rho|b\rangle}{p_b}.
\label{eq:supp-conditioned-residual-state}
\end{equation}
This is a normalized block, not an instruction to postselect.  For
$X_b:=\sqrt{\widetilde\rho}\Pi_b\sqrt{\widetilde\rho}$ and
$\sigma_b:=|b\rangle\!\langle b|\otimes\rho_{\mathrm{res}}$, the identity
$\sqrt{\widetilde\rho}\sigma_b\sqrt{\widetilde\rho}=X_b^2/p_b$ yields
\begin{equation}
\F(\widetilde\rho,\sigma_b)
=
\Tr\sqrt{\frac{X_b^2}{p_b}}
=
\frac{\Tr X_b}{\sqrt{p_b}}
=
\sqrt{p_b}.
\label{eq:supp-conditioned-sector-fidelity}
\end{equation}
If $p_b=0$, Eq.~\eqref{eq:supp-dominant-syndrome-weight} forces the
right-hand side of Eq.~\eqref{eq:supp-stabilizer-fidelity} to be zero, so
any residual state suffices.  Otherwise
Eq.~\eqref{eq:supp-dominant-syndrome-weight} and the displayed identity
prove Eq.~\eqref{eq:supp-stabilizer-fidelity}.  For $t=0$, $b$ is empty,
$p_b=1$, and $\rho_{\mathrm{res}}=\widetilde\rho$.
\end{proof}

\subsection{Empirical span certification}
\label{sec:supp-empirical-span-certification}

Scores near a fixed threshold need not determine a unique empirical span, so
we certify equality of inner and outer empirical spans.  Fix
\begin{equation}
\frac12<h_{\min}<h_{\max}<1,
\qquad
\mathsf I_h:=[h_{\min},h_{\max}],
\qquad
\Delta_h:=h_{\max}-h_{\min}>0.
\label{eq:supp-peeling-threshold-window}
\end{equation}
For a target mesh width $\eta>0$, set
\begin{equation}
N_{\mathrm{grid}}:=\left\lceil\frac{\Delta_h}{\eta}\right\rceil,
\qquad
\eta_{\mathrm{grid}}:=\frac{\Delta_h}{N_{\mathrm{grid}}}\le\eta,
\qquad
\mathsf H_{\mathrm{peel}}
:=
\left\{
h_{\min}+j\eta_{\mathrm{grid}}
:
j=0,\ldots,N_{\mathrm{grid}}
\right\}.
\label{eq:supp-peeling-threshold-grid}
\end{equation}
Thus every point of $\mathsf I_h$ is within
$\eta_{\mathrm{grid}}/2$ of the grid.  Fix deterministic grid and Pauli
orders, the latter for basis tie-breaking.  For $M_1$ Bell rounds, set
\begin{equation}
\tau_1
:=
\tau(M_1,n,\zeta_{\mathrm{BS}}).
\label{eq:supp-peeling-score-radius}
\end{equation}
The corresponding uniform event is
\begin{equation}
\mathcal E_{\mathrm{BS}}
:=
\left\{
|\hat s(P)-s_\rho(P)|\le\tau_1
\text{ for every }P\in\Pc_n
\right\}.
\label{eq:supp-peeling-uniform-event}
\end{equation}
Proposition~\ref{prop:supp-uniform-concentration} gives
$\Pr(\mathcal E_{\mathrm{BS}})\ge1-\zeta_{\mathrm{BS}}$ using $2M_1$
copies.

For every real threshold $x$, define the true high-score set and its binary
span by
\begin{equation}
\mathcal T(x)
:=
\{P\in\Pc_n:s_\rho(P)\ge x\},
\qquad
V(x)
:=
\Span_{\Fc}\mathcal T(x).
\label{eq:supp-true-high-score-span}
\end{equation}
For every candidate $h\in\mathsf H_{\mathrm{peel}}$, define the empirical inner and outer
high-score sets by
\begin{equation}
\begin{aligned}
\widehat{\mathcal T}_{\mathrm{in}}(h)
&:=
\bigl\{
P\in\Pc_n:\hat s(P)\ge h+\tau_1
\bigr\},
\\
\widehat{\mathcal T}_{\mathrm{out}}(h)
&:=
\bigl\{
P\in\Pc_n:\hat s(P)\ge h-\tau_1
\bigr\}.
\end{aligned}
\label{eq:supp-empirical-inner-outer-sets}
\end{equation}
\begin{lemma}[Empirical span sandwich]
\label{lem:supp-empirical-span-sandwich}
On $\mathcal E_{\mathrm{BS}}$, every $h\in\mathsf H_{\mathrm{peel}}$ satisfies the set
inclusions
\begin{equation}
\begin{aligned}
\mathcal T(h+2\tau_1)
\subseteq
\widehat{\mathcal T}_{\mathrm{in}}(h)
\subseteq
\mathcal T(h)
\subseteq
\widehat{\mathcal T}_{\mathrm{out}}(h)
\subseteq
\mathcal T(h-2\tau_1),
\end{aligned}
\label{eq:supp-empirical-high-score-set-sandwich}
\end{equation}
and hence the span inclusions
\begin{equation}
\begin{aligned}
V(h+2\tau_1)
\subseteq
\Span_{\Fc}\widehat{\mathcal T}_{\mathrm{in}}(h)
\subseteq
V(h)
\subseteq
\Span_{\Fc}\widehat{\mathcal T}_{\mathrm{out}}(h)
\subseteq
V(h-2\tau_1).
\end{aligned}
\label{eq:supp-empirical-high-score-span-sandwich}
\end{equation}
\end{lemma}

\begin{proof}
On $\mathcal E_{\mathrm{BS}}$, subtracting or adding $\tau_1$ to the
defining score inequalities gives each adjacent inclusion in
Eq.~\eqref{eq:supp-empirical-high-score-set-sandwich}.  Taking binary spans
gives Eq.~\eqref{eq:supp-empirical-high-score-span-sandwich}.
\end{proof}

The algorithm accepts only a matching inner/outer span.  Its direct
all-Pauli implementation is information-theoretic and may take exponential
classical time.

\begin{inlinealgorithm}{alg:supp-stabilizer-peeling}{Certified empirical stabilizer peeling}
\begin{algorithmic}[1]
\Require Empirical scores $\hat s(P)$ for $P\in\Pc_n$; a uniform error
radius $\tau_1$; numbers
$1/2<h_{\min}<h_{\max}<1$ defining
$\mathsf I_h=[h_{\min},h_{\max}]$; a finite grid
$\mathsf H_{\mathrm{peel}}\subset\mathsf I_h$ with a fixed
deterministic ordering; a fixed Pauli ordering for basis tie-breaking.
\Ensure Either failure or an accepted threshold $h$, peeling parameter
$\lambda:=1-h\in(0,1/2)$, an independent rank-$t$ basis
$g_1,\dots,g_t$, and a Clifford unitary $U_{\mathrm{stab}}$.
\For{each $h\in\mathsf H_{\mathrm{peel}}$ in the fixed order}
    \State Form $\widehat{\mathcal T}_{\mathrm{in}}(h)$ and
    $\widehat{\mathcal T}_{\mathrm{out}}(h)$ from
    Eq.~\eqref{eq:supp-empirical-inner-outer-sets}.
    \If{
    $
    \Span_{\Fc}\widehat{\mathcal T}_{\mathrm{in}}(h)
    =
    \Span_{\Fc}\widehat{\mathcal T}_{\mathrm{out}}(h)
    $
    }
        \State Choose
        $\{g_1,\dots,g_t\}\subseteq
        \widehat{\mathcal T}_{\mathrm{in}}(h)$ as the first independent basis
        of the common binary span in the fixed Pauli order.
        \If{some $g_j$ and $g_k$ anticommute}
            \State \textbf{continue} to the next threshold.
        \EndIf
        \If{$t=0$}
            \State Set $U_{\mathrm{stab}}:=I$.
        \Else
            \State Construct $U_{\mathrm{stab}}$ such that
            $U_{\mathrm{stab}}^\dagger g_jU_{\mathrm{stab}}=Z_j$ for
            $j=1,\dots,t$.
        \EndIf
        \State \Return
        $(h,\lambda:=1-h,\{g_1,\dots,g_t\},U_{\mathrm{stab}})$.
    \EndIf
\EndFor
\State \Return failure.
\end{algorithmic}
\end{inlinealgorithm}

It remains to show that a matching candidate exists and is commuting.  For
$c\in(h_{\min},h_{\max})$, define
\begin{equation}
V(c^+)
:=
\Span_{\Fc}\{P\in\Pc_n:s_\rho(P)>c\},
\end{equation}
and define the internal span-change set by
\begin{equation}
\mathsf C
:=
\left\{
c\in(h_{\min},h_{\max})
:
V(c^+)\subsetneq V(c)
\right\},
\label{eq:supp-span-change-point}
\end{equation}
so $c\in\mathsf C$ precisely when the Paulis of score exactly $c$ enlarge
the binary span.  Augment these internal changes by the window endpoints:
\begin{equation}
\mathsf B
:=
\mathsf C\cup\{h_{\min},h_{\max}\}.
\label{eq:supp-span-change-boundaries}
\end{equation}
The components of $\mathsf I_h\setminus\mathsf B$ are
\emph{span-stable}: $V(x)$ is constant on each one.

For $x_1>x_2$, monotonicity of the score cutoff gives
\begin{equation}
\mathcal T(x_1)\subseteq\mathcal T(x_2),
\qquad
V(x_1)\subseteq V(x_2).
\label{eq:supp-high-score-span-monotonicity}
\end{equation}
Every strict change increases $\dim_{\Fc}V(x)$ as the threshold decreases.
Since the ambient binary dimension is $2n$, $|\mathsf C|\le2n$ and there
are at most $2n+1$ span-stable components.

\begin{lemma}[Stable grid threshold]
\label{lem:supp-stable-grid-threshold}
Let $\mathsf I_h$, $\mathsf H_{\mathrm{peel}}$, $\eta$,
$\eta_{\mathrm{grid}}$, and $\tau_1$ be as above.
Assume the target mesh satisfies the sufficient design condition
\begin{equation}
\eta
\le
\frac{\Delta_h}{4(2n+1)},
\label{eq:supp-certified-peeling-grid-condition}
\end{equation}
so the actual mesh obeys
$\eta_{\mathrm{grid}}\le\eta\le \Delta_h/[4(2n+1)]$.  Assume also that the uniform
score radius satisfies
\begin{equation}
\tau_1
\le
\frac{\Delta_h}{8(2n+1)}.
\label{eq:supp-certified-peeling-tau-condition}
\end{equation}
Then there exists $h\in\mathsf H_{\mathrm{peel}}$ for which
$h-2\tau_1$ and $h+2\tau_1$ lie in the same span-stable component and,
in particular,
\begin{equation}
V(h+2\tau_1)
=
V(h-2\tau_1).
\label{eq:supp-certified-peeling-true-span-gap}
\end{equation}
\end{lemma}

\begin{proof}
Choose a widest span-stable component $\mathcal I_\star$, of length $w$
and midpoint $h_\star$.  The component count and total length give
\begin{equation}
w
\ge
\frac{\Delta_h}{2n+1}.
\label{eq:supp-required-span-stable-width}
\end{equation}
A nearest grid point $h$ has
$|h-h_\star|\le\eta_{\mathrm{grid}}/2
\le\Delta_h/[8(2n+1)]$, so
\begin{equation}
\begin{aligned}
\operatorname{dist}(h,\partial\mathcal I_\star)
&\ge
\frac{w}{2}-\frac{\eta_{\mathrm{grid}}}{2}
\\
&\ge
\frac{\Delta_h}{2(2n+1)}-\frac{\Delta_h}{8(2n+1)}
=
\frac{3\Delta_h}{8(2n+1)}
>
\frac{\Delta_h}{4(2n+1)}
\ge
2\tau_1.
\end{aligned}
\end{equation}
Thus $[h-2\tau_1,h+2\tau_1]\subset\mathcal I_\star$, where $V(x)$ is
constant, proving Eq.~\eqref{eq:supp-certified-peeling-true-span-gap}.
\end{proof}

For each hidden block $a$, let
\begin{equation}
\mathcal A_a
:=
\left\{
U_{\mathrm c}
\bigl(R_{B_a}\otimes I_{B_a^c}\bigr)
U_{\mathrm c}^\dagger
:
R_{B_a}\in\Pc_{B_a}
\right\}
\subseteq
\Pc_n,
\end{equation}
the transformed phase-free single-block Pauli subspace.  We use the hidden
support decomposition and score factorization from
Lemma~\ref{lem:supp-hidden-block-pauli-decomposition} and
Eq.~\eqref{eq:supp-block-factorization-score}, with their signed/phase-free
convention.

\begin{lemma}[Blockwise high-score span]
\label{lem:supp-blockwise-high-score-span}
Assume that $\rho$ has the CEBP form of
Eq.~\eqref{eq:supp-clifford-encoded-block-product}.  For every real threshold
$h$, the true high-score span is generated by transformed single-block
high-score Paulis and satisfies
\begin{equation}
V(h)
=
\Span_{\Fc}
\left(
\bigcup_{a=1}^K
\bigl(\mathcal T(h)\cap\mathcal A_a\bigr)
\right).
\label{eq:supp-blockwise-high-score-span}
\end{equation}
\end{lemma}

\begin{proof}
If $P\in\mathcal T(h)$, every participating factor in
Eq.~\eqref{eq:supp-block-factorization-score} has score at least
$s_\rho(P)\ge h$, since all factors lie in $[0,1]$.  Hence the phase-free
components of $P$ lie in the displayed union and generate $P$.  This proves
one inclusion; the reverse follows because that union is contained in
$\mathcal T(h)$.
\end{proof}

\begin{lemma}[Anticommuting-Pauli score uncertainty]
\label{lem:supp-anticommuting-pauli-score-uncertainty}
For any state $\omega$ and any anticommuting $P,Q\in\Pc_n$,
\begin{equation}
s_\omega(P)+s_\omega(Q)\le1.
\label{eq:supp-anticommuting-score-uncertainty}
\end{equation}
\end{lemma}

\begin{proof}
Writing $\mu=\Tr(\omega P)$ and $\nu=\Tr(\omega Q)$, anticommutation gives
$(aP+bQ)^2=(a^2+b^2)I$.  Thus
$|a\mu+b\nu|\le\|aP+bQ\|_\infty=\sqrt{a^2+b^2}$.  Taking
$(a,b)=(\mu,\nu)$ (with the zero case immediate) yields
$\mu^2+\nu^2\le1$, which is
Eq.~\eqref{eq:supp-anticommuting-score-uncertainty}.
\end{proof}

Assume the CEBP model, the window and grid above, and
Eq.~\eqref{eq:supp-certified-peeling-grid-condition}.  Choose
\begin{equation}
M_1
\ge
\left\lceil
\frac{128(2n+1)^2}{\Delta_h^2}
\log\!\left(\frac{2\cdot 4^n}{\zeta_{\mathrm{BS}}}\right)
\right\rceil.
\label{eq:supp-certified-peeling-bell-round-count}
\end{equation}
Then $M_1$ Bell rounds use $2M_1$ copies and
\begin{equation}
\Pr(\mathcal E_{\mathrm{BS}})
\ge
1-\zeta_{\mathrm{BS}}.
\label{eq:supp-certified-peeling-event-probability}
\end{equation}
Run Algorithm~\ref{alg:supp-stabilizer-peeling} with this record and the
fixed grid and orders.

\begin{proposition}[Certified blockwise peeling span]
\label{prop:supp-certified-blockwise-peeling-span}
Under the assumptions and sampling condition stated above, on the event
$\mathcal E_{\mathrm{BS}}$, this run returns a threshold
$h\in\mathsf H_{\mathrm{peel}}$, a rank-$t$ basis $g_1,\dots,g_t$, and a Clifford
$U_{\mathrm{stab}}$ with the following properties.

The returned span satisfies
\begin{equation}
\Span_{\Fc}\{g_1,\dots,g_t\}
=
\Span_{\Fc}\widehat{\mathcal T}_{\mathrm{in}}(h)
=
V(h)
=
\Span_{\Fc}\widehat{\mathcal T}_{\mathrm{out}}(h).
\label{eq:supp-certified-peeling-returned-common-span}
\end{equation}
Setting $\lambda:=1-h\in(0,1/2)$, the basis is pairwise commuting and
satisfies
\begin{equation}
s_\rho(g_j)\ge 1-\lambda,
\qquad
j=1,\dots,t.
\label{eq:supp-certified-peeling-generator-score}
\end{equation}

If $t=0$, the returned basis is empty and the returned Clifford is
$U_{\mathrm{stab}}:=I$.  If $t>0$, the returned Clifford satisfies
\begin{equation}
U_{\mathrm{stab}}^\dagger g_jU_{\mathrm{stab}}=Z_j,
\qquad
j=1,\dots,t.
\end{equation}

The returned span also admits a blockwise basis: there exist independent
transformed single-block Paulis
$Q_{a(1)}^{\ell_1},\dots,Q_{a(t)}^{\ell_t}$, with
$Q_{a(i)}^{\ell_i}\in\mathcal T(h)\cap\mathcal A_{a(i)}$, that form a
basis of the same span:
\begin{equation}
\Span_{\Fc}\{g_1,\dots,g_t\}
=
\Span_{\Fc}
\bigl\{
Q_{a(i)}^{\ell_i}:\ i=1,\dots,t
\bigr\}.
\label{eq:supp-certified-peeling-single-block-span}
\end{equation}
\end{proposition}

\begin{proof}
Equations~\eqref{eq:supp-certified-peeling-bell-round-count}
and~\eqref{eq:supp-tauM} give
$\tau_1\le\Delta_h/[8(2n+1)]$.  Hence
Lemmas~\ref{lem:supp-stable-grid-threshold}
and~\ref{lem:supp-empirical-span-sandwich} supply a matching-span grid
candidate on $\mathcal E_{\mathrm{BS}}$.

For any matching candidate, the middle span-sandwich inclusions identify the
common span with $V(h)$, proving
Eq.~\eqref{eq:supp-certified-peeling-returned-common-span}; then
Lemma~\ref{lem:supp-blockwise-high-score-span} supplies the independent
single-block basis in
Eq.~\eqref{eq:supp-certified-peeling-single-block-span}.  Every inner-set
member has score at least $h>1/2$, so
Lemma~\ref{lem:supp-anticommuting-pauli-score-uncertainty} makes that set
commuting.  Thus the scan cannot reject a matching candidate and returns by
the guaranteed one.

For the returned $h$, membership
$g_j\in\widehat{\mathcal T}_{\mathrm{in}}(h)$ gives
$s_\rho(g_j)\ge\hat s(g_j)-\tau_1\ge h=1-\lambda$, proving
Eq.~\eqref{eq:supp-certified-peeling-generator-score}.  The algorithm uses
$U_{\mathrm{stab}}=I$ when $t=0$; otherwise symplectic completion gives
$U_0^\dagger g_jU_0=\operatorname{sgn}_jZ_j$, and multiplying $U_0$ by
$X_j$ for each negative sign gives
$U_{\mathrm{stab}}^\dagger g_jU_{\mathrm{stab}}=Z_j$.
\end{proof}

\subsection{Peeling guarantee and dominant syndrome}
\label{sec:supp-peeling-guarantee-syndrome-recovery}

We now turn the certified generator scores into a syndrome-product fidelity
guarantee.

Condition on the Bell event $\mathcal E_{\mathrm{BS}}$ and fix the accepted
output $h$, $\lambda:=1-h$, $g_1,\dots,g_t$, and
$U_{\mathrm{stab}}$ returned by
Algorithm~\ref{alg:supp-stabilizer-peeling} under the hypotheses of
Proposition~\ref{prop:supp-certified-blockwise-peeling-span}.  All certified
consequences in the first part of this subsection concern this fixed output
on $\mathcal E_{\mathrm{BS}}$.

On $\mathcal E_{\mathrm{BS}}$,
Proposition~\ref{prop:supp-certified-blockwise-peeling-span} gives
$s_\rho(g_j)\ge1-\lambda$ for every returned generator.  Set
\begin{equation}
\varepsilon_{\mathrm{peel}}
:=
\frac{t\lambda}{2}.
\label{eq:supp-eps-peel}
\end{equation}

\begin{corollary}[Empirical stabilizer peeling]
\label{cor:supp-empirical-stabilizer-peeling}
\label{prop:supp-empirical-peeling}
Let $h,\lambda,g_1,\dots,g_t,U_{\mathrm{stab}}$ be the output of
Algorithm~\ref{alg:supp-stabilizer-peeling} under the hypotheses of
Proposition~\ref{prop:supp-certified-blockwise-peeling-span}, with
$\lambda=1-h$.  Then, on $\mathcal E_{\mathrm{BS}}$, the sign-defined bit
string $b\in\{0,1\}^t$ of
Eq.~\eqref{eq:supp-deterministic-syndrome-bit} and a residual state
$\rho_{\mathrm{res}}$ satisfy
\begin{equation}
\F\!\left(
U_{\mathrm{stab}}^\dagger \rho\, U_{\mathrm{stab}},\,
|b\rangle\!\langle b|\otimes \rho_{\mathrm{res}}
\right)
\ge
\sqrt{\max\{0,1-\varepsilon_{\mathrm{peel}}\}}.
\label{eq:supp-empirical-peeling-fidelity}
\end{equation}
Consequently, the conclusion holds with probability at least
$1-\zeta_{\mathrm{BS}}$.
\end{corollary}

\begin{proof}
On $\mathcal E_{\mathrm{BS}}$, the returned generators are commuting and
satisfy $s_\rho(g_j)\ge1-\lambda$.  Apply
Proposition~\ref{prop:supp-approximate-stabilizer-decomposition} with
$\varepsilon_j^{\mathrm{stab}}=\lambda$ for all $j$ to obtain
Eq.~\eqref{eq:supp-empirical-peeling-fidelity}.
The event has probability at least $1-\zeta_{\mathrm{BS}}$, completing the
proof.
\end{proof}

Structural recovery uses only the generators, Clifford, accepted threshold,
error certificate, and comparison-state existence from
Eq.~\eqref{eq:supp-empirical-peeling-fidelity}; the learner recovers the
syndrome later by fresh signed Pauli measurements.

\section{Rank-Guided Recovery of Visible Pauli directions}
\label{sec:supp-rank-guided-recovery}

\subsection{Peeling-to-recovery interface}
\label{sec:supp-residual-representatives-score-control}

Condition on $\mathcal E_{\mathrm{BS}}$ and fix
$(h,\lambda:=1-h,t,g_1,\dots,g_t,U_{\mathrm{stab}})$ from
Proposition~\ref{prop:supp-certified-blockwise-peeling-span}.  Write
\begin{equation}
\mathcal S_{\mathrm{peel}}
:=
\Span_{\Fc}\{g_1,\dots,g_t\},
\qquad
m:=n-t,
\label{eq:supp-residual-qubit-number}
\end{equation}
and define the actual peeled state
\begin{equation}
\widetilde\rho
:=
U_{\mathrm{stab}}^\dagger\rho U_{\mathrm{stab}}.
\label{eq:supp-recovery-peeled-state}
\end{equation}
Corollary~\ref{cor:supp-empirical-stabilizer-peeling} supplies
$b\in\{0,1\}^t$ and a state $\rho_{\mathrm{res}}$ on the $m$ residual
qubits for the proof-only comparison
$|b\rangle\!\langle b|\otimes\rho_{\mathrm{res}}$; the protocol operates
only on $\widetilde\rho$.

For $A\in\Pc_n$, write
$s_{\mathrm{full}}(A):=s_{\widetilde\rho}(A)$.
Define the high-score residual-candidate class
\begin{equation}
\mathcal Q_{\mathrm{cand}}(\lambda)
:=
\{Q\in\Pc_n:s_\rho(Q)>\lambda\}.
\label{eq:supp-high-score-residual-candidate-class}
\end{equation}

\begin{proposition}[Residual representatives after peeling]
\label{prop:supp-residual-representatives-after-peeling}
For every $Q\in\mathcal Q_{\mathrm{cand}}(\lambda)$, the peeled Pauli
$R:=U_{\mathrm{stab}}^\dagger Q U_{\mathrm{stab}}$ commutes with
$Z_1,\dots,Z_t$.  Consequently, there exist
$c\in\Fc^t$, $P_{\mathrm{res}}\in\Pc_m$, and
$\chi_Q\in\{\pm1\}$ such that
\begin{equation}
R=
\chi_Q\bigl(
Z^c\otimes P_{\mathrm{res}}
\bigr),
\qquad
Z^c:=Z_1^{c_1}\cdots Z_t^{c_t}.
\label{eq:supp-peeled-pauli-residual-representative}
\end{equation}
\end{proposition}

\begin{proof}
The candidate and generator scores satisfy
$s_\rho(Q)>\lambda$ and $s_\rho(g_j)\ge1-\lambda$; hence
Lemma~\ref{lem:supp-anticommuting-pauli-score-uncertainty} forces
$[Q,g_j]=0$ for every $j$.  After conjugation, $R$ commutes with every
$Z_j$ and therefore has the form
\eqref{eq:supp-peeled-pauli-residual-representative}, including a harmless
Hermitian-representative sign.
\end{proof}

\subsection{Peeled hidden-block geometry}

Conjugation by $U_{\mathrm{stab}}^\dagger$ transports the hidden
single-block groups $\mathcal A_a$, their internal-product decomposition,
and the support factorization of
Lemma~\ref{lem:supp-hidden-block-pauli-decomposition}.  Write
\begin{equation}
\mathcal R_a
:=
U_{\mathrm{stab}}^\dagger\mathcal A_aU_{\mathrm{stab}},
\qquad
R_a^\ell
:=
U_{\mathrm{stab}}^\dagger Q_a^\ell U_{\mathrm{stab}},
\qquad
a\in[K],
\quad
\ell\in[4^{L_a}-1].
\label{eq:supp-peeled-block-operators}
\end{equation}
With all group operations phase-free, unitary invariance gives
$s_{\mathrm{full}}(R_a^\ell)=s_\rho(Q_a^\ell)$.

The peeled syndrome group is
\begin{equation}
\mathcal K
:=
U_{\mathrm{stab}}^\dagger
\mathcal S_{\mathrm{peel}}
U_{\mathrm{stab}}
=
\{Z^c\otimes I:c\in\Fc^t\}.
\label{eq:supp-peeled-kernel-blockwise-proof}
\end{equation}
Set
\begin{equation}
\mathcal S_a
:=
\mathcal S_{\mathrm{peel}}\cap\mathcal A_a,
\qquad
\mathcal K_a
:=
U_{\mathrm{stab}}^\dagger\mathcal S_aU_{\mathrm{stab}}
\subseteq\mathcal R_a.
\end{equation}
The blockwise basis in
Proposition~\ref{prop:supp-certified-blockwise-peeling-span} and transported
internal-product structure give
\begin{equation}
\mathcal K
=
\mathcal K_1\cdots\mathcal K_K
\label{eq:supp-blockwise-kernel-factorization}
\end{equation}
as an internal product.

Let
\begin{equation}
\mathcal Z_{\mathrm{pref}}
:=
\{Z^c\otimes P:c\in\Fc^t,\ P\in\Pc_m\}
\end{equation}
and define the residual quotient homomorphism
\begin{equation}
\pi:\mathcal Z_{\mathrm{pref}}\longrightarrow\Pc_m,
\qquad
\pi(Z^c\otimes P):=P.
\end{equation}
Then $\ker\pi=\mathcal K$.  For a Pauli subgroup $\mathcal G$, write
\begin{equation}
\pi(\mathcal G)
:=
\pi(\mathcal G\cap\mathcal Z_{\mathrm{pref}}).
\end{equation}
\begin{proposition}[Blockwise peeling preserves residual block separation]
\label{prop:supp-blockwise-peeling-preserves-residual-blocks}
The residual images of the hidden block algebras
remain separated:
\begin{equation}
\pi(\mathcal R_a)
\cap
\left\langle
\pi(\mathcal R_b): b\ne a
\right\rangle
=
\{I\},
\qquad
a\in[K].
\label{eq:supp-residual-block-separation}
\end{equation}
\end{proposition}

\begin{proof}
If $\prod_a\bar R_a=I$ with
$\bar R_a\in\pi(\mathcal R_a)$, choose lifts
$R_a\in\mathcal R_a\cap\mathcal Z_{\mathrm{pref}}$.  Then
$\prod_aR_a\in\ker\pi=\mathcal K$, so
Eq.~\eqref{eq:supp-blockwise-kernel-factorization} gives
$\prod_aR_a=\prod_aK_a$ for some
$K_a\in\mathcal K_a\subseteq\mathcal R_a$.  Hence
$\prod_aR_aK_a^{-1}=I$.  The transported internal-product property forces
$R_a=K_a$ blockwise, and therefore $\bar R_a=\pi(R_a)=I$ for every $a$.
\end{proof}

Thus the residual block images commute across blocks, are individually
multiplication-closed, and give every nonidentity member a unique block
label.  We use these three consequences in the recovery induction.

\subsection{Visible residual directions}

A peeled block Pauli $R_a^\ell$ has a residual representative precisely
when $R_a^\ell\in\mathcal Z_{\mathrm{pref}}$.  In that case there are unique
$c(a,\ell)\in\Fc^t$, $\bar R_a^\ell\in\Pc_m$, and
$\chi_{a,\ell}\in\{\pm1\}$ such that
\begin{equation}
R_a^\ell
=
\chi_{a,\ell}
\bigl(Z^{c(a,\ell)}\otimes\bar R_a^\ell\bigr).
\label{eq:supp-visible-residual-representative}
\end{equation}
We call the direction \emph{visible} when $\bar R_a^\ell\ne I$ and
\emph{prefix-only} when $\bar R_a^\ell=I$.

If a visible $R_a^\ell$ had
$s_{\mathrm{full}}(R_a^\ell)\ge h=1-\lambda$, then
$Q_a^\ell\in\mathcal T(h)\cap\mathcal A_a$ would lie in the certified span
$V(h)=\mathcal S_{\mathrm{peel}}$.  Its peeled image would be prefix-only,
contradicting visibility.  Hence
\begin{equation}
s_{\mathrm{full}}(R_a^\ell)
\le
1-\lambda,
\label{eq:supp-post-peeling-upper-block}
\end{equation}
for every visible $R_a^\ell$; this is the full-score ceiling used by
ranking.

The transported support from
Lemma~\ref{lem:supp-hidden-block-pauli-decomposition} is the unique
$\supp_R(A)$ contained in
\begin{equation}
\bigcup_{a=1}^K
\bigl(\{a\}\times[4^{L_a}-1]\bigr)
\label{eq:supp-R-support-domain}
\end{equation}
with at most one label per block and satisfying
\begin{equation}
A
=
\chi_R(A)
\prod_{(a,\ell)\in\supp_R(A)}R_a^\ell,
\qquad
\chi_R(A)\in\{\pm1\}.
\label{eq:supp-block-proof-support-decomp}
\end{equation}
The sign compares Hermitian
representatives, while support and span statements are phase-free.  Its
transported score factorization is
\begin{equation}
s_{\mathrm{full}}(A)
=
\prod_{(a,\ell)\in\supp_R(A)}
s_{\mathrm{full}}(R_a^\ell).
\label{eq:supp-high-score-visible-decomp-factorization}
\end{equation}

Let
\begin{equation}
\mathcal V
:=
\{(a,\ell):R_a^\ell\text{ is visible}\}
\end{equation}
be the visible label set.  A visible residual decomposition of
$P\in\Pc_m\setminus\{I\}$ is a subset
$S\subseteq \mathcal V$
containing at most one label from each hidden block and satisfying, up to
phase,
\begin{equation}
P
=
\prod_{(a,\ell)\in S}
\bar R_a^\ell .
\label{eq:supp-visible-residual-decomp}
\end{equation}
When such subsets exist, choose one with minimum cardinality, using a fixed
deterministic ordering to break ties, and denote it by
$\supp_{\bar R}(P)$.  We call $P$ a \emph{visible residual composite} if
\begin{equation}
|\supp_{\bar R}(P)|\ge 2.
\end{equation}
Prefix-only factors, for which $\bar R_a^\ell=I$, are discarded in
$\supp_{\bar R}(P)$.

\begin{proposition}[High-score full Paulis have visible residual decompositions]
\label{prop:supp-high-score-visible-residual-decomposition}
Let $A\in\Pc_n$ be a full peeled Pauli of the form
\begin{equation}
A=Z^c\otimes P,
\qquad
P\in\Pc_m\setminus\{I\}.
\end{equation}
If
\begin{equation}
s_{\mathrm{full}}(A)>\lambda,
\label{eq:supp-high-score-visible-decomp-condition}
\end{equation}
then $P$ admits a visible residual decomposition in the sense of
Eq.~\eqref{eq:supp-visible-residual-decomp}.
\end{proposition}

\begin{proof}
Equation~\eqref{eq:supp-high-score-visible-decomp-factorization} and
$s_{\mathrm{full}}(A)>\lambda$ imply
$s_{\mathrm{full}}(R_a^\ell)>\lambda$ for every factor.  Since this equals
$s_\rho(Q_a^\ell)$,
Proposition~\ref{prop:supp-residual-representatives-after-peeling} applies to
each factor.  Multiplying the resulting prefix--residual representatives and
comparing residual parts with $A=Z^c\otimes P$ expresses $P$ as the product
of the nonidentity $\bar R_a^\ell$.  The transported support contains at most
one label per block, so this is a visible residual decomposition; it is
nonempty because $P\ne I$.
\end{proof}

\subsection{Rank-guided symplectic recovery}

Conditioned on the fixed peeling output, take $2M_2$ fresh copies of
$\rho$, apply $U_{\mathrm{stab}}^\dagger$ to every copy, and pair them into
$M_2$ Bell-sampling rounds.  These are operational copies of $\widetilde\rho$;
the procedure never prepares or postselects $\rho_{\mathrm{res}}$.
For $A\in\Pc_n$, let $\hat s_{\mathrm{full}}(A)$ denote the Bell-score
estimator from these $M_2$ fresh peeled Bell rounds on $\widetilde\rho$, so it
estimates $s_{\mathrm{full}}(A)=s_{\widetilde\rho}(A)$.

Set
\begin{equation}
\tau_{\mathrm{rank}}
:=
\tau(M_2,n,\zeta_{\mathrm{rank}})
\end{equation}
and define the separate recovery concentration event
\begin{equation}
\mathcal E_{\mathrm{rank}}
:=
\left\{
|\hat s_{\mathrm{full}}(A)-s_{\mathrm{full}}(A)|
\le
\tau_{\mathrm{rank}}
\text{ for every }A\in\Pc_n
\right\}.
\label{eq:supp-block-uniform-score-event}
\end{equation}
For every fixed peeling record and resulting $U_{\mathrm{stab}}$,
Proposition~\ref{prop:supp-uniform-concentration} gives
\begin{equation}
\Pr(\mathcal E_{\mathrm{rank}}\mid\text{fixed peeling output})
\ge
1-\zeta_{\mathrm{rank}}.
\end{equation}
Thus $\mathcal E_{\mathrm{rank}}$ is distinct from
$\mathcal E_{\mathrm{BS}}$ and uses only the fresh recovery rounds.

For a recovery threshold $\theta\in(0,1)$, define
\begin{equation}
\mathcal L_\theta
:=
\{A\in\Pc_n\setminus\{I\}:\hat s_{\mathrm{full}}(A)\ge\theta\}.
\end{equation}
Order $\mathcal L_\theta$ by decreasing empirical full score, breaking
ties by a fixed deterministic ordering of $\Pc_n$.

A recovered sector is either a singleton
$\mathcal G_\alpha=\{\widehat R_\alpha^x\}$ or a completed triple
\begin{equation}
\mathcal G_\alpha
=
\{\widehat R_\alpha^x,\widehat R_\alpha^z,\widehat R_\alpha^y\},
\qquad
\widehat R_\alpha^y
=
\widehat R_\alpha^x\widehat R_\alpha^z
\end{equation}
in the phase-free residual Pauli group.

For a sector collection $\mathfrak G$, define its independent-axis set and
recovered symplectic span by
\begin{equation}
\begin{aligned}
\operatorname{Ax}(\mathfrak G)
&:=
\{\widehat R_\alpha^x:\mathcal G_\alpha\text{ is a singleton}\}
\\
&\quad\cup
\{\widehat R_\alpha^x,\widehat R_\alpha^z:
\mathcal G_\alpha\text{ is a completed triple}\},
\\
V(\mathfrak G)
&:=
\Span_{\Fc}\operatorname{Ax}(\mathfrak G).
\end{aligned}
\end{equation}
We call $\mathfrak G$ \emph{symplectically valid} when the three
nonidentity members of each completed triple obey the displayed
Pauli-frame relation and are pairwise anticommuting, distinct sectors
commute elementwise, and the Pauli directions in
$\operatorname{Ax}(\mathfrak G)$ are linearly independent.

For a symplectically valid collection and residual Pauli $P$, define the
completed-sector dressing
\begin{equation}
\operatorname{Dress}_{\mathfrak G}(P)
:=
P
\prod_{\beta:\mathcal G_\beta
=\{\widehat R_\beta^x,\widehat R_\beta^z,\widehat R_\beta^y\}}
(\widehat R_\beta^x)^{[P,\widehat R_\beta^z]}
(\widehat R_\beta^z)^{[P,\widehat R_\beta^x]},
\label{eq:supp-triple-dressing-map}
\end{equation}
where $[\cdot,\cdot]$ is the symplectic commutation form.  For
$P^\circ:=\operatorname{Dress}_{\mathfrak G}(P)$, set
\begin{equation}
A_{\mathfrak G}(P^\circ)
:=
\{\gamma:\mathcal G_\gamma=\{\widehat R_\gamma^x\}
\text{ and }[P^\circ,\widehat R_\gamma^x]=1\}.
\label{eq:supp-singleton-anticommutation-set}
\end{equation}
When this set is nonempty, pivot rebasing chooses its fixed first element
$\alpha$, replaces
$\widehat R_\gamma^x$ by
$\widehat R_\gamma^x\widehat R_\alpha^x$ for every
$\gamma\in A_{\mathfrak G}(P^\circ)\setminus\{\alpha\}$, and
then completes the pivot with
$\widehat R_\alpha^z:=P^\circ$ and
$\widehat R_\alpha^y:=\widehat R_\alpha^x\widehat R_\alpha^z$.

\begin{inlinealgorithm}{alg:supp-block-rank-guided-recovery}{Rank-Guided Recovery of Visible Single-Block Pauli Sectors}
\begin{algorithmic}[1]
\Require
$\mathcal L_\theta$ from the $M_2$ fresh Bell rounds.
\Ensure A recovered sector collection $\mathfrak G$.

\State Initialize $\mathfrak G\gets\emptyset$.

\For{each $A\in\mathcal L_\theta$ in the fixed ranked order}
    \If{$A$ cannot be written, up to phase, as $Z^c\otimes P$ with
    $c\in\Fc^t$ and $P\in\Pc_m$}
        \State \textbf{continue}
    \EndIf
    \If{$P=I$}
        \State \textbf{continue}
    \EndIf
    \If{$P\in V(\mathfrak G)$}
        \State \textbf{continue}
    \EndIf

    \State Set $P^\circ\gets \operatorname{Dress}_{\mathfrak G}(P)$ and
    $B\gets A_{\mathfrak G}(P^\circ)$.

    \If{$B=\emptyset$}
        \State Append the singleton sector
        $\mathcal G_{\mathrm{new}}:=\{P^\circ\}$
        to $\mathfrak G$.
    \Else
        \State Let $\alpha$ be the fixed first element of $B$.
        \For{each $\gamma\in B\setminus\{\alpha\}$}
            \State Rebase the singleton axis
            $\widehat R_\gamma^x\gets \widehat R_\gamma^x\widehat R_\alpha^x$.
        \EndFor
        \State Complete the pivot singleton by setting
        $\widehat R_\alpha^z\gets P^\circ, \qquad \widehat R_\alpha^y\gets \widehat R_\alpha^x\widehat R_\alpha^z$.
        \State Replace
        $
        \mathcal G_\alpha=\{\widehat R_\alpha^x\}
        $
        by
        $\mathcal G_\alpha := \{\widehat R_\alpha^x,\widehat R_\alpha^z,\widehat R_\alpha^y\}$.
    \EndIf
\EndFor

\State \Return $\mathfrak G$.
\end{algorithmic}
\end{inlinealgorithm}

\subsection{Ranking and recovery correctness}

Assume the threshold margin
\begin{equation}
\theta-\tau_{\mathrm{rank}}>\lambda
\label{eq:supp-visible-threshold-condition}
\end{equation}
and the ranking margin
\begin{equation}
\lambda(\theta-\tau_{\mathrm{rank}})
-2\tau_{\mathrm{rank}}>0.
\label{eq:supp-block-ranking-gap-condition}
\end{equation}
For the remainder of this subsection, condition on the distinct events
$\mathcal E_{\mathrm{BS}}$ and $\mathcal E_{\mathrm{rank}}$.

The threshold margin first gives the required visible decompositions.  If
$A=Z^c\otimes P\in\mathcal L_\theta$ with $P\ne I$, then
\begin{equation}
s_{\mathrm{full}}(A)
\ge
\hat s_{\mathrm{full}}(A)-\tau_{\mathrm{rank}}
\ge
\theta-\tau_{\mathrm{rank}}
>
\lambda.
\end{equation}
Proposition~\ref{prop:supp-high-score-visible-residual-decomposition}
therefore gives a visible residual decomposition of $P$.

\begin{lemma}[Ranking gap for visible composites]
\label{lem:supp-visible-composite-ranking-gap}
Let $A=Z^c\otimes P\in\mathcal L_\theta$, with $P\neq I$, and suppose
that $P$ is a visible residual composite.  Then the decomposition induced
by $\supp_R(A)$ has the property that every visible component $(a,\ell)$
satisfies
\begin{equation}
\hat s_{\mathrm{full}}(R_a^\ell)>\hat s_{\mathrm{full}}(A).
\label{eq:supp-block-proof-empirical-gap}
\end{equation}
In particular, each such visible component also survives the recovery
threshold.
\end{lemma}

\begin{proof}
By Eq.~\eqref{eq:supp-visible-threshold-condition},
$\theta-\tau_{\mathrm{rank}}>\lambda$.  Thus
Proposition~\ref{prop:supp-high-score-visible-residual-decomposition} ensures
that $P$ has a visible residual decomposition.  Expanding $A$ with
$\supp_R(A)$ and removing prefix-only factors gives such a decomposition.
Because $P$ is a visible residual composite, this visible part contains at
least two labels.

Fix a visible label $(a,\ell)\in\supp_R(A)$.  The block-factorization identity
in the full peeled basis gives
\begin{equation}
\begin{aligned}
&s_{\mathrm{full}}(R_a^\ell)
\left(1-
\prod_{\substack{(b,\ell')\in\supp_R(A)\\(b,\ell')\ne(a,\ell)}}
s_{\mathrm{full}}(R_b^{\ell'})
\right)
\\
&\qquad =
s_{\mathrm{full}}(R_a^\ell)-s_{\mathrm{full}}(A).
\end{aligned}
\end{equation}
The product over the other labels contains at least one visible nonidentity
factor, so the post-peeling separation condition and the bound
$s_{\mathrm{full}}(\cdot)\le1$ imply
\begin{equation}
s_{\mathrm{full}}(R_a^\ell)-s_{\mathrm{full}}(A)
\ge
\lambda\,s_{\mathrm{full}}(R_a^\ell).
\end{equation}
Moreover,
\begin{equation}
s_{\mathrm{full}}(R_a^\ell)
\ge
s_{\mathrm{full}}(A)
\ge
\hat s_{\mathrm{full}}(A)-\tau_{\mathrm{rank}}
\ge
\theta-\tau_{\mathrm{rank}} .
\end{equation}
Combining the last two displays gives
\begin{equation}
s_{\mathrm{full}}(R_a^\ell)-s_{\mathrm{full}}(A)
\ge
\lambda(\theta-\tau_{\mathrm{rank}}).
\label{eq:supp-block-proof-gap}
\end{equation}
Applying Eq.~\eqref{eq:supp-block-uniform-score-event} to both $A$ and
$R_a^\ell$ gives Eq.~\eqref{eq:supp-block-proof-empirical-gap} by
Eq.~\eqref{eq:supp-block-ranking-gap-condition}.  Since
$\hat s_{\mathrm{full}}(R_a^\ell)>\hat s_{\mathrm{full}}(A)\ge\theta$,
each such component also survives the threshold.
\end{proof}

\begin{definition}[Block-pure recovered sector]
A recovered sector is \emph{block-pure} if all of its nonidentity operators
belong to $\pi(\mathcal R_a)$ for one hidden block $a$.
\end{definition}

\begin{proposition}[Recovery validity, block purity, and completeness]
\label{prop:supp-block-triple-single-block}
Let
\begin{equation}
\mathfrak G
=
\{\mathcal G_1,\dots,\mathcal G_L\}
\end{equation}
be the output of Algorithm~\ref{alg:supp-block-rank-guided-recovery}.  Then
every recovered sector is block-pure; every completed triple obeys
$\widehat R_\alpha^y=\widehat R_\alpha^x\widehat R_\alpha^z$ phase-free and consists
of three pairwise anticommuting operators; distinct sectors commute
elementwise; and the Pauli directions in $\operatorname{Ax}(\mathfrak G)$ are linearly
independent over $\Fc$.  Moreover, every threshold-surviving nonidentity
residual action is generated by the recovered sectors:
\begin{equation}
\begin{aligned}
&\{P\in\Pc_m\setminus\{I\}:\exists c\in\Fc^t
\text{ such that }
\hat s_{\mathrm{full}}(Z^c\otimes P)\ge\theta\}
\\
&\qquad\subseteq
V(\mathfrak G).
\end{aligned}
\label{eq:supp-threshold-span-completeness}
\end{equation}
\end{proposition}

\begin{proof}
Induct over the ranked list $\mathcal L_\theta$, fixing the visible
decomposition of each non-prefix survivor.  The invariant is: all current
sectors are block-pure; completed triples are Pauli frames; distinct sectors
commute; their Pauli directions are independent; and every processed nonidentity residual
action lies in $V:=V(\mathfrak G)$.  It holds initially.  We use throughout
the cross-block commutation, same-block closure, and block-label uniqueness
following
Proposition~\ref{prop:supp-blockwise-peeling-preserves-residual-blocks}.

Let $A$ be next.  The recovery event and threshold margin give
$s_{\mathrm{full}}(A)\ge\theta-\tau_{\mathrm{rank}}>\lambda$.
Proposition~\ref{prop:supp-residual-representatives-after-peeling}
therefore gives $A=Z^c\otimes P$ up to phase.  The algorithm skips $P=I$
and $P\in V$, so consider $P\notin V$.

Suppose, for contradiction, that $P$ is a visible residual composite.
Let $S_A$ be the non-prefix labels in $\supp_R(A)$, so phase-free
$P=\prod_{(a,\ell)\in S_A}\bar R_a^\ell$ and $|S_A|\ge2$.
Lemma~\ref{lem:supp-visible-composite-ranking-gap} places every corresponding
$R_a^\ell$ strictly before $A$ in $\mathcal L_\theta$.  Their nonidentity
residual representatives lie in $V$ by the invariant, so their product
$P$ lies in $V$, a contradiction.  Thus $P$ is a single visible
$\bar R_a^\ell$, with unique hidden-block label $a$.

In dressing, cross-block Pauli directions have zero exponent and same-block-products stay
in $\pi(\mathcal R_a)$; hence $P^\circ$ remains block-$a$ pure and commutes
with every completed triple.  Because the dressing factors lie in $V$,
\begin{equation}
V+\Span_{\Fc}\{P^\circ\}
=
V+\Span_{\Fc}\{P\}.
\end{equation}
Thus $P^\circ\ne I$ and $P^\circ\notin V$.

If $P^\circ$ commutes with every singleton sector, the algorithm appends
the block-pure singleton $\{P^\circ\}$, preserving commutation and
independence while adding the direction $P$.

Otherwise, let the algorithm choose a pivot among the singleton sectors
that anticommute with $P^\circ$.  By uniqueness of nonidentity block
labels and cross-block commutation, every such singleton is block-$a$
pure.  Rebasing each nonpivot singleton by multiplication with the pivot
is an invertible same-block basis change; the rebased axis commutes with
$P^\circ$ because both factors anticommute with it, and other cross-sector
commutators remain zero.  Completing the pivot sets
$\widehat R_\alpha^z=P^\circ$ and
$\widehat R_\alpha^y=\widehat R_\alpha^x\widehat R_\alpha^z$.  The resulting
$x$, $z$, and $y$ operators anticommute pairwise and remain block-$a$
pure by same-block multiplication closure.  The rebasing preserves the
old span, and $P^\circ\notin V$ makes the new $z$ axis independent of all
previous Pauli directions.  Thus completion adds $\Span_{\Fc}\{P\}$ and preserves the
invariant.

At termination every nonidentity survivor has entered the span or was already
in it, proving Eq.~\eqref{eq:supp-threshold-span-completeness}; the other
invariant clauses give the stated validity and block-purity conclusions.
\end{proof}

The recovery stage uses exactly $M_2$ fresh two-copy Bell rounds, hence
$2M_2$ fresh copies of $\rho$, with $U_{\mathrm{stab}}^\dagger$ applied
before measurement.  Freshness and the conditional recovery guarantee give
\begin{equation}
\Pr(\mathcal E_{\mathrm{BS}}\cap\mathcal E_{\mathrm{rank}})
\ge
1-\zeta_{\mathrm{BS}}-\zeta_{\mathrm{rank}}.
\end{equation}
Although these rounds estimate all full-Pauli scores simultaneously,
forming and sorting $\mathcal L_\theta$ still scans $4^n$ phase-free
Paulis, so the classical runtime is not optimized.

Finally, block purity does not imply one recovered sector per hidden block.
Several mutually commuting recovered sectors can carry the same unique
hidden-block label, so the rank-guided output may over-refine one hidden
block.  The next section uses mixed cumulants to group such block-pure
sectors into the hidden-block structure.

\section{Hidden Block Grouping by Mixed Cumulants}
\label{sec:supp-cumulant-grouping}

\subsection{Grouping setup and target sector partition}

Condition throughout on successful certified peeling and on the recovery
event of Proposition~\ref{prop:supp-block-triple-single-block}.  Thus every
recovered sector is block-pure.  Write the recovered sectors as
\begin{equation}
\mathfrak G=\{\mathcal G_1,\dots,\mathcal G_L\},
\end{equation}
where $L\le m\le n$.  All operators in these sectors are residual Paulis on
the $m$ unpeeled qubits.  Each sector is either a singleton axis or a
completed Pauli triple; distinct sectors commute elementwise, and the
recovered binary Pauli directions are independent.

For each sector $\mathcal G_\alpha$, let $a(\alpha)$ be its unique
transformed hidden-block label, and define the true recovered-sector blocks
and their represented labels by
\begin{equation}
\mathcal B_a^\star
:=
\{\alpha\in[L]:a(\alpha)=a\},
\qquad
\mathcal A^\star
:=
\{a\in[K]:\mathcal B_a^\star\ne\varnothing\},
\qquad
K^\star:=|\mathcal A^\star|.
\label{eq:supp-true-recovered-sector-blocks}
\end{equation}
They form the target partition
\begin{equation}
[L]
=
\bigsqcup_{a\in\mathcal A^\star}\mathcal B_a^\star.
\label{eq:supp-true-recovered-sector-partition}
\end{equation}
If $L\le1$, this partition is already known and there is nothing to group;
henceforth assume $L\ge2$.

Choose a maximum grouping order $\ell_{\mathrm{grp}}$ satisfying
\begin{equation}
\ell_{\mathrm{grp}}
\ge
\max_{a\in\mathcal A^\star}|\mathcal B_a^\star|.
\label{eq:supp-grouping-order-promise}
\end{equation}
The recovered-sector cap follows from
Proposition~\ref{prop:supp-block-triple-single-block}.  Indeed, fix a hidden
block $a$ of size $L_a$.  Its recovered triple sectors contribute mutually
orthogonal hyperbolic planes, while its recovered singleton Pauli directions are
independent and span an isotropic subspace in the symplectic complement of
those planes.  If there are $x$ triples and $s$ singletons, this gives
$s\le L_a-x$ and hence
$|\mathcal B_a^\star|=x+s\le L_a$.  Thus, when the available block-size
prior gives $L_a\le d$, the known recovered-sector cap is
\begin{equation}
\max_{a\in\mathcal A^\star}|\mathcal B_a^\star|\le d,
\end{equation}
and we use the specialization $\ell_{\mathrm{grp}}=d$.  If
$\ell_{\mathrm{grp}}=1$, every true recovered-sector block is a singleton,
so again no tests or merges are needed.  The nontrivial analysis below
therefore assumes $\ell_{\mathrm{grp}}\ge2$.

For $A\subseteq[L]$, define the generated residual Pauli group
\begin{equation}
\mathcal P(A)
:=
\left\langle
\bigcup_{\alpha\in A}\mathcal G_\alpha
\right\rangle,
\label{eq:supp-cluster-pauli-set}
\end{equation}
where Pauli phases are ignored and the recovered-sector generator convention
is retained.  Each recovered sector is either a singleton, which supplies
one independent binary Pauli axis, or a completed Pauli triple, whose three
nonidentity elements have only two independent binary Pauli directions.  Consequently,
the generators contributed by $A$ have binary rank at most $2|A|$, and
therefore
\begin{equation}
|\mathcal P(A)|
\le
2^{2|A|}
=
4^{|A|}.
\label{eq:supp-generated-group-cardinality}
\end{equation}

For $P_0\in\Pc_m$ and $\chi\in\{\pm1\}$, define the residual signed moment
on signed Hermitian Paulis by
\begin{equation}
\mu_{\mathrm{res}}(\chi P_0)
:=
\Tr\!\left[
\widetilde\rho\,
(I^{\otimes t}\otimes \chi P_0)
\right]
=
\chi\,\mu_{\mathrm{res}}(P_0).
\label{eq:supp-residual-signed-moment}
\end{equation}
When $t=0$, this is the usual signed Pauli moment.  We use canonical Hermitian
representatives of phase-free Paulis.  Every product used for a subset moment
is reduced to its canonical representative together with its known Hermitian
sign.  Changing the sign of one displayed tuple member changes the full mixed
cumulant only by the same overall sign, so the absolute-value tests below are
representative-independent.

\subsection{Mixed cumulants and moment perturbations}
\label{sec:supp-mixed-cumulants}
Pairwise cumulants can miss higher-order dependence, as the next subsection
shows.  We therefore use standard classical mixed cumulants of the commuting
Pauli tuples tested by grouping.  Their partition and logarithmic
generating-function identities are standard
\cite{doubilet1972foundations,smith2020multivariateKstatistics}; the new
ingredient below is their robust transfer through peeling.

\begin{definition}[Mixed cumulant of a labeled commuting tuple]
\label{def:supp-mixed-cumulant-sector-set}
Let $\rho$ be a state and let
$\mathbf O=(O_i)_{i\in[q]}$
be a finite labeled tuple of mutually commuting operators.  For commuting
observables this is the classical joint cumulant of their joint outcomes.
Position labels are retained even when operator values repeat.  For
$U\subseteq[q]$, set
\begin{equation}
M_\rho(U;\mathbf O)
:=
\Tr\!\left(\rho\prod_{i\in U}O_i\right),
\label{eq:supp-signed-pauli-moment-general}
\end{equation}
with the convention $M_\rho(\varnothing;\mathbf O):=1$.

Its \emph{mixed cumulant} is
\begin{equation}
\begin{aligned}
\kappa_\rho(\mathbf O)
:={}&
\sum_{\mathfrak p\in \Pi([q])}
(|\mathfrak p|-1)!(-1)^{|\mathfrak p|-1}
\prod_{E\in \mathfrak p}
M_\rho(E;\mathbf O),
\end{aligned}
\label{eq:supp-general-mixed-cumulant}
\end{equation}
where $\Pi([q])$ denotes the set of all partitions of the position labels
$[q]$ into nonempty disjoint blocks.  Thus coincident operator values
$O_i=O_j$ remain distinct labeled entries when $i\ne j$.
\end{definition}

We write
$M_\rho(O_1,\dots,O_q):=M_\rho([q];\mathbf O)$ and
$\kappa_\rho(O_1,\dots,O_q):=\kappa_\rho(\mathbf O)$, always retaining the
position labels.  For commuting residual Paulis $O_1,\dots,O_q\in\Pc_m$,
write
\begin{equation}
\kappa_{\mathrm{res}}(O_1,\dots,O_q)
:=
\kappa_{\widetilde\rho}
\bigl(
I^{\otimes t}\otimes O_1,\dots,
I^{\otimes t}\otimes O_q
\bigr),
\label{eq:supp-residual-cumulant-shorthand}
\end{equation}
computed from the signed subset moments using the chosen Hermitian
representatives.  The following one-sided implication is the
product-split vanishing principle used for grouping
\cite{doubilet1972foundations,smith2020multivariateKstatistics}, specialized
to the labeled subset moments used here.

\begin{proposition}[Product-split moment factorization forces cumulant vanishing]
\label{prop:supp-nonzero-cumulant-certifies-no-product-split}
Let $\rho$ be a state and let
$\mathbf O=(O_i)_{i\in[q]}$ be a nonempty finite labeled tuple of mutually
commuting operators.  Suppose there exists a nontrivial
split
\begin{equation}
[q]=\mathsf U\sqcup\mathsf V,
\qquad \mathsf U\neq\varnothing,
\qquad \mathsf V\neq\varnothing,
\end{equation}
such that the moment family factorizes across this split: for every
$U\subseteq[q]$,
\begin{equation}
M_\rho(U;\mathbf O)
=
M_\rho(U\cap\mathsf U;\mathbf O)
M_\rho(U\cap\mathsf V;\mathbf O).
\label{eq:supp-moment-factorization-split}
\end{equation}
Then
$\kappa_\rho(\mathbf O)=0$.
\end{proposition}

\begin{proof}
For formal variables indexed by the position labels, define
\begin{equation}
F_{\mathbf O}(z)
:=
\sum_{U\subseteq[q]}
M_\rho(U;\mathbf O)
\prod_{i\in U}z_i .
\end{equation}
The standard formal-power-series identity gives
$\kappa_\rho(\mathbf O)=[z_1\cdots z_q]\log F_{\mathbf O}(z)$.
Equation~\eqref{eq:supp-moment-factorization-split} implies
$F_{\mathbf O}=F_{\mathsf U}F_{\mathsf V}$, hence
$\log F_{\mathbf O}=\log F_{\mathsf U}+\log F_{\mathsf V}$.
Neither summand contains every position variable, so the displayed
coefficient vanishes.
\end{proof}

A nonzero cumulant therefore rules out the tested product split.  This
certifies only the selected tuple (or its generated algebras); a zero
cumulant does not imply state factorization.

For $q\ge 1$, define
\begin{equation}
\Gamma_q
:=
\sum_{\mathfrak p\in\Pi([q])}
(|\mathfrak p|-1)!\,|\mathfrak p|
=
\sum_{k=1}^q S(q,k)k!,
\label{eq:supp-cumulant-gamma-s}
\end{equation}
where $S(q,k)$ is a Stirling number of the second kind.
The constants $\Gamma_q$ are the ordered Bell, or Fubini, numbers.  They are
distinct from the ordinary Bell numbers
$B_q:=|\Pi([q])|$, which count unordered set partitions
\cite{smith2020multivariateKstatistics} and enter the runtime analysis below.

\begin{proposition}[Bounded-moment cumulant perturbation]
\label{prop:supp-bounded-moment-cumulant-perturbation}
Let $M$ and $\widehat M$ be two moment families indexed by the nonempty
subsets of $[q]$, and suppose
\begin{equation}
|M(U)|,|\widehat M(U)|\le1,
\qquad
\varnothing\ne U\subseteq[q].
\end{equation}
If, for some $\tau_\mu>0$, uniformly over these subsets,
\begin{equation}
|\widehat M(U)-M(U)|\le\tau_\mu,
\end{equation}
then the corresponding cumulants obey
\begin{equation}
|\widehat\kappa-\kappa|
\le
\Gamma_q\tau_\mu.
\label{eq:supp-bounded-moment-cumulant-perturbation}
\end{equation}
\end{proposition}

\begin{proof}
For each $\mathfrak p\in\Pi([q])$, telescope its product over the
$|\mathfrak p|$ blocks.  Since every other factor has magnitude at most one,
the product changes by at most $|\mathfrak p|\tau_\mu$.  Summing with the
cumulant coefficients gives
\begin{equation}
|\widehat\kappa-\kappa|
\le
\sum_{\mathfrak p\in\Pi([q])}
(|\mathfrak p|-1)!\,|\mathfrak p|\,\tau_\mu
=
\Gamma_q\tau_\mu,
\end{equation}
as claimed.
\end{proof}

Both the ordered numbers $\Gamma_q$ and ordinary Bell numbers
$B_q:=|\Pi([q])|$ are increasing, so their maxima through the grouping order
occur at $\ell_{\mathrm{grp}}$.
The ordered-Bell exponential generating function is $(2-e^x)^{-1}$
\cite{doubilet1972foundations}, and hence $\Gamma_q
\sim
\frac{q!}{2(\log 2)^{q+1}}$ holds.

\subsection{Why pairwise cumulants are insufficient}
\label{sec:supp-pairwise-cumulants-insufficient}
The following recovered-sector example shows why order two does not suffice
\cite{smith2020multivariateKstatistics}.  Consider the $3$-qubit peeled
state

\begin{align}
\widetilde\rho_{\mathrm{ex}}
=
\frac{1}{32}
\Bigl(
&13\,|000\rangle\!\langle 000|
\notag\\
&+
5\!\!\sum_{x\in\{001,010,100\}}\!\!|x\rangle\!\langle x|
\notag\\
&+
\!\!\sum_{x\in\{011,101,110,111\}}\!\!|x\rangle\!\langle x|
\Bigr).
\label{eq:supp-example-state-three-sectors}
\end{align}
Here $t=0$, every Pauli containing $X$ or $Y$ has score zero, and
$s(Z_i)=1/4$, $s(Z_iZ_j)=1/16$, while $s(Z_1Z_2Z_3)=0$.  Thus for
$\theta\in(0,1/4)$ rank-guided recovery inserts the independent Pauli directions and
skips any surviving two-body products already in their span, giving
$\mathfrak G=\{\{Z_1\},\{Z_2\},\{Z_3\}\}$.

Relabel the Pauli directions as $P_\alpha,P_\beta,P_\gamma$.  The complete moments used
below are
\begin{equation}
M(P_u)=\frac12,
\qquad
M(P_u,P_v)=\frac14\quad(u\ne v),
\qquad
M(P_\alpha,P_\beta,P_\gamma)=0,
\end{equation}
and hence
\begin{align*}
\kappa(P_u,P_v)
&=\frac14-\frac12\cdot\frac12=0,
\qquad u\ne v,\\
\kappa(P_\alpha,P_\beta,P_\gamma)
&=0-3\cdot\frac14\cdot\frac12
+2\left(\frac12\right)^3=-\frac18\ne0.
\end{align*}
All three sectors belong to one true block.  Across any split
$\{u\}\sqcup\{v,w\}$, the allowed witness $(P_u,P_vP_w)$ likewise has
cumulant $-1/8$.  The block is therefore $\eta_{\mathrm{irr}}$-irreducible
for $0\le\eta_{\mathrm{irr}}<1/8$ when $\ell_{\mathrm{grp}}\ge3$, yet its
singleton clusters have no pairwise edge: order three merges them, whereas
order two cannot start.  We next transfer such full-lift witnesses through
peeling.

\subsection{Recovered sectors, full lifts, and the peeling floor}
\label{sec:supp-recovered-blocks-full-lifts-peeling-floor}

Recall the peeled hidden-block Pauli groups $\mathcal R_a$ and the residual
quotient $\pi$ from the definitions preceding
Proposition~\ref{prop:supp-blockwise-peeling-preserves-residual-blocks}.
Block-purity gives
\begin{equation}
\bigcup_{\alpha\in\mathcal B_a^\star}\mathcal G_\alpha
\subseteq
\pi(\mathcal R_a).
\end{equation}
Because this is a quotient statement, fix a hidden-block lift
$Z^{c(\widehat R)}\otimes\widehat R\in\mathcal R_{a(\alpha)}$ for every
recovered generator $\widehat R\in\mathcal G_\alpha$.  For
$A\subseteq\mathcal B_a^\star$ and $O\in\mathcal P(A)$, multiply the lifts
in any fixed generator expression for $O$ and track the Hermitian Pauli sign.
This gives a permitted full lift
\begin{equation}
O^\uparrow
=
Z^{c(O)}\otimes O
\in
\mathcal R_a,
\qquad
c(O)\in\mathbb F_2^t,
\label{eq:supp-generated-full-lift}
\end{equation}
where membership is phase-free and the equality includes the tracked sign.
The lift need not be canonical: another expression may change its prefix and
sign correction, but not its residual representative, and all bounds below
are uniform over these choices.  For several true blocks, construct and
multiply the lifts blockwise.

Distinct sectors commute elementwise, so tuples
$O_j\in\mathcal P(C_j)$ from disjoint sector sets commute.  Their permitted
full lifts also commute whenever each $C_j$ lies in one true block.

\begin{proposition}[Full-lift moment factorization across hidden blocks]
\label{prop:supp-full-lift-moment-factorization}
Let $C_1,\dots,C_q$ be disjoint sector sets such that
$C_j\subseteq\mathcal B_{a_j}^\star$, choose
$O_j\in\mathcal P(C_j)$, and let
$O_j^\uparrow=Z^{c_j}\otimes O_j\in\mathcal R_{a_j}$ be Hermitian full lifts.
For every $U\subseteq[q]$, the full moment factors after grouping the tuple
members by their hidden-block label:
\begin{equation}
\Tr\!\left(
\widetilde\rho\prod_{j\in U}O_j^\uparrow
\right)
=
\prod_{a\in\{a_j:j\in U\}}
\Tr\!\left(
\widetilde\rho
\prod_{\substack{j\in U\\a_j=a}}O_j^\uparrow
\right).
\label{eq:supp-true-block-moment-factorization}
\end{equation}
The empty-set moment is one on both sides.
\end{proposition}

\begin{proof}
Conjugating $O_j^\uparrow\in\mathcal R_{a_j}$ back through
$U_{\mathrm{stab}}^\dagger U_{\mathrm c}$ gives a Pauli on $B_{a_j}$.
The latent state $\bigotimes_a\rho_a$ therefore factors the expectation of
each subset product by its distinct labels, including all tracked signs.
Conjugation invariance gives
Eq.~\eqref{eq:supp-true-block-moment-factorization}.
\end{proof}

The algorithm drops the $Z$ prefixes on which exact factorization holds.
The next lemma bounds that loss uniformly.

\begin{lemma}[Signed prefix stability]
\label{lem:supp-signed-prefix-stability}
Let $\sigma=|b\rangle\!\langle b|\otimes\rho_{\mathrm{res}}$ be the
product state from Proposition~\ref{prop:supp-empirical-peeling}, so that
\begin{equation}
\|\widetilde\rho-\sigma\|_1\le 2\sqrt{\varepsilon_{\mathrm{peel}}}.
\end{equation}
For every $c\in\mathbb F_2^t$ and every signed Hermitian residual Pauli
$P=\pm P_0$, $P_0\in\Pc_m$,
\begin{equation}
\left|
\Tr\!\left[\widetilde\rho\,(Z^c\otimes P)\right]
-
(-1)^{b\cdot c}
\mu_{\mathrm{res}}(P)
\right|
\le
4\sqrt{\varepsilon_{\mathrm{peel}}}.
\label{eq:supp-signed-prefix-stability}
\end{equation}
Moreover, let $O_1,\dots,O_q\in\Pc_m$ be pairwise commuting residual
Paulis with full lifts
\begin{equation}
O_j^\uparrow=Z^{c_j}\otimes O_j,
\qquad j=1,\dots,q,
\end{equation}
and define $\kappa_{\mathrm{full}}(O_1^\uparrow,\dots,O_q^\uparrow)$ from the full signed
moments $\Tr(\widetilde\rho\prod_{j\in U}O_j^\uparrow)$.  Set
\begin{equation}
\varsigma_{\mathbf c}
:=
\prod_{j=1}^q(-1)^{b\cdot c_j}.
\end{equation}
Then
\begin{equation}
\begin{aligned}
\left|
\kappa_{\mathrm{full}}(O_1^\uparrow,\dots,O_q^\uparrow)
-
\varsigma_{\mathbf c}
\kappa_{\mathrm{res}}(O_1,\dots,O_q)
\right|
&\le
4\Gamma_q\sqrt{\varepsilon_{\mathrm{peel}}}.
\end{aligned}
\label{eq:supp-cumulant-prefix-stability}
\end{equation}
\end{lemma}

\begin{proof}
For $P^{(c)}:=Z^c\otimes P$ and
$P^{(0)}:=I^{\otimes t}\otimes P$, their expectations in $\sigma$ differ
by the sign $(-1)^{b\cdot c}$.  Hence
\begin{align*}
&
\left|
\Tr(\widetilde\rho P^{(c)})
-
(-1)^{b\cdot c}\Tr(\widetilde\rho P^{(0)})
\right| \\
&\qquad\le
\left|\Tr((\widetilde\rho-\sigma)P^{(c)})\right|
+
\left|\Tr((\widetilde\rho-\sigma)P^{(0)})\right|
\le
4\sqrt{\varepsilon_{\mathrm{peel}}},
\end{align*}
which proves Eq.~\eqref{eq:supp-signed-prefix-stability}.

Apply the same bound to every nonempty subset product.  Its prefix sign is
$\prod_{j\in U}(-1)^{b\cdot c_j}$; multilinearity therefore makes the
cumulant of the sign-twisted residual moments
$\varsigma_{\mathbf c}\kappa_{\mathrm{res}}(O_1,\dots,O_q)$.
All moments have magnitude at most one, so
Proposition~\ref{prop:supp-bounded-moment-cumulant-perturbation} with
tolerance $4\sqrt{\varepsilon_{\mathrm{peel}}}$ proves
Eq.~\eqref{eq:supp-cumulant-prefix-stability}.
\end{proof}

Set
\begin{equation}
\beta_{\mathrm{peel}}
:=
4\Gamma_{\ell_{\mathrm{grp}}}
\sqrt{\varepsilon_{\mathrm{peel}}}.
\label{eq:supp-beta-peel-cumulants}
\end{equation}
This is both the cross-block residual-cumulant bound and the detection floor
created by peeling leakage.

\begin{corollary}[Cross-block residual cumulant floor]
\label{cor:supp-cross-block-cumulant-floor}
Let $2\le q\le\ell_{\mathrm{grp}}$.  Suppose
$C_1,\dots,C_q\subseteq[L]$ are pairwise disjoint sector-index sets, each
contained inside a true recovered-sector block, but
$C_1\cup\cdots\cup C_q$ is not contained inside a single true
recovered-sector block.  Then for every
\begin{equation}
O_j\in\mathcal P(C_j),
\qquad j=1,\dots,q,
\end{equation}
we have
\begin{equation}
\left|
\kappa_{\mathrm{res}}(O_1,\dots,O_q)
\right|
\le
4\Gamma_q\sqrt{\varepsilon_{\mathrm{peel}}}
\le
\beta_{\mathrm{peel}}.
\label{eq:supp-cross-block-cumulant-floor}
\end{equation}
If $\varepsilon_{\mathrm{peel}}=0$, the residual cross-block cumulant
vanishes exactly.
\end{corollary}

\begin{proof}
Choose permitted lifts from
Eq.~\eqref{eq:supp-generated-full-lift}.  Their subset moments factor across
a nontrivial hidden-block split by
Proposition~\ref{prop:supp-full-lift-moment-factorization}; hence
Proposition~\ref{prop:supp-nonzero-cumulant-certifies-no-product-split}
implies
\begin{equation}
\kappa_{\mathrm{full}}(O_1^\uparrow,\dots,O_q^\uparrow)=0.
\end{equation}
Now apply Eq.~\eqref{eq:supp-cumulant-prefix-stability}.
\end{proof}

This intrinsic bound will rule out false merges once the algorithmic
partition and its true-block containment invariant are introduced.

\subsection{Residual irreducibility and hierarchical grouping}
\label{sec:supp-hierarchical-cumulant-grouping-algorithm}

Block-purity says only that a sector belongs to one hidden-block algebra.  It
does not imply that the dependence visible in the recovered algebras is
irreducible; nor does purity of a latent physical block preclude a product
split relative to the particular recovered subalgebras.  Exact grouping
therefore needs the following intrinsic post-peeling identifiability
condition.

\begin{definition}[Residual $\eta_{\mathrm{irr}}$-irreducibility]
\label{def:supp-eta-irreducible-hidden-blocks}
For $\eta_{\mathrm{irr}}\ge0$, the true recovered-sector partition is
residually $\eta_{\mathrm{irr}}$-irreducible if, for every
$a\in\mathcal A^\star$ and every nontrivial partition
\begin{equation}
\mathcal B_a^\star
=
\mathcal S_1\sqcup\cdots\sqcup\mathcal S_r,
\qquad r\ge 2,
\end{equation}
there exist an order
\begin{equation}
2\le q\le\min\{r,\ell_{\mathrm{grp}}\},
\end{equation}
distinct indices
$i_1,\dots,i_q\in [r]$,
and generated Pauli operators
\begin{equation}
O_\ell
\in
\mathcal P(\mathcal S_{i_\ell})\setminus\{I\},
\qquad
\ell=1,\dots,q,
\end{equation}
such that
\begin{equation}
\left|
\kappa_{\mathrm{res}}(O_1,\dots,O_q)
\right|
>
\eta_{\mathrm{irr}} .
\label{eq:supp-emp-prop-eta-irreducible}
\end{equation}
\end{definition}

Let $\mathfrak C$ denote the current partition into clusters.  Its active
clusters are
\begin{equation}
\mathfrak C_{<\ell_{\mathrm{grp}}}
:=
\{C\in\mathfrak C: |C|<\ell_{\mathrm{grp}}\}.
\label{eq:supp-active-cluster-set}
\end{equation}
Whenever the true-block containment invariant holds, these current clusters
are pairwise disjoint sector-index sets of the kind used in
Proposition~\ref{prop:supp-full-lift-moment-factorization} and
Corollary~\ref{cor:supp-cross-block-cumulant-floor}.

A \emph{cumulant interface} $\kappa^\sharp$ assigns a real value
$\kappa^\sharp(O_1,\dots,O_q)$ to every admissible labeled commuting tuple.
For $q\ge 2$, the interface-dependent $q$-uniform test hypergraph
$H_q(\mathfrak C;\kappa^\sharp)$ has vertex set
$\mathfrak C_{<\ell_{\mathrm{grp}}}$ and contains a hyperedge
$\{C_1,\dots,C_q\}$ when the $C_j$'s are distinct,
$|C_1\cup\cdots\cup C_q|\le\ell_{\mathrm{grp}}$, and there exist
$O_j\in\mathcal P(C_j)\setminus\{I\}$ such that
\begin{equation}
\left|
\kappa^\sharp(O_1,\dots,O_q)
\right|
>
\eta_{\mathrm{test}} .
\end{equation}
Identity operators are excluded from tuple searches.  The selected operators
commute because the selected clusters are pairwise disjoint and every
recovered sector in one cluster commutes elementwise with every recovered
sector in another; therefore the groups generated by different selected
clusters commute.  Their cumulants are computed from signed Pauli moments
using Eqs.~\eqref{eq:supp-signed-pauli-moment-general}
and~\eqref{eq:supp-general-mixed-cumulant}.

\begin{inlinealgorithm}{alg:supp-hierarchical-cumulant-grouping}{Hierarchical Cumulant Grouping of Recovered Sectors}
\begin{algorithmic}[1]
\Require
Recovered sectors $\mathfrak G=\{\mathcal G_1,\dots,\mathcal G_L\}$,
a maximum grouping order $\ell_{\mathrm{grp}}\ge1$ satisfying
Eq.~\eqref{eq:supp-grouping-order-promise}, a selected cumulant interface
$\kappa^\sharp$, and a test threshold $\eta_{\mathrm{test}}\ge0$.

\Ensure
A partition of the sectors into estimated recovered-sector blocks.

\State Initialize the current partition $\mathfrak C\gets \{\{1\},\dots,\{L\}\}$.

\State Set $q\gets 2$.

\While{$q\le\ell_{\mathrm{grp}}$}
    \State Build $H_q(\mathfrak C;\kappa^\sharp)$ using the hyperedge rule above.

    \State Let $\mathcal D_1,\dots,\mathcal D_r$ be the nontrivial connected components of $H_q(\mathfrak C;\kappa^\sharp)$, where nontrivial means that the component contains at least one hyperedge.

    \If{$r=0$}
        \State Set $q\gets q+1$.
    \Else
        \For{$j=1,\dots,r$}
            \State Merge all clusters in $\mathcal D_j$ into
            $C_j^{\mathrm{new}}:=\bigcup_{C\in\mathcal D_j} C$.
        \EndFor

        \State Replace every cluster contained in some $\mathcal D_j$ by the corresponding merged cluster $C_j^{\mathrm{new}}$.

        \State Reset $q\gets 2$.
    \EndIf
\EndWhile

\State \Return $\mathfrak C$.
\end{algorithmic}
\end{inlinealgorithm}

The exact-data specialization selects
$\kappa^\sharp=\kappa_{\mathrm{res}}$ and
$\eta_{\mathrm{test}}=0$.  The empirical implementation selects
$\kappa^\sharp=\widehat\kappa_{\mathrm{res}}$ at its calibrated positive
threshold.  Thus both implementations use the same algorithmic map while
constructing formally distinct interface-dependent hypergraphs.

Hyperedges that overlap are closed under connected components before any
merge, so a chain of witnessed dependencies becomes one new cluster.  This
merge enlarges its generated group from the separate $\mathcal P(C)$'s to
the group generated by their union.  Products newly available in that larger
group can expose lower-order witnesses that were absent before the merge.
Accordingly, every merge resets $q$ to two.  The algorithm terminates only
after one complete scan from order two through $\ell_{\mathrm{grp}}$ finds no
hyperedge.

\subsection{Correctness under uniform cumulant accuracy}
\label{sec:supp-correctness-empirical-complexity}

Let $\widehat\kappa_{\mathrm{res}}$ denote the empirical cumulant interface.
Run Algorithm~\ref{alg:supp-hierarchical-cumulant-grouping} with
$\kappa^\sharp=\widehat\kappa_{\mathrm{res}}$.  The queried
cumulants satisfy the uniform accuracy event at tolerance
$\tau_\kappa>0$ when
\begin{equation}
\left|
\widehat\kappa_{\mathrm{res}}(O_1,\dots,O_q)
-
\kappa_{\mathrm{res}}(O_1,\dots,O_q)
\right|
\le
\tau_\kappa
\label{eq:supp-emp-cumulant-uniform-event}
\end{equation}
for every admissible tuple queried while constructing the empirical
hypergraphs.  This deterministic event is the only estimator property used
in the correctness argument.

\begin{lemma}[No false merges under uniform cumulant accuracy]
\label{lem:supp-cumulant-no-false-merge}
Suppose Eq.~\eqref{eq:supp-emp-cumulant-uniform-event} holds and
\begin{equation}
\beta_{\mathrm{peel}}+\tau_\kappa<\eta_{\mathrm{test}}.
\end{equation}
Then every cluster produced by
Algorithm~\ref{alg:supp-hierarchical-cumulant-grouping} is contained in one
true recovered-sector block; in particular, its output refines the true
partition.
\end{lemma}

\begin{proof}
This containment holds for the initial singleton clusters.  If it holds
before a scan, then for clusters $C_1,\dots,C_q$ not all in one true block,
Corollary~\ref{cor:supp-cross-block-cumulant-floor} and the uniform event give
\begin{equation}
\left|\widehat\kappa_{\mathrm{res}}(O_1,\dots,O_q)\right|
\le\beta_{\mathrm{peel}}+\tau_\kappa<\eta_{\mathrm{test}}
\end{equation}
for every admissible candidate $O_j\in\mathcal P(C_j)\setminus\{I\}$.
No empirical hyperedge therefore crosses true blocks, and neither can a
connected component of such edges.  Every merge preserves containment.
\end{proof}

\begin{proposition}[Correctness under a uniform cumulant event]
\label{prop:supp-hierarchical-cumulant-correctness}
Assume that the true recovered-sector partition is residually
$\eta_{\mathrm{irr}}$-irreducible.  Suppose the empirical cumulant estimates
queried by Algorithm~\ref{alg:supp-hierarchical-cumulant-grouping} satisfy
Eq.~\eqref{eq:supp-emp-cumulant-uniform-event}, and choose the test threshold so
that
\begin{equation}
\beta_{\mathrm{peel}}+\tau_\kappa
<
\eta_{\mathrm{test}}
<
\eta_{\mathrm{irr}}-\tau_\kappa .
\label{eq:supp-hidden-grouping-threshold-window}
\end{equation}
Such a threshold exists exactly when
\begin{equation}
\beta_{\mathrm{peel}}+2\tau_\kappa
<
\eta_{\mathrm{irr}}.
\end{equation}
Then Algorithm~\ref{alg:supp-hierarchical-cumulant-grouping}, run with
$\kappa^\sharp=\widehat\kappa_{\mathrm{res}}$, outputs the true
recovered-sector partition:
\begin{equation}
\widehat{\mathfrak C}
=
\{\mathcal B_a^\star:a\in\mathcal A^\star\}.
\end{equation}
\end{proposition}

\begin{proof}
Lemma~\ref{lem:supp-cumulant-no-false-merge} supplies the containment
invariant.  For progress, suppose some $a\in\mathcal A^\star$ remains split
into its current clusters,
\begin{equation}
\mathcal B_a^\star
=
\mathcal S_1\sqcup\cdots\sqcup\mathcal S_r,
\qquad r\ge 2.
\end{equation}
Every $\mathcal S_i$ is active because it is a proper subset of
$\mathcal B_a^\star$ and
$|\mathcal B_a^\star|\le\ell_{\mathrm{grp}}$.  Residual
$\eta_{\mathrm{irr}}$-irreducibility supplies
\begin{equation}
2\le q\le\min\{r,\ell_{\mathrm{grp}}\},
\qquad
i_1,\dots,i_q\in[r],
\end{equation}
and operators
\begin{equation}
O_\ell
\in
\mathcal P(\mathcal S_{i_\ell})\setminus\{I\},
\qquad
\ell=1,\dots,q,
\end{equation}
such that
\begin{equation}
\left|
\kappa_{\mathrm{res}}(O_1,\dots,O_q)
\right|
>
\eta_{\mathrm{irr}} .
\end{equation}
$|\mathcal S_{i_1}\cup\cdots\cup\mathcal S_{i_q}|\le|\mathcal B_a^\star|
\le\ell_{\mathrm{grp}}$, so this nonidentity operator tuple is an admissible
hyperedge candidate and is queried when its order is scanned.
By Eqs.~\eqref{eq:supp-emp-cumulant-uniform-event}
and~\eqref{eq:supp-hidden-grouping-threshold-window},
\begin{equation}
\left|
\widehat\kappa_{\mathrm{res}}(O_1,\dots,O_q)
\right|
>
\eta_{\mathrm{irr}}-\tau_\kappa
>
\eta_{\mathrm{test}} .
\end{equation}
Thus it is a hyperedge of
$H_q(\mathfrak C;\widehat\kappa_{\mathrm{res}})$ and merges at least two
current pieces.  Each merge lowers the cluster count, so the algorithm
terminates after at most $L-1$ merge rounds and one final complete scan.  A
remaining split would supply the displayed witness in that scan, a
contradiction.  Together with containment this gives
\begin{equation}
\widehat{\mathfrak C}
=
\{\mathcal B_a^\star:a\in\mathcal A^\star\}.
\end{equation}
\end{proof}

At the exact-data endpoint, if $\varepsilon_{\mathrm{peel}}=0$ and the true
partition is residually $0$-irreducible, run the same algorithm with
$\kappa^\sharp=\kappa_{\mathrm{res}}$ and $\eta_{\mathrm{test}}=0$.
Then $\beta_{\mathrm{peel}}=0$, cross-block cumulants vanish, and every
remaining within-block split has a strictly nonzero witness, so the preceding
no-false-merge and progress arguments recover the true partition.  This is
the endpoint of that proof, not a substitution into the strict window
\eqref{eq:supp-hidden-grouping-threshold-window}.

The remaining task is estimator-specific: construct an empirical cumulant
interface satisfying Eq.~\eqref{eq:supp-emp-cumulant-uniform-event}, then
count the copies and classical work required under each permitted measurement
model.

\subsection{Empirical cumulant estimation and copy complexity}

For the empirical bounds in this subsection, let $\tau_\kappa>0$ and
$\delta_{\mathrm{tuple}},
\delta_{\mathrm{grp}}^{\mathrm{ordinary}}\in(0,1)$.
Condition on the completed peeling and recovery transcript.  For every
admissible tuple, each signed subset product is a known sign times a Pauli
in the group generated by the sectors in that subset.  The ordinary
baseline estimates these adaptive tuples on separate fresh batches after
$U_{\mathrm{stab}}^\dagger$, or equivalently measures the pulled-back
Paulis on $\rho$.

\subsubsection{Ordinary adaptive cumulant estimation}

We begin with a fixed commuting tuple and use the same joint-measurement
record to estimate all of its signed subset moments.

\begin{proposition}[Ordinary estimator for one commuting cumulant]
\label{prop:supp-commuting-cumulant-sample-count}
Fix a labeled commuting tuple of $q$ Hermitian Pauli observables and a target
cumulant error $\tau_\kappa>0$.  Jointly measuring that tuple on
$N_{\mathrm{cum}}$ fresh copies and forming every nonempty empirical subset
moment gives
\begin{equation}
|\widehat\kappa-\kappa|\le\tau_\kappa
\end{equation}
with probability at least $1-\delta_{\mathrm{tuple}}$ whenever
\begin{equation}
N_{\mathrm{cum}}
\ge
\frac{2\Gamma_q^2}{\tau_\kappa^2}
\log\!\left(
\frac{2(2^q-1)}{\delta_{\mathrm{tuple}}}
\right).
\label{eq:supp-fixed-tuple-ordinary-copies}
\end{equation}
In particular,
\begin{equation}
N_{\mathrm{cum}}
=
O\!\left(
\frac{\Gamma_q^2}{\tau_\kappa^2}
\log\frac{2^q}{\delta_{\mathrm{tuple}}}
\right).
\end{equation}
\end{proposition}

\begin{proof}
One joint outcome $\mathsf X\in\{\pm1\}^q$ supplies the unbiased subset
estimators $\prod_{j\in U}\mathsf X_j$.  Hoeffding's inequality
\cite{Hoeffding1963} and a union bound over the $2^q-1$ nonempty subsets give
\begin{equation}
|\widehat M(U)-M(U)|
\le
\frac{\tau_\kappa}{\Gamma_q}
\qquad
\text{for all }\varnothing\ne U\subseteq[q]
\end{equation}
under Eq.~\eqref{eq:supp-fixed-tuple-ordinary-copies}.  Since the moments lie
in $[-1,1]$,
Proposition~\ref{prop:supp-bounded-moment-cumulant-perturbation} gives the
claimed cumulant error.
\end{proof}

There are at most $L$ adaptive phases.  An order-$q$ phase examines at most
$\binom Lq$ cluster collections, each with at most
$4^{|C_1\cup\cdots\cup C_q|}\le4^{\ell_{\mathrm{grp}}}$ operator tuples.
Thus the realized query count satisfies
\begin{align}
N_{\mathrm{test}}
&\le N_{\mathrm{test}}^{\max}
:=
L4^{\ell_{\mathrm{grp}}}
\sum_{q=2}^{\min\{\ell_{\mathrm{grp}},L\}}\binom{L}{q},
\label{eq:supp-number-adaptive-cumulant-tests}\\
N_{\mathrm{test}}^{\max}
&\le
\ell_{\mathrm{grp}}n^{\ell_{\mathrm{grp}}+1}4^{\ell_{\mathrm{grp}}},
\label{eq:supp-number-adaptive-cumulant-tests-simplified}
\end{align}
using at most $\ell_{\mathrm{grp}}$ summands and
$\binom Lq\le L^q\le n^{\ell_{\mathrm{grp}}}$.

For each adaptive query, condition on the preceding history and use a fresh
batch, so Proposition~\ref{prop:supp-commuting-cumulant-sample-count}
applies to the now-fixed tuple.  Allocate
$\delta_{\mathrm{tuple}}:=
\delta_{\mathrm{grp}}^{\mathrm{ordinary}}/N_{\mathrm{test}}^{\max}$.
A conditional union bound over at most
$N_{\mathrm{test}}^{\max}$ realized tests is therefore valid.  The
sufficient total is
\begin{equation}
N_{\mathrm{copy}}^{\mathrm{ordinary}}
\le
N_{\mathrm{test}}^{\max}
\left\lceil
\frac{2\Gamma_{\ell_{\mathrm{grp}}}^2}{\tau_\kappa^2}
\log\!\left(
\frac{
2(2^{\ell_{\mathrm{grp}}}-1)
N_{\mathrm{test}}^{\max}
}{\delta_{\mathrm{grp}}^{\mathrm{ordinary}}}
\right)
\right\rceil.
\end{equation}
Consequently, the generic asymptotic ordinary copy bound is
\begin{equation}
N_{\mathrm{copy}}^{\mathrm{ordinary}}
=
O\!\left(
N_{\mathrm{test}}^{\max}
\frac{\Gamma_{\ell_{\mathrm{grp}}}^2}{\tau_\kappa^2}
\log\frac{
2^{\ell_{\mathrm{grp}}}N_{\mathrm{test}}^{\max}
}{\delta_{\mathrm{grp}}^{\mathrm{ordinary}}}
\right).
\label{eq:supp-ordinary-grouping-generic-copies}
\end{equation}

\paragraph{Coherent simultaneous-estimation comparison.}
Under the stronger coherent-memory access model of
Ref.~\cite{HuangKuengPreskill2021}, one may instead estimate, before
grouping, every signed moment generated by at most
$\ell_{\mathrm{grp}}$ recovered sectors, with failure budget
$\delta_{\mathrm{grp}}^{\mathrm{coherent}}\in(0,1)$.  With
$\ell_{\mathrm{eff}}:=\min\{\ell_{\mathrm{grp}},L\}$, the family has size at
most
$\sum_{\nu=0}^{\ell_{\mathrm{eff}}}\binom L\nu4^\nu$; for
$\ell_{\mathrm{eff}}\ge1$ this is at most
$(4eL/\ell_{\mathrm{eff}})^{\ell_{\mathrm{eff}}}$, while for
$\ell_{\mathrm{eff}}=0$ there is no nontrivial moment to estimate.
For $\ell_{\mathrm{eff}}\ge1$, taking moment tolerance
$\tau_\kappa/\Gamma_{\ell_{\mathrm{grp}}}$, clipping to $[-1,1]$, and using
the known subset-product signs gives the same uniform cumulant event, with
\begin{equation}
N_{\mathrm{copy}}^{\mathrm{coherent}}
=
O\!\left(
\frac{\Gamma_{\ell_{\mathrm{grp}}}^4}{\tau_\kappa^4}
\log\!\left[
\frac{1}{\delta_{\mathrm{grp}}^{\mathrm{coherent}}}
\sum_{\nu=0}^{\ell_{\mathrm{eff}}}\binom L\nu4^\nu
\right]
\right).
\label{eq:supp-coherent-grouping-generic-copies}
\end{equation}
All correctness arguments are unchanged.  This alternative assumes coherent
storage and collective access; it implies neither a smaller classical runtime
nor a coherent-circuit bound.

\subsubsection{Classical runtime}

Grouping enumerates
$O(\ell_{\mathrm{grp}}n^{\ell_{\mathrm{grp}}+1}4^{\ell_{\mathrm{grp}}})$
operator tuples.  For order $q$, canonical subset products cost $O(n2^q)$
and the $B_q=|\Pi([q])|$ partition terms cost $O(qB_q)$
\cite{smith2020multivariateKstatistics}.  Since $q\le n$,
$B_q\ge2^{q-1}$, $B_q\le B_{\ell_{\mathrm{grp}}}$, and
$B_{\ell_{\mathrm{grp}}}\le
\ell_{\mathrm{grp}}^{\ell_{\mathrm{grp}}}$,
\begin{equation}
T_{\mathrm{class}}
=
O\!\left(
\ell_{\mathrm{grp}}n^{\ell_{\mathrm{grp}}+2}
4^{\ell_{\mathrm{grp}}}B_{\ell_{\mathrm{grp}}}
\right)
=
O\!\left(
\ell_{\mathrm{grp}}n^{\ell_{\mathrm{grp}}+2}
(4\ell_{\mathrm{grp}})^{\ell_{\mathrm{grp}}}
\right).
\end{equation}
For $\ell_{\mathrm{grp}}=d$, this gives
\begin{equation}
T_{\mathrm{class}}
=
O\!\left(
d\,n^{d+2}(4d)^d
\right).
\label{eq:supp-runtime-emp-cumulant-asymptotic}
\end{equation}
This counts tuple/partition enumeration, Pauli arithmetic, and cumulant
evaluation, but not measurement-record aggregation.

\subsection{Grouping with a guessed scale}
\label{sec:supp-guessed-scale-overrefinement}

When $\eta_{\mathrm{irr}}$ is unknown, grouping can instead be calibrated at
a guessed scale $\eta_s>0$.  The resulting guarantee is deliberately
one-sided: it prevents cross-block merges and certifies weak remaining
dependence, but it need not recover the exact partition.

\begin{proposition}[Certified over-refined grouping at a guessed scale]
\label{prop:supp-guessed-scale-overrefined-grouping}
Let $\eta_s>0$.  Run
Algorithm~\ref{alg:supp-hierarchical-cumulant-grouping} with
$\kappa^\sharp=\widehat\kappa_{\mathrm{res}}$ and threshold
\begin{equation}
\eta_{\mathrm{test}}:=\frac{\eta_s}{2},
\qquad
\tau_\kappa:=\frac{\eta_s}{4},
\end{equation}
and suppose the queried cumulants satisfy the uniform event
\eqref{eq:supp-emp-cumulant-uniform-event}.  If
$\beta_{\mathrm{peel}}<\frac{\eta_s}{4}$, then no merge crosses the true
recovered-sector partition and the output
is a potentially non-exact refinement
\begin{equation}
\widehat{\mathfrak C}
=
\{C_{a,u}:a\in\mathcal A^\star,\ u=1,\dots,r_a\},
\qquad
\mathcal B_a^\star
=
\bigsqcup_{u=1}^{r_a}C_{a,u}.
\label{eq:supp-guessed-scale-refined-partition}
\end{equation}
Moreover, for every $a\in\mathcal A^\star$, every
$T\subseteq[r_a]$ with $|T|\ge2$, and every choice
\begin{equation}
O_u\in\mathcal P(C_{a,u})\setminus\{I\},
\qquad u\in T,
\end{equation}
one has
\begin{equation}
\left|
\kappa_{\mathrm{res}}(O_u:u\in T)
\right|
\le
\xi_s,
\qquad
\xi_s:=\frac{3\eta_s}{4}.
\label{eq:supp-guessed-scale-small-final-cumulants}
\end{equation}
\end{proposition}

\begin{proof}
Here
$\beta_{\mathrm{peel}}+\tau_\kappa<\eta_s/2=\eta_{\mathrm{test}}$, so
Lemma~\ref{lem:supp-cumulant-no-false-merge} proves the refinement claim.
For the terminal certificate, fix
$a\in\mathcal A^\star$ and $T\subseteq[r_a]$ with
$|T|=q\ge2$.  Since
$\bigsqcup_{u\in T}C_{a,u}\subseteq\mathcal B_a^\star$ and
$|\mathcal B_a^\star|\le\ell_{\mathrm{grp}}$, we have
$q\le\ell_{\mathrm{grp}}$ and
$|\bigcup_{u\in T}C_{a,u}|\le\ell_{\mathrm{grp}}$.  If $r_a>1$, every
$C_{a,u}$ is a proper subset of
$\mathcal B_a^\star$ and hence active.  Thus the collection is tested during the final
complete scan.  Because that scan has no hyperedge at order
$q\le\ell_{\mathrm{grp}}$, every corresponding empirical cumulant satisfies
\begin{equation}
\left|
\widehat\kappa_{\mathrm{res}}(O_u:u\in T)
\right|
\le
\eta_{\mathrm{test}}
=
\frac{\eta_s}{2}.
\end{equation}
The uniform event then gives
\begin{equation}
\left|
\kappa_{\mathrm{res}}(O_u:u\in T)
\right|
\le
\frac{\eta_s}{2}
+
\frac{\eta_s}{4}
=
\frac{3\eta_s}{4},
\end{equation}
which proves Eq.~\eqref{eq:supp-guessed-scale-small-final-cumulants}.
\end{proof}

The ordinary or coherent construction supplies the required event by
substituting $\tau_\kappa=\eta_s/4$ into, respectively,
Eqs.~\eqref{eq:supp-ordinary-grouping-generic-copies}
or~\eqref{eq:supp-coherent-grouping-generic-copies}; these are alternative
access models for the same certified refinement.

\section{Clifford Localization and Block-Product Approximation}
\label{sec:supp-block-normalization}

\subsection{Interface from recovery and grouping}
\label{sec:supp-localization-interface}

Under the ranking margins
\eqref{eq:supp-visible-threshold-condition}
and~\eqref{eq:supp-block-ranking-gap-condition}, and the applicable grouping
hypotheses, let $\mathcal E_{\mathrm{grp}}$ be the uniform cumulant event
\eqref{eq:supp-emp-cumulant-uniform-event} (or the sure event when grouping
is skipped), and condition on
$\mathcal E_{\mathrm{str}}:=\mathcal E_{\mathrm{BS}}\cap
\mathcal E_{\mathrm{rank}}\cap\mathcal E_{\mathrm{grp}}$.
The visible transcript contains $t$, $m=n-t$, the generators and
$U_{\mathrm{stab}}$, the recovered sectors and residual representatives,
and the computable empirical partition $\widehat{\mathfrak C}$.

By Proposition~\ref{prop:supp-hierarchical-cumulant-correctness} or
Proposition~\ref{prop:supp-guessed-scale-overrefined-grouping}, no empirical
group crosses a true recovered-sector block.  Thus only in the proof may we
write
\begin{equation}
\widehat{\mathfrak C}
=
\{C_{a,h}:a\in\mathcal A^\star,\ h=1,\dots,r_a\},
\qquad
\mathcal B_a^\star
=
\bigsqcup_{h=1}^{r_a}C_{a,h}.
\label{eq:supp-localization-hidden-refinement}
\end{equation}
The labels $a$, map $C\mapsto a$, hidden aggregates, comparison states,
$U_{\mathrm c}$, and hidden partition are proof-only; only the visible
transcript computes the registers $J_C$ and structural certificate.

We also centralize the trivial peeling branch here: if
$\varepsilon_{\mathrm{peel}}\ge1$, use the universal certificate
$E_{\mathrm{struct}}^{\mathrm{cert}}=2$.  All nontrivial structural analysis
below assumes $\varepsilon_{\mathrm{peel}}<1$, while branch-explicit theorem
statements retain the universal case.

\subsection{Simultaneous Clifford localization of recovered algebras}
\label{sec:supp-simultaneous-clifford-localization}

We first turn the delocalized empirical algebras into disjoint physical
registers using only the recovered Pauli labels.
For every $C\in\widehat{\mathfrak C}$, define its phase-free generated
Pauli group and binary symplectic subspace by
\begin{equation}
\begin{aligned}
\mathcal P(C)
&:=
\left\langle\bigcup_{\alpha\in C}\mathcal G_\alpha\right\rangle,
\\
V_C
&:=
\Span_{\Fc}
\operatorname{Ax}\!\left(\{\mathcal G_\alpha:\alpha\in C\}\right)
\subseteq\Fc^{2m}.
\end{aligned}
\label{eq:supp-empirical-group-algebras}
\end{equation}
Pauli phases are ignored.  Distinct recovered sectors commute elementwise, so
\begin{equation}
V_C\perp_{\mathrm{sp}}V_{C'},
\qquad C\ne C'.
\label{eq:supp-empirical-group-symplectic-orthogonality}
\end{equation}
Moreover, the Pauli directions in $\operatorname{Ax}(\mathfrak G)$ are globally
linearly independent by
Proposition~\ref{prop:supp-block-triple-single-block}.  Since distinct
empirical groups use disjoint subsets of these Pauli directions,
\begin{equation}
\sum_{C\in\widehat{\mathfrak C}}V_C
=
\bigoplus_{C\in\widehat{\mathfrak C}}V_C.
\label{eq:supp-empirical-group-direct-sum}
\end{equation}
This argument applies equally to two empirical groups contained in the same
true hidden block; it does not invoke hidden-block separation.

Let the restricted symplectic form on $V_C$ have a nondegenerate part of
dimension $2h_C$ and radical dimension $q_C$, where
\begin{equation}
\operatorname{rad}(V_C)
:=
\left\{
v\in V_C :
[v,w]=0
\text{ for every } w\in V_C
\right\}.  
\end{equation}
A symplectic space
containing $V_C$ must contain one conjugate partner for every radical
direction, so the required register size is
\begin{equation}
k_C:=h_C+q_C.
\label{eq:supp-empirical-register-size}
\end{equation}
Each completed triple and each singleton sector contributes one localized
qubit, so $k_C=|C|$.  The recovered-sector cap following
Eq.~\eqref{eq:supp-grouping-order-promise} and
Proposition~\ref{prop:supp-block-triple-single-block} therefore give, for
$C=C_{a,h}$,
\begin{equation}
k_C\le d,
\qquad
\sum_{h=1}^{r_a}k_{C_{a,h}}
=
|\mathcal B_a^\star|
\le
L_a
\le
d.
\label{eq:supp-empirical-register-size-bounds}
\end{equation}
Because every empirical group is nonempty, this also implies
\begin{equation}
r_a
\le
\sum_{h=1}^{r_a}k_{C_{a,h}}
\le L_a\le d.
\label{eq:supp-refined-piece-count-bound}
\end{equation}
This derives the refined-piece bound before it is used below.

\begin{proposition}[Simultaneous Clifford localization of empirical groups]
\label{prop:supp-simultaneous-empirical-localization}
From the recovered binary Pauli labels, one can compute a residual Clifford
$U_{\mathrm{rec}}$ and a disjoint decomposition
\begin{equation}
[m]
=
\left(
\bigsqcup_{C\in\widehat{\mathfrak C}}J_C
\right)
\sqcup J_{\mathrm{aux}},
\qquad
|J_C|=k_C,
\label{eq:supp-empirical-localized-register-decomposition}
\end{equation}
such that, for every empirical group $C$,
\begin{equation}
U_{\mathrm{rec}}^\dagger\mathcal P(C)U_{\mathrm{rec}}
\subseteq\Pc(J_C). 
\label{eq:supp-simultaneous-empirical-localization}
\end{equation}
The binary symplectic linear algebra uses
$O(m^3+Lm^2)$ operations over $\Fc$, and the associated Clifford can be
synthesized in polynomial classical time.
\end{proposition}

\begin{proof}
For each $C$, choose hyperbolic pairs
\[
(x_{C,1},z_{C,1}),\dots,(x_{C,h_C},z_{C,h_C})
\]
for the nondegenerate part of $V_C$ and a basis
$u_{C,1},\dots,u_{C,q_C}$ of its radical.  By
Eqs.~\eqref{eq:supp-empirical-group-symplectic-orthogonality}
and~\eqref{eq:supp-empirical-group-direct-sum}, the union of these vectors
over all $C$ is independent, the displayed hyperbolic pairs have their
prescribed pairings, and every radical vector is orthogonal to every
vector in the union.

We supply the missing conjugate partners simultaneously.  Let
\begin{equation}
H
:=
\Span_{\Fc}
\{x_{C,i},z_{C,i}:C\in\widehat{\mathfrak C},\ i\in[h_C]\}.
\end{equation}
The prescribed pairings make $H$ nondegenerate.  Pass to its symplectic
orthogonal complement $H^{\perp_{\mathrm{sp}}}$.  All $u_{C,j}$ lie in
$H^{\perp_{\mathrm{sp}}}$ and form one independent isotropic family
there.  Symplectic Gram--Schmidt, applied to that entire family, produces
vectors $v_{C,j}$ satisfying
\begin{equation}
[u_{C,j},v_{C',j'}]
=
\delta_{C,C'}\delta_{j,j'},
\qquad
[v_{C,j},v_{C',j'}]=0,
\end{equation}
while keeping every $v_{C,j}$ orthogonal to the previously chosen
hyperbolic planes.  Thus
\begin{equation}
W_C
:=
\Span_{\Fc}
\{x_{C,i},z_{C,i},u_{C,j},v_{C,j}:i\in[h_C],\,j\in[q_C]\}
\end{equation}
is a $2k_C$-dimensional nondegenerate symplectic subspace, the $W_C$ are
pairwise symplectically orthogonal, and $V_C\subseteq W_C$.  This also
proves $\sum_Ck_C\le m$.

Extend the union of the bases of the $W_C$ to a symplectic basis of
$\Fc^{2m}$.  Map the basis of each $W_C$ to the standard $X$--$Z$ pairs
on a dedicated $k_C$-qubit set $J_C$, and map the remaining pairs to
$J_{\mathrm{aux}}$.  The resulting symplectic transformation is realized
by a Clifford $U_{\mathrm{rec}}$ and gives
Eqs.~\eqref{eq:supp-empirical-localized-register-decomposition}
and~\eqref{eq:supp-simultaneous-empirical-localization}.

All steps use only the binary labels of the recovered Pauli directions.  Gaussian
elimination together with symplectic Gram--Schmidt finds the adapted bases
and the simultaneous radical partners in $O(m^3+Lm^2)$ binary operations.
Standard symplectic Clifford synthesis then produces $U_{\mathrm{rec}}$
in polynomial classical time.
\end{proof}

\subsection{Recovered hidden-local algebras and truncation}
\label{sec:supp-hidden-local-truncation}

We now use proof-only hidden objects.  Recall that
\[
\rho
=
U_{\mathrm c}
\left(\bigotimes_{a=1}^K\rho_a\right)
U_{\mathrm c}^\dagger,
\qquad
\rho_a\in\D\!\left((\mathbb C^2)^{\otimes L_a}\right).
\]
For each hidden block $B_a$, define the prefix-only hidden-local subgroup
\begin{equation}
\mathcal Z_a^{\mathrm{peel}}
:=
\left\{
\Lambda\in\Pc_{B_a}:
U_{\mathrm{stab}}^\dagger
U_{\mathrm c}
\bigl(\Lambda\otimes I_{B_a^c}\bigr)
U_{\mathrm c}^\dagger
U_{\mathrm{stab}}
\in\mathcal K
\right\},
\label{eq:supp-prefix-only-local-subgroup}
\end{equation}
where
$\mathcal K=\{Z^c\otimes I:c\in\Fc^t\}$ is the peeled syndrome
subgroup.  Equivalently, these are the hidden-local Paulis whose peeled
images have trivial residual component.
For every $C_{a,h}$ and every $P\in\mathcal P(C_{a,h})$, choose a
hidden-local full lift
$\Lambda_a^\uparrow(P)\in\Pc_{B_a}$ satisfying, phase-free,
\begin{equation}
U_{\mathrm{stab}}^\dagger
U_{\mathrm c}
\left(\Lambda_a^\uparrow(P)\otimes I_{B_a^c}\right)
U_{\mathrm c}^\dagger
U_{\mathrm{stab}}
=
Z^{c(P)}\otimes P
\label{eq:supp-hidden-local-full-lift}
\end{equation}
for some $c(P)\in\Fc^t$.  Let
\begin{equation}
\mathcal H_a
:=
\left\langle
\mathcal Z_a^{\mathrm{peel}},\
\Lambda_a^\uparrow(P):
h\in[r_a],\ P\in\mathcal P(C_{a,h})
\right\rangle
\subseteq\Pc_{B_a},
\label{eq:supp-recovered-local-subgroup}
\end{equation}
and set $\mathcal H_a=\mathcal Z_a^{\mathrm{peel}}$ when
$a\notin\mathcal A^\star$.  These definitions are analytical: neither the
hidden assignment nor the lifts are available to the learner.

The group in Eq.~\eqref{eq:supp-recovered-local-subgroup} is independent
of every lift choice.  Indeed, if $\Lambda_a^\uparrow(P)$ and
$\widetilde{\Lambda}_a^\uparrow(P)$ are two hidden-local lifts of the same
residual representative, their peeled images are
$Z^c\otimes P$ and $Z^{c'}\otimes P$ phase-free.  Their quotient is the
prefix-only hidden-local Pauli whose peeled image is
$Z^{c+c'}\otimes I$.  Hence
\begin{equation}
\Lambda_a^\uparrow(P)\widetilde{\Lambda}_a^\uparrow(P)
\in\mathcal Z_a^{\mathrm{peel}},
\end{equation}
and adjoining the full prefix-only subgroup makes the generated
$\mathcal H_a$ unchanged.

Let
\begin{equation}
\mathfrak H_a
:=
\Span_{\mathbb C}\{\Lambda:\Lambda\in\mathcal H_a\}.
\end{equation}
Because $\mathcal H_a$ is a phase-free Pauli subgroup containing the
identity, products of representatives differ from representatives in
$\mathcal H_a$ only by scalar phases, and adjoints remain in the same
span.   If
$W_a^{\mathrm{rec}}\subseteq\Fc^{2L_a}$ is its binary label subspace,
define
\begin{equation}
\mathbb E_a(X)
:=
\frac{1}{|(W_a^{\mathrm{rec}})^{\perp_{\mathrm{sp}}}|}
\sum_{q\in(W_a^{\mathrm{rec}})^{\perp_{\mathrm{sp}}}}
P(q)XP^\dagger(q),
\label{eq:supp-hidden-pauli-conditional-expectation}
\end{equation}
where one Pauli representative is chosen for each label in the sum.
This Pauli twirl is completely positive, trace preserving, and unital.
It fixes exactly the Paulis whose labels lie in
$((W_a^{\mathrm{rec}})^{\perp_{\mathrm{sp}}})^{\perp_{\mathrm{sp}}}
=W_a^{\mathrm{rec}}$, so it is the trace-preserving conditional
expectation onto the operator algebra $\mathfrak H_a$.

Write
\begin{equation}
\rho_a
=
2^{-L_a}
\sum_{\Lambda\in\Pc_{B_a}}
\varphi_a(\Lambda)\Lambda,
\qquad
\varphi_a(\Lambda):=\Tr(\rho_a\Lambda),
\label{eq:supp-hidden-block-pauli-expansion-normalization}
\end{equation}
and define
\begin{equation}
\rho_a^{\mathrm{rec}}
:=
\mathbb E_a(\rho_a)
=
2^{-L_a}
\sum_{\Lambda\in\mathcal H_a}
\varphi_a(\Lambda)\Lambda.
\label{eq:supp-truncated-hidden-block-state}
\end{equation}
The coefficient-normalized missed Pauli mass and its trace-norm proxy are
\begin{equation}
\varepsilon_a^{\mathrm{miss}}
:=
\sum_{\Lambda\in\Pc_{B_a}\setminus\mathcal H_a}
|\varphi_a(\Lambda)|^2,
\qquad
E_{\mathrm{miss},a}:=\sqrt{\varepsilon_a^{\mathrm{miss}}}.
\label{eq:supp-blockwise-missed-pauli-mass}
\end{equation}
Complete positivity and trace preservation of $\mathbb E_a$ make
$\rho_a^{\mathrm{rec}}$ a density operator, while its fixed-point action
preserves the Pauli coefficients in $\mathcal H_a$ and removes all others.
Pauli orthogonality gives
\begin{equation}
\|\rho_a-\rho_a^{\mathrm{rec}}\|_2^2
=
2^{-L_a}\varepsilon_a^{\mathrm{miss}},
\qquad
\|\rho_a-\rho_a^{\mathrm{rec}}\|_1
\le2^{L_a/2}\|\rho_a-\rho_a^{\mathrm{rec}}\|_2
=E_{\mathrm{miss},a}.
\label{eq:supp-blockwise-truncation-trace-distance}
\end{equation}

Define the global truncated product model and its aggregate error by
\begin{equation}
\rho^{\mathrm{rec}}
:=
U_{\mathrm c}
\left(\bigotimes_{a=1}^K\rho_a^{\mathrm{rec}}\right)
U_{\mathrm c}^\dagger,
\qquad
E_{\mathrm{miss}}
:=
\sum_{a=1}^KE_{\mathrm{miss},a}.
\label{eq:supp-global-truncated-product-model}
\end{equation}
The hidden partition is unchanged, and telescoping the tensor product gives
\begin{equation}
\|\rho-\rho^{\mathrm{rec}}\|_1
\le E_{\mathrm{miss}}.
\label{eq:supp-hidden-product-truncation-bound}
\end{equation}

\subsection{Threshold control of missed Pauli mass}
\label{sec:supp-missed-mass-threshold}

We next control the Pauli coefficients removed by the blockwise
conditional expectations using the recovery thresholds.
Set
\begin{equation}
\theta_{\mathrm{rec}}
:=
\theta+\tau_{\mathrm{rank}}.
\end{equation}
Every hidden-local Pauli is exactly one of the following: prefix-only,
invisible non-prefix, or visible.  Prefix-only directions have peeled
image $Z^c\otimes I$ and are retained in
$\mathcal Z_a^{\mathrm{peel}}\subseteq\mathcal H_a$.  Consequently, the
omitted directions split disjointly into two classes.  An invisible
non-prefix direction has peeled image outside
$\mathcal Z_{\mathrm{pref}}$ and hence has no prefix-compatible residual
representative under the quotient map $\pi$.
A visible omitted direction has peeled image
$Z^c\otimes\bar R_a^\ell$ with $\bar R_a^\ell\ne I$, but its residual
representative is outside $V(\mathfrak G)$.  The last assertion follows
from residual block separation,
Proposition~\ref{prop:supp-blockwise-peeling-preserves-residual-blocks},
and the block purity in
Proposition~\ref{prop:supp-block-triple-single-block}: a block-$a$
direction in the global recovered span is generated entirely by recovered
directions from block $a$.  Its full lift therefore belongs to
$\mathcal H_a$ up to a retained prefix.

Let $\mathcal M_a^{\mathrm{inv}}$ and
$\mathcal M_a^{\mathrm{vis}}$ denote the two omitted classes.  They are
disjoint and exhaust
$\Pc_{B_a}\setminus\mathcal H_a$; in particular, neither the identity nor
any element of the peeling subgroup is counted.

\begin{proposition}[Threshold certificate for missed Pauli mass]
\label{prop:supp-missed-pauli-threshold-certificate}
Under the recovery concentration event and the margin hypotheses of
Proposition~\ref{prop:supp-block-triple-single-block}, every hidden block,
including one with no recovered residual direction, satisfies
\begin{align}
\varepsilon_a^{\mathrm{miss}}
&\le
|\mathcal M_a^{\mathrm{inv}}|\,\lambda
+
|\mathcal M_a^{\mathrm{vis}}|\,\theta_{\mathrm{rec}}
\nonumber\\
&\le
\bigl(4^{L_a}-|\mathcal H_a|\bigr)\theta_{\mathrm{rec}},
\label{eq:supp-block-threshold-missing-bound}\\
E_{\mathrm{miss}}
&\le
\sum_{a=1}^K
\sqrt{
\bigl(4^{L_a}-|\mathcal H_a|\bigr)\theta_{\mathrm{rec}}
}.
\label{eq:supp-global-missed-mass-sharp-certificate}
\end{align}
Consequently, using $K\le n$ and $L_a\le d$,
\begin{equation}
\sum_{a=1}^K\varepsilon_a^{\mathrm{miss}}
\le n4^d\theta_{\mathrm{rec}},
\qquad
E_{\mathrm{miss}}
\le n2^d\sqrt{\theta_{\mathrm{rec}}}.
\label{eq:supp-global-missed-mass-computable-certificate}
\end{equation}
\end{proposition}

\begin{proof}
For an invisible direction, the contrapositive of
Proposition~\ref{prop:supp-residual-representatives-after-peeling} gives
\begin{equation}
|\varphi_a(\Lambda)|^2
=
s_{\mathrm{full}}(R_a^\ell)
\le\lambda.
\end{equation}
For a visible omitted direction, if
$s_{\mathrm{full}}(R_a^\ell)>\theta+\tau_{\mathrm{rank}}$, then on
$\mathcal E_{\mathrm{rank}}$ its empirical score exceeds $\theta$, and
Eq.~\eqref{eq:supp-threshold-span-completeness} would put
$\bar R_a^\ell$ in $V(\mathfrak G)$, a contradiction.  Hence
\begin{equation}
|\varphi_a(\Lambda)|^2
=
s_{\mathrm{full}}(R_a^\ell)
\le\theta+\tau_{\mathrm{rank}}
=\theta_{\mathrm{rec}}.
\end{equation}
The threshold margin
$\theta-\tau_{\mathrm{rank}}>\lambda$ implies
$\lambda<\theta_{\mathrm{rec}}$.  Summing the two exact classes proves the
first line and then the coarser blockwise bound.  The argument also covers
$a\notin\mathcal A^\star$, for which
$\mathcal H_a=\mathcal Z_a^{\mathrm{peel}}$.  Finally,
$4^{L_a}-|\mathcal H_a|\le4^d$ and $K\le n$ give the two known-parameter
certificates.
\end{proof}

\subsection{Localized true-block-product approximation}
\label{sec:supp-localized-true-block-product}

We now compare the actual state in computable localized coordinates with a
proof-only product over aggregate true-block registers.
Let
\begin{equation}
\bar U_{\mathrm{rec}}
:=
I^{\otimes t}\otimes U_{\mathrm{rec}},
\qquad
\rho_{\mathrm{loc}}
:=
\bar U_{\mathrm{rec}}^\dagger
U_{\mathrm{stab}}^\dagger\rho U_{\mathrm{stab}}
\bar U_{\mathrm{rec}}.
\label{eq:supp-computable-localized-state}
\end{equation}
This is the actual computable localized state.  Its physical partition is
the peeled register $[t]$, the empirical residual registers
$\{J_C:C\in\widehat{\mathfrak C}\}$, and $J_{\mathrm{aux}}$.

For analysis only, use Eq.~\eqref{eq:supp-localization-hidden-refinement}
to define
\begin{equation}
J_a
:=
\bigsqcup_{h=1}^{r_a}J_{C_{a,h}},
\qquad
a\in\mathcal A^\star.
\label{eq:supp-proof-only-hidden-register-aggregation}
\end{equation}
Neither the assignment $C\mapsto a$ nor the aggregate register $J_a$ is
known to the learner.

\begin{proposition}[Localized true-block-product comparison]
\label{prop:supp-localized-true-block-product}
There exist a syndrome $b\in\{0,1\}^t$ and states
\[
\omega_a\in\D\!\left((\mathbb C^2)^{\otimes |J_a|}\right),
\qquad a\in\mathcal A^\star,
\]
such that the proof-only comparison state
\begin{equation}
\Omega_{\mathrm{true}}
:=
|b\rangle\!\langle b|
\otimes
\left(\bigotimes_{a\in\mathcal A^\star}\omega_a\right)
\otimes
\frac{I_{J_{\mathrm{aux}}}}{2^{|J_{\mathrm{aux}}|}}
\label{eq:supp-localized-true-block-model}
\end{equation}
satisfies
\begin{equation}
\|\rho_{\mathrm{loc}}-\Omega_{\mathrm{true}}\|_1
\le
2\sqrt{\varepsilon_{\mathrm{peel}}}
+
E_{\mathrm{miss}}.
\label{eq:supp-block-normalized-close-product-form}
\end{equation}
\end{proposition}

\begin{proof}
If $\varepsilon_{\mathrm{peel}}\ge1$, choose arbitrary states $\omega_a$;
then the universal trace-norm bound gives
$\|\rho_{\mathrm{loc}}-\Omega_{\mathrm{true}}\|_1\le2
\le2\sqrt{\varepsilon_{\mathrm{peel}}}+E_{\mathrm{miss}}$.
Hence assume $\varepsilon_{\mathrm{peel}}<1$, when $b$ is the dominant
syndrome.
Let $\widetilde\rho^{\mathrm{rec}}
:=
U_{\mathrm{stab}}^\dagger\rho^{\mathrm{rec}}U_{\mathrm{stab}}$. 
First, truncation preserves the complete peeling subgroup.  Let
$\mathcal S_a=\mathcal S_{\mathrm{peel}}\cap\mathcal A_a$ as in the
peeled hidden-block geometry.  The internal-product factorization
$\mathcal S_{\mathrm{peel}}=\mathcal S_1\cdots\mathcal S_K$ means that
every phase-free $g\in\mathcal S_{\mathrm{peel}}$ has a unique product
of block-local factors.  Choose the signed Hermitian representative
$\widetilde g$ induced by the commuting selected generators
$g_1,\dots,g_t$.  After applying $U_{\mathrm{stab}}^\dagger$, every
block-local factor is prefix-only, so its hidden-local preimage belongs
to $\mathcal Z_a^{\mathrm{peel}}\subseteq\mathcal H_a$.  The blockwise
conditional expectations therefore preserve every product in the
peeling subgroup and hence
\begin{equation}
\Tr(\rho^{\mathrm{rec}}\widetilde g)
=
\Tr(\rho\widetilde g),
\qquad
g\in\mathcal S_{\mathrm{peel}}.
\label{eq:supp-full-peeling-subgroup-moment-preservation}
\end{equation}
The Fourier formula therefore preserves the complete syndrome distribution.
For its dominant string $b$, let
\begin{equation}
\Pi_b
:=
\prod_{j=1}^t
\frac{I+(-1)^{b_j}g_j}{2}
=
2^{-t}
\sum_{c\in\Fc^t}
(-1)^{b\cdot c}g(c),
\qquad
g(c):=\prod_{j=1}^tg_j^{c_j}.
\label{eq:supp-phase-safe-syndrome-projector}
\end{equation}
Here the $g_j$ are their signed Hermitian commuting representatives, so
the second expression is the Fourier expansion over the signed
representatives $g(c)$.  Set
\begin{equation}
p_b
:=
\Tr(\Pi_b\rho^{\mathrm{rec}})
=
\Tr(\Pi_b\rho).
\end{equation}
Certified peeling gives
\begin{equation}
p_b
\ge
1-\varepsilon_{\mathrm{peel}}
>
0.
\label{eq:supp-positive-compressed-syndrome-weight}
\end{equation}

In peeled coordinates,
$U_{\mathrm{stab}}^\dagger\Pi_bU_{\mathrm{stab}}
=|b\rangle\!\langle b|\otimes I$.  Because the peeling subgroup is the
internal product of its block-local subgroups, the signed representative
and the character $g(c)\mapsto(-1)^{b\cdot c}$ restrict blockwise, and
\begin{equation}
U_{\mathrm c}^\dagger\Pi_bU_{\mathrm c}
=
\bigotimes_{a=1}^K\Pi_{a,b},
\qquad
\Pi_{a,b}
:=
\frac{1}{|\mathcal S_a|}
\sum_{g_a\in\mathcal S_a}
\chi_b(\widetilde g_a)
U_{\mathrm c}^\dagger\widetilde g_aU_{\mathrm c},
\label{eq:supp-blockwise-syndrome-projector-factorization}
\end{equation}
where $\widetilde g_a$ is the signed representative induced by the
selected generators and $\chi_b(\widetilde g_a)$ is the restricted
character.
Thus the unnormalized compressed truncated model factorizes exactly:
\begin{equation}
U_{\mathrm c}^\dagger
\Pi_b\rho^{\mathrm{rec}}\Pi_b
U_{\mathrm c}
=
\bigotimes_{a=1}^K
\Pi_{a,b}\rho_a^{\mathrm{rec}}\Pi_{a,b}.
\label{eq:supp-unnormalized-compressed-product-factorization}
\end{equation}
Its positive trace $p_b$ is the product of the blockwise traces, so
normalization preserves the factorization.  Define
$\rho_{\mathrm{res}}^{\mathrm{rec}}$ by
\begin{equation}
\frac{
U_{\mathrm{stab}}^\dagger
\Pi_b\rho^{\mathrm{rec}}\Pi_b
U_{\mathrm{stab}}
}{p_b}
=
|b\rangle\!\langle b|
\otimes
\rho_{\mathrm{res}}^{\mathrm{rec}}.
\label{eq:supp-normalized-compressed-residual-state}
\end{equation}

Equation~\eqref{eq:supp-full-peeling-subgroup-moment-preservation} also
preserves the signed expectations of the selected generators.  Applying
the certified peeling comparison to the truncated model with the same
dominant syndrome gives
\begin{equation}
\left\|
\widetilde\rho^{\mathrm{rec}}
-
|b\rangle\!\langle b|
\otimes\rho_{\mathrm{res}}^{\mathrm{rec}}
\right\|_1
\le
2\sqrt{\varepsilon_{\mathrm{peel}}}.
\label{eq:supp-truncated-model-peeling-bound}
\end{equation}

Syndrome compression replaces each retained prefix $Z^c$ by
$(-1)^{b\cdot c}$.  Every remaining block-$a$ Pauli lies in
$\langle\mathcal P(C_{a,h}):h\in[r_a]\rangle$, which
Proposition~\ref{prop:supp-simultaneous-empirical-localization} maps into
$J_a$; unrepresented blocks reduce to scalars.  Thus
$U_{\mathrm{rec}}^\dagger\rho_{\mathrm{res}}^{\mathrm{rec}}U_{\mathrm{rec}}$
factors over the $J_a$, defining $\omega_a$, while the absence of retained
nonidentity coefficients on $J_{\mathrm{aux}}$ gives its maximally mixed
factor.  This proves Eq.~\eqref{eq:supp-localized-true-block-model}.

Finally, Eq.~\eqref{eq:supp-hidden-product-truncation-bound}, unitary
invariance, Eq.~\eqref{eq:supp-truncated-model-peeling-bound}, and the
triangle inequality give
Eq.~\eqref{eq:supp-block-normalized-close-product-form}, with each error
included once.
\end{proof}

In the nontrivial exact-grouping branch, every represented hidden block
supplies one empirical group $C=\mathcal B_a^\star$.  Then $r_a=1$ and
$J_a=J_C$, so Eq.~\eqref{eq:supp-localized-true-block-model} is already a
product over the computable empirical registers.  Its split error is zero,
regardless of whether $E_{\mathrm{miss}}$ vanishes.

\subsection{Approximation under certified over-refinement}
\label{sec:supp-certified-overrefinement-product}

We first treat guessed-scale over-refinement and then consolidate every
branch.  By the trivial-branch convention of
Sec.~\ref{sec:supp-localization-interface}, assume
$\varepsilon_{\mathrm{peel}}<1$ and suppose that a represented hidden block
is split as in
Eq.~\eqref{eq:supp-localization-hidden-refinement}.  Write
\begin{equation}
J_{a,h}:=J_{C_{a,h}},
\qquad
k_{a,h}:=|J_{a,h}|,
\qquad
k_a:=|J_a|=\sum_{h=1}^{r_a}k_{a,h},
\end{equation}
and define the marginal
\begin{equation}
\omega_{a,h}
:=
\Tr_{J_a\setminus J_{a,h}}(\omega_a).
\end{equation}
Define the localized empirical subgroups
\begin{equation}
G_{a,h}
:=
U_{\mathrm{rec}}^\dagger
\mathcal P(C_{a,h})
U_{\mathrm{rec}}
\subseteq
\Pc(J_{a,h}).
\label{eq:supp-localized-empirical-subgroups}
\end{equation}
The disjoint $J_{a,h}$ make tuples drawn from distinct $G_{a,h}$ commuting,
exactly as required by the terminal tests.

Keep the terminal guessed-scale ceiling $\xi_s$ from
Eq.~\eqref{eq:supp-guessed-scale-small-final-cumulants} distinct from the
effective localized tolerance
\begin{equation}
\xi_{\mathrm{eff}}
:=
\xi_s+\beta_{\mathrm{peel}}.
\label{eq:supp-effective-localized-cumulant-tolerance}
\end{equation}
For the calibrated guessed-scale specialization,
$\xi_s=3\eta_s/4$ and
$\beta_{\mathrm{peel}}<\eta_s/4$, so
$\xi_{\mathrm{eff}}<\eta_s$.

For $q\ge1$ and $x\ge0$, define
\begin{equation}
F_q(x):=(1+x)^{2^q-q-1}-1,
\label{eq:supp-split-cumulant-factor}
\end{equation}
and let
\begin{equation}
N_{\mathrm{cross},a}
:=4^{k_a}-1-\sum_{h=1}^{r_a}(4^{k_{a,h}}-1)
\label{eq:supp-split-cross-pauli-count}
\end{equation}
be the exact number of phase-free Pauli strings acting nontrivially on at
least two refined pieces.

\begin{proposition}[Localized product approximation under over-refinement]
\label{prop:supp-localized-overrefined-product}
Assume the guessed-scale over-refinement branch and
$\varepsilon_{\mathrm{peel}}<1$.
For each $a\in\mathcal A^\star$, every $T\subseteq[r_a]$ with $|T|\ge2$,
and
\begin{equation}
P_h\in G_{a,h}\setminus\{I\},
\qquad h\in T,
\end{equation}
one has
\begin{equation}
\left|
\kappa_{\omega_a}(P_h:h\in T)
\right|
\le
\xi_{\mathrm{eff}}.
\label{eq:supp-localized-split-cumulants}
\end{equation}
For every $a\in\mathcal A^\star$, one also has
\begin{equation}
\left\|\omega_a-\bigotimes_{h=1}^{r_a}\omega_{a,h}\right\|_1
\le
\sqrt{N_{\mathrm{cross},a}}\,F_{r_a}(\xi_{\mathrm{eff}}).
\label{eq:supp-split-product-error-bound}
\end{equation}
\end{proposition}

\begin{proof}
\emph{Cumulant transfer.}
Set $O_h:=U_{\mathrm{rec}}P_hU_{\mathrm{rec}}^\dagger$.  For every
nonempty $S\subseteq T$, closure of $\mathcal H_a$ supplies a full lift
with peeled image
\begin{equation}
U_{\mathrm{stab}}^\dagger
U_{\mathrm c}
\left(\Lambda_{a,S}^\uparrow\otimes I_{B_a^c}\right)
U_{\mathrm c}^\dagger
U_{\mathrm{stab}}
=
Z^{c_S}\otimes\prod_{h\in S}O_h
\end{equation}
phase-free.  Multiplication by the retained full prefix-only element removes
$Z^{c_S}$, so truncation preserves the signed residual subset moment:
\begin{equation}
\Tr\!\left[
\widetilde\rho^{\mathrm{rec}}
\left(I^{\otimes t}\otimes\prod_{h\in S}O_h\right)
\right]
=
\Tr\!\left[
\widetilde\rho
\left(I^{\otimes t}\otimes\prod_{h\in S}O_h\right)
\right].
\label{eq:supp-tested-moment-truncation-preservation}
\end{equation}
The prefix-subgroup argument makes this independent of the lift choice.
Proposition~\ref{prop:supp-localized-true-block-product} identifies these
localized subset moments with those of $\omega_a$, while
Eq.~\eqref{eq:supp-truncated-model-peeling-bound} changes each by at most
$2\sqrt{\varepsilon_{\mathrm{peel}}}$.  Therefore
Proposition~\ref{prop:supp-bounded-moment-cumulant-perturbation} gives
\begin{equation}
\left|
\kappa_{\omega_a}(P_h:h\in T)
-
\kappa_{\mathrm{res}}(O_h:h\in T)
\right|
\le
2\Gamma_{|T|}\sqrt{\varepsilon_{\mathrm{peel}}}
\le
\beta_{\mathrm{peel}}.
\end{equation}
Combining this with the terminal residual ceiling $\xi_s$ proves
Eq.~\eqref{eq:supp-localized-split-cumulants}.

\emph{Product approximation.}
The retained-coefficient argument in
Proposition~\ref{prop:supp-localized-true-block-product} gives
$\operatorname{supp}_{\Pc}(\omega_a)
\subseteq\prod_{h=1}^{r_a}G_{a,h}$, where the product is unambiguous on
the disjoint registers and
$\operatorname{supp}_{\Pc}(\omega_a)
:=\{P\in\Pc(J_a):\Tr(\omega_aP)\ne0\}$.
Fix a phase-free Pauli string that acts nontrivially on the pieces indexed
by $T$, with $q:=|T|\ge2$, and write its factors as $P_h$, $h\in T$.
If some $P_h\notin G_{a,h}$, both its joint coefficient and the
corresponding product-of-marginals coefficient vanish.  Otherwise, set
\begin{equation}
M(S)
:=
\Tr\!\left(\omega_a\prod_{h\in S}P_h\right),
\qquad S\subseteq T.
\end{equation}
The moment--cumulant relation gives
\begin{equation}
M(T)
=
\sum_{\mathfrak p\in\Pi(T)}
\prod_{E\in\mathfrak p}
\kappa_{\omega_a}(P_h:h\in E).
\end{equation}
The discrete partition is the product of the one-piece marginal
coefficients.  A remaining partition with $j$ nonsingleton blocks has
magnitude at most $\xi_{\mathrm{eff}}^j$.  Since there are
$2^q-q-1$ nonsingleton subsets of $T$, discarding the disjointness
constraint gives
\begin{align}
\left|
M(T)-\prod_{h\in T}M(\{h\})
\right|
&\le
\sum_{j\ge1}
\binom{2^q-q-1}{j}\xi_{\mathrm{eff}}^j
\nonumber\\
&=
F_q(\xi_{\mathrm{eff}})
\le
F_{r_a}(\xi_{\mathrm{eff}}).
\label{eq:supp-split-moment-factorization-error}
\end{align}

Let
$\Delta_a:=\omega_a-\bigotimes_h\omega_{a,h}$.
All Pauli coefficients of $\Delta_a$ supported on at most one piece
vanish, and there are exactly $N_{\mathrm{cross},a}$ remaining strings.
Pauli orthogonality and
$\|\Delta_a\|_1\le2^{k_a/2}\|\Delta_a\|_2$ therefore give
\begin{equation}
\|\Delta_a\|_1
\le
\sqrt{N_{\mathrm{cross},a}}\,F_{r_a}(\xi_{\mathrm{eff}}),
\end{equation}
proving the claim with the paper's trace-norm convention.
\end{proof}

For each block with $r_a\ge2$, choose one empirical group $C=C_{a,h}$ as
an anchor.  Dropping the non-anchor contributions and using $k_a\le d$
gives
\begin{align}
N_{\mathrm{cross},a}
&=
4^{k_a}-1
-
\sum_{h=1}^{r_a}(4^{k_{a,h}}-1)
\nonumber\\
&\le
4^{k_a}-4^{k_C}
\nonumber\\
&\le
4^d-4^{k_C}.
\end{align}
The anchors are distinct, $r_a\le d$, and $F_q(x)$ is increasing in both
arguments.  Blockwise application followed by tensor-product telescoping
therefore gives the direct proof-only-to-computable chain
\begin{align}
\sum_{a\in\mathcal A^\star}
\left\|\omega_a-\bigotimes_{h=1}^{r_a}\omega_{a,h}\right\|_1
&\le
F_d(\xi_{\mathrm{eff}})
\sum_{\substack{a\in\mathcal A^\star\\r_a\ge2}}
\sqrt{N_{\mathrm{cross},a}}
\nonumber\\
&\le
E_{\mathrm{split}}^{\mathrm{cert}}
:=
F_d(\xi_{\mathrm{eff}})
\sum_{C\in\widehat{\mathfrak C}}
\sqrt{\max\{0,4^d-4^{k_C}\}}.
\label{eq:supp-computable-split-error-certificate}
\end{align}
The left-hand aggregation and anchors use the hidden assignment only in the
proof; the final expression depends solely on learner-visible register sizes,
$d$, and $\xi_{\mathrm{eff}}$.  Since $\xi_{\mathrm{eff}}<\eta_s$, replacing
it by $\eta_s$ gives the coarser bound used in the calibrated ledger.

Combine the universal branch, the exact branch with zero split error, and
the guessed-scale certificate into the single branchwise quantity
\begin{equation}
E_{\mathrm{struct}}^{\mathrm{cert}}
:=
\begin{cases}
2,
&\text{if }\varepsilon_{\mathrm{peel}}\ge1,\\[1mm]
2\sqrt{\varepsilon_{\mathrm{peel}}}
+n2^d\sqrt{\theta_{\mathrm{rec}}},
&\text{if }\varepsilon_{\mathrm{peel}}<1
\text{ under the exact-grouping hypotheses},\\[2mm]
2\sqrt{\varepsilon_{\mathrm{peel}}}
+n2^d\sqrt{\theta_{\mathrm{rec}}}
+E_{\mathrm{split}}^{\mathrm{cert}},
&\text{if }\varepsilon_{\mathrm{peel}}<1
\text{ in the guessed-scale branch}.
\end{cases}
\label{eq:supp-localization-structural-certificate}
\end{equation}

Let $b$ be the certified-peeling string and define the proof-only empirical
factors once, branchwise, by
\begin{equation}
\omega_C
:=
\begin{cases}
\text{an arbitrary state in }
\D\!\left((\mathbb C^2)^{\otimes |J_C|}\right),
&\varepsilon_{\mathrm{peel}}\ge1,\\[1mm]
\omega_a,
&\varepsilon_{\mathrm{peel}}<1\text{ under the exact-grouping hypotheses, }
C=\mathcal B_a^\star,\\[1mm]
\Tr_{J_a\setminus J_C}(\omega_a),
&\varepsilon_{\mathrm{peel}}<1\text{ under guessed-scale over-refinement, }
C=C_{a,h}.
\end{cases}
\label{eq:supp-branchwise-empirical-factors}
\end{equation}
These are proof-only witnesses because the assignment $C\mapsto a$ is
hidden, whereas every supporting register $J_C$ is learner-computable.

\begin{theorem}[Empirical block-product approximation]
\label{thm:supp-empirical-block-product-approximation}
Let the empirical factors be those in
Eq.~\eqref{eq:supp-branchwise-empirical-factors}.  Then
\begin{equation}
\Omega_{\mathrm{emp}}
:=
\left(|b\rangle\!\langle b|\right)_{[t]}
\otimes
\left(
\bigotimes_{C\in\widehat{\mathfrak C}}
\left(\omega_C\right)_{J_C}
\right)
\otimes
\frac{I_{J_{\mathrm{aux}}}}{2^{|J_{\mathrm{aux}}|}}
\label{eq:supp-empirical-block-product-witness}
\end{equation}
is defined in every branch and satisfies the single certificate
\begin{equation}
\|\rho_{\mathrm{loc}}-\Omega_{\mathrm{emp}}\|_1
\le E_{\mathrm{struct}}^{\mathrm{cert}}.
\end{equation}
\end{theorem}

\begin{proof}
If $\varepsilon_{\mathrm{peel}}\ge1$, the trace-norm diameter gives the
universal bound $2$.  Otherwise,
Proposition~\ref{prop:supp-localized-true-block-product} contributes
$2\sqrt{\varepsilon_{\mathrm{peel}}}+E_{\mathrm{miss}}$.  Exact grouping
already places one factor on each empirical register.  In the guessed-scale
branch, replacing each $\omega_a$ by its marginal product adds at most
$E_{\mathrm{split}}^{\mathrm{cert}}$ by
Eq.~\eqref{eq:supp-computable-split-error-certificate}.  Finally use
Eq.~\eqref{eq:supp-global-missed-mass-computable-certificate}; the three
cases are exactly the branchwise definition of
$E_{\mathrm{struct}}^{\mathrm{cert}}$.
\end{proof}

Thus a larger guessed scale reduces estimation cost but permits a larger
split certificate; a smaller scale does the reverse and requires a lower
peeling floor.

\section{Recovered Block Tomography}
\label{sec:supp-block-tomography}
\label{sec:supp-deferred-syndrome-sign-recovery}

On $\mathcal E_{\mathrm{str}}$, the structural handoff is
$(U_{\mathrm{stab}},U_{\mathrm{rec}},\widehat{\mathfrak C},
\{(J_C,k_C)\}_C,J_{\mathrm{aux}})$ with $k_C\le d$, together with the
computable certificate
\eqref{eq:supp-localization-structural-certificate}; the proof-only witness
$\Omega_{\mathrm{emp}}$ is supplied by
Theorem~\ref{thm:supp-empirical-block-product-approximation}.  Only syndrome
signs and empirical-register marginals remain to be measured.

\subsection{Syndrome recovery and empirical marginals}

Set $\bar U_{\mathrm{rec}}=I^{\otimes t}\otimes U_{\mathrm{rec}}$,
$\widehat K=|\widehat{\mathfrak C}|$, and choose $\varepsilon_{\mathrm{tom}}>0$
and $\zeta_{\mathrm{sgn}},\zeta_{\mathrm{tom}}\in(0,1)$.  Sign measurements
depend only on the peeling output, so their
analysis may be deferred past construction of the remaining transcript
although the end-to-end algorithm performs them earlier.

Recall the localized state $\rho_{\mathrm{loc}}$ from
Eq.~\eqref{eq:supp-computable-localized-state}.  For every learner-known
empirical register, define its marginal by
\begin{equation}
\nu_C
:=
\Tr_{\overline{J_C}}(\rho_{\mathrm{loc}}),
\qquad
C\in\widehat{\mathfrak C},
\label{eq:supp-actual-empirical-marginals}
\end{equation}
where the complement is taken in the full peeled-and-residual system.
The operational residual product target is
\begin{equation}
\nu_{\mathrm{emp}}
:=
\left(
\bigotimes_{C\in\widehat{\mathfrak C}}\nu_C
\right)
\otimes
\frac{I_{J_{\mathrm{aux}}}}{2^{|J_{\mathrm{aux}}|}}.
\label{eq:supp-operational-empirical-product-target}
\end{equation}
The maximally mixed auxiliary factor is prescribed by the structural
model; it need not equal the true $J_{\mathrm{aux}}$-marginal of
$\rho_{\mathrm{loc}}$.

Bell sampling retains only
$s_\rho(g_j)=|\Tr(\rho g_j)|^2$, so it cannot determine the syndrome signs.
Fresh signed measurements recover the sign-defined bits using the zero-safe
convention
\begin{equation}
(-1)^{b_j}\Tr(\rho g_j)=|\Tr(\rho g_j)|.
\label{eq:supp-peeled-syndrome-sign-convention}
\end{equation}
On $\mathcal E_{\mathrm{BS}}$,
Eq.~\eqref{eq:supp-certified-peeling-generator-score} and
$h\in\mathsf I_h=[h_{\min},h_{\max}]$ give the nonzero bound
\begin{equation}
|\Tr(\rho g_j)|^2\ge h_{\min},
\qquad
|\Tr(\rho g_j)|
\ge
\sqrt{h_{\min}}.
\label{eq:supp-selected-generator-certified-magnitude}
\end{equation}

\begin{proposition}[Syndrome-sign recovery]
\label{prop:supp-peeled-syndrome-sign-recovery}
Fix the learner-visible peeling output, with selected generators
$g_1,\dots,g_t$ and peeling interval
$\mathsf I_h=[h_{\min},h_{\max}]\subset(1/2,1)$.  For every $j$, measure
$g_j$ on $M_{\mathrm{sgn}}$ fresh copies of $\rho$, form the empirical
mean $\widehat m_j$, and set
\begin{equation}
\widehat b_j
:=
\begin{cases}
0,&\widehat m_j\ge0,\\
1,&\widehat m_j<0.
\end{cases}
\end{equation}
If $t\ge1$ and
\begin{equation}
M_{\mathrm{sgn}}
\ge
\frac{2}{h_{\min}}
\log\frac{2n}{\zeta_{\mathrm{sgn}}},
\label{eq:supp-syndrome-sign-sample-choice}
\end{equation}
then, on $\mathcal E_{\mathrm{BS}}$ and conditional on the fixed peeling
output, $\widehat b=b$ with probability at least
$1-\zeta_{\mathrm{sgn}}$.  Direct Pauli measurements use $N_{\mathrm{sgn}}=tM_{\mathrm{sgn}}$ copies. 
If $t=0$, the syndrome string is empty, recovery is vacuously exact, and
$N_{\mathrm{sgn}}=0$.
\end{proposition}

\begin{proof}
On $\mathcal E_{\mathrm{BS}}$,
Eqs.~\eqref{eq:supp-peeled-syndrome-sign-convention}
and~\eqref{eq:supp-selected-generator-certified-magnitude}, followed by
Hoeffding concentration for the fresh $\{\pm1\}$ outcomes, give
\begin{equation}
\Pr(\widehat b_j\ne b_j\mid\text{fixed transcript})
\le
\exp\!\left(-\frac{M_{\mathrm{sgn}}h_{\min}}{2}\right).
\end{equation}
Equation~\eqref{eq:supp-syndrome-sign-sample-choice} and a union bound over
$t\le n$ prove the claim; the case $t=0$ is immediate.
\end{proof}

\subsection{Tomography of empirical registers}

With the syndrome signs recovered, the learner reconstructs the marginal
on every empirical register.
Choose $\varepsilon_C>0$ and $\zeta_C\in(0,1)$ for
$C\in\widehat{\mathfrak C}$ with
\begin{equation}
\sum_{C\in\widehat{\mathfrak C}}\varepsilon_C
\le\varepsilon_{\mathrm{tom}},
\qquad
\sum_{C\in\widehat{\mathfrak C}}\zeta_C
\le\zeta_{\mathrm{tom}}.
\label{eq:supp-empirical-tomography-local-budgets}
\end{equation}
All local Paulis may be measured after the known inverse Cliffords, or
equivalently as the pulled-back Pauli on $\rho$,
\begin{equation}
U_{\mathrm{stab}}\bar U_{\mathrm{rec}}
\left(I^{\otimes t}\otimes P\right)
\bar U_{\mathrm{rec}}^\dagger U_{\mathrm{stab}}^\dagger,
\label{eq:supp-empirical-pauli-pullback}
\end{equation}
where $P$ is embedded as the identity on residual registers outside
$J_C$.

\begin{proposition}[Brute-force tomography of empirical-register marginals]
\label{prop:supp-block-product-bruteforce-tomography}
Given fresh copies of the localized state $\rho_{\mathrm{loc}}$ and budgets satisfying
Eq.~\eqref{eq:supp-empirical-tomography-local-budgets}, there is a direct
Pauli tomography procedure that outputs physical states
$\widehat\nu_C\in\D((\mathbb C^2)^{\otimes k_C})$.  Write
\begin{equation}
\widehat\nu_{\mathrm{emp}}
:=
\left(
\bigotimes_{C\in\widehat{\mathfrak C}}\widehat\nu_C
\right)
\otimes
\frac{I_{J_{\mathrm{aux}}}}{2^{|J_{\mathrm{aux}}|}}.
\end{equation}
With probability at least $1-\zeta_{\mathrm{tom}}$, simultaneously for
every $C\in\widehat{\mathfrak C}$,
\begin{equation}
\left\|
\widehat\nu_C-\nu_C
\right\|_1
\le
\varepsilon_C,
\end{equation}
and consequently
\begin{equation}
\left\|
\widehat\nu_{\mathrm{emp}}-\nu_{\mathrm{emp}}
\right\|_1
\le
\varepsilon_{\mathrm{tom}}.
\label{eq:supp-block-product-tomography-error}
\end{equation}
When $\widehat K\ge1$, the procedure uses
\begin{equation}
N_{\mathrm{bp}}
=
O\!\left(
\displaystyle
\max_{C\in\widehat{\mathfrak C}}
\frac{(4^{k_C}-1)^2}{\varepsilon_C^2}
\log\frac{4^{k_C}}{\zeta_C}
\right)
\label{eq:supp-block-product-tomography-general-cost}
\end{equation}
fresh copies, equivalently global joint-measurement rounds.

If $\widehat K=0$, no marginal or auxiliary register is tomographed:
set $N_{\mathrm{bp}}=0$ and use the empty product
$\widehat\nu_{\mathrm{emp}}=\nu_{\mathrm{emp}}
=I_{J_{\mathrm{aux}}}/2^{|J_{\mathrm{aux}}|}$, interpreted as the scalar
$1$ when $J_{\mathrm{aux}}=\varnothing$.  When $\widehat K\ge1$, define
$k_{\max}^{\mathrm{emp}}:=\max_C k_C$ and choose the uniform budgets
$\varepsilon_C=\varepsilon_{\mathrm{tom}}/\widehat K$ and
$\zeta_C=\zeta_{\mathrm{tom}}/\widehat K$.  Then
\begin{equation}
N_{\mathrm{bp}}
=
O\!\left(
\frac{\widehat K^2 16^{k_{\max}^{\mathrm{emp}}}}
{\varepsilon_{\mathrm{tom}}^2}
\log\frac{\widehat K4^{k_{\max}^{\mathrm{emp}}}}
{\zeta_{\mathrm{tom}}}
\right)
=O\!\left(
\frac{\widehat K^2 16^d}
{\varepsilon_{\mathrm{tom}}^2}
\log\frac{\widehat K4^d}{\zeta_{\mathrm{tom}}}
\right), 
\qquad
\label{eq:supp-block-product-tomography-max-block-cost}
\end{equation}
as $k_{\max}^{\mathrm{emp}}\le d$. 
\end{proposition}

\begin{proof}
The stated empty-product convention handles $\widehat K=0$.  Assume
$\widehat K\ge1$, fix $C\in\widehat{\mathfrak C}$, and set
\begin{align}
N_C^{\mathrm{Pauli}}&:=4^{k_C}-1,
&
\tau_C^{\mathrm{tom}}
&:=\frac{\varepsilon_C}{4\sqrt{N_C^{\mathrm{Pauli}}}},
\nonumber\\
M_C^{\mathrm{Pauli}}
&:=
\left\lceil
\frac{2}{(\tau_C^{\mathrm{tom}})^2}
\log\frac{2N_C^{\mathrm{Pauli}}}{\zeta_C}
\right\rceil
=
O\!\left(
\frac{N_C^{\mathrm{Pauli}}}{\varepsilon_C^2}
\log\frac{4^{k_C}}{\zeta_C}
\right),
\nonumber\\
L_C
&:=
N_C^{\mathrm{Pauli}}M_C^{\mathrm{Pauli}}
=
O\!\left(
\frac{(4^{k_C}-1)^2}{\varepsilon_C^2}
\log\frac{4^{k_C}}{\zeta_C}
\right),
&
L_{\max}
&:=
\max_{C\in\widehat{\mathfrak C}}L_C.
\end{align}
Here $M_C^{\mathrm{Pauli}}$ is the Hoeffding sample count for each
$\mu_C(P):=\Tr(\nu_CP)$, estimated by $\widehat\mu_C(P)$.  In each
of the $L_{\max}$ fresh-copy rounds, jointly measure the next scheduled Pauli
on every active disjoint register.  Hence copies are the maximum of the local
schedules, not their sum; Eq.~\eqref{eq:supp-empirical-pauli-pullback}
implements the schedule on $\rho$.  With probability at least $1-\zeta_C$,
\begin{equation}
|\widehat\mu_C(P)-\mu_C(P)|
\le\tau_C^{\mathrm{tom}}
\quad\text{for every }P\ne I_{J_C}.
\end{equation}
Same-round cross-register correlations are harmless to these coefficientwise
bounds.  A union bound over registers proves
Eq.~\eqref{eq:supp-block-product-tomography-general-cost}.

Pauli linear inversion~\cite{AltepeterJamesKwiat2004} gives
\begin{equation}
\widehat\nu_C^{\mathrm{lin}}
:=
2^{-k_C}
\left(
I_{J_C}
+
\sum_{P\in\Pc(J_C)\setminus\{I_{J_C}\}}
\widehat\mu_C(P)P
\right).
\end{equation}
It is Hermitian and trace one, though not necessarily positive.  On the
uniform moment event, Pauli orthogonality gives
\begin{equation}
\left\|\widehat\nu_C^{\mathrm{lin}}-\nu_C\right\|_2
\le2^{-k_C/2}\sqrt{N_C^{\mathrm{Pauli}}}\,\tau_C^{\mathrm{tom}}.
\end{equation}
Let $\nu_C^\star$ be its exact Hilbert--Schmidt projection onto the density
operators, obtained by projecting its eigenvalues onto the probability
simplex, and return a finite-precision physical approximation satisfying
\begin{equation}
\left\|\widehat\nu_C-\nu_C^\star\right\|_1
\le
\frac{\varepsilon_C}{2}.
\label{eq:supp-empirical-projection-numerical-precision}
\end{equation}
The minimizing property and $\|X\|_1\le2^{k_C/2}\|X\|_2$ give
\begin{align}
\left\|\nu_C^\star-\nu_C\right\|_1
&\le
2^{1+k_C/2}
\left\|\widehat\nu_C^{\mathrm{lin}}-\nu_C\right\|_2
\le
2\sqrt{N_C^{\mathrm{Pauli}}}\,\tau_C^{\mathrm{tom}}
=
\frac{\varepsilon_C}{2},\\
\left\|\widehat\nu_C-\nu_C\right\|_1
&\le
\left\|\widehat\nu_C-\nu_C^\star\right\|_1
+\left\|\nu_C^\star-\nu_C\right\|_1
\le
\varepsilon_C,
\end{align}
using Eq.~\eqref{eq:supp-empirical-projection-numerical-precision}.

A union bound and tensor-product telescoping yield
\begin{equation}
\left\|\bigotimes_C\widehat\nu_C-\bigotimes_C\nu_C\right\|_1
\le
\sum_C\left\|\widehat\nu_C-\nu_C\right\|_1
\le
\varepsilon_{\mathrm{tom}}.
\end{equation}
Appending the same normalized identity on $J_{\mathrm{aux}}$ does not
change the trace norm.  This proves
Eq.~\eqref{eq:supp-block-product-tomography-error}; substituting the
uniform budgets and maximizing over the empirical sizes proves
Eq.~\eqref{eq:supp-block-product-tomography-max-block-cost}.
\end{proof}

Thus $N_{\mathrm{bp}}$ counts global rounds.  Storage, naive assembly, and
dense projection remain exponential only in $k_C\le d$;
$J_{\mathrm{aux}}$ is never tomographed.

\section{End-to-End Block-Product Tomography}
\label{sec:supp-end-to-end}

This section composes the preceding learner-visible interfaces using the
standing fresh-pool and conditional-transcript conventions; the final union
bound charges each stage budget once.

\subsection{End-to-End Protocol and Output}

For the generic block-product schedule in this subsection, assume
$2\le d\le n$.  The $d=1$ case, for which grouping is unnecessary, is
treated separately in
Sec.~\ref{sec:supp-single-qubit-hidden-block-specialization}.

The learner is given \(n\), target trace-norm accuracy
\(\varepsilon\in(0,1)\), target failure probability
\(\delta\in(0,1)\), and a known hidden-block-size cap \(d\), so every hidden
block has size at most \(d\).  It has independent-copy access to the unknown
CEBP state and can perform the two-copy Bell measurements and adaptive
commuting Pauli measurements used above.  Before sampling, it derives
compatible thresholds, tolerances, accuracy and failure budgets, and
reserves five disjoint sample pools at learner-known worst-case sizes.  Once
the structural transcript is available, the learner uses only the
transcript-specific portion of the reserved sign and tomography pools.  The
calibrated choices below are one fully learner-known implementation; the
abstract theorem isolates the stagewise properties that this schedule
verifies.  All comparisons and matrix projections use enough classical bit
precision to meet
Eq.~\eqref{eq:supp-empirical-projection-numerical-precision}.

The canonical compact physical output is
\begin{equation}
\left(
U_{\mathrm{stab}},
\bar U_{\mathrm{rec}},
\widehat b,
\{(J_C,\widehat\nu_C)\}_{C\in\widehat{\mathfrak C}},
J_{\mathrm{aux}}
\right),
\label{eq:supp-abstract-output-tuple}
\end{equation}
which represents the factorized estimator returned below.

\begin{inlinealgorithm}{alg:supp-block-product-main-tomography}{End-to-end block-product tomography}
\begin{algorithmic}[1]
\Require Independent-copy access to \(\rho\); \(n\), target accuracy
\(\varepsilon\), target failure probability \(\delta\), and a known valid
hidden-block-size cap \(2\le d\le n\).
\Ensure Either \textsc{failure} or a compact physical factorized estimator
\(\widehat\rho_{\mathrm{bp}}\).
\State Compute the learner-known worst-case parameter schedule in
Sec.~\ref{sec:supp-calibrated-schedule} and reserve its five fresh pools.
\State From $M_1$ Bell rounds, run certified peeling; reject any visible
precondition failure.  Recover the syndrome signs immediately on
$tM_{\mathrm{sgn}}$ fresh copies (zero when $t=0$).
\State From $M_2$ fresh peeled Bell rounds, run rank-guided recovery at
$\theta$ to obtain the learner-visible sectors $\mathfrak G$.
\State Estimate the required signed cumulants through order $d$ on at most
$N_{\mathrm{grp}}$ fresh copies, group the sectors, and compute
$U_{\mathrm{rec}}$, the disjoint $J_C$, and $J_{\mathrm{aux}}$ by
simultaneous localization.
\State On the final fresh pool, jointly tomograph the disjoint $J_C$ by
Proposition~\ref{prop:supp-block-product-bruteforce-tomography}, using the
transcript-specific portion of the reserved copies and returning
finite-precision physical factors $\widehat\nu_C$.
\State Set $\bar U_{\mathrm{rec}}:=I^{\otimes t}\otimes U_{\mathrm{rec}}$
and $\widehat\rho_{\mathrm{loc}}:=|\widehat b\rangle\!\langle\widehat b|
\otimes(\bigotimes_C\widehat\nu_C)\otimes
I_{J_{\mathrm{aux}}}/2^{|J_{\mathrm{aux}}|}$.
\State Output $\widehat\rho_{\mathrm{bp}}
:=
U_{\mathrm{stab}}\bar U_{\mathrm{rec}}
\widehat\rho_{\mathrm{loc}}
\bar U_{\mathrm{rec}}^\dagger U_{\mathrm{stab}}^\dagger$,
together with its compact description
Eq.~\eqref{eq:supp-abstract-output-tuple}.
\end{algorithmic}
\end{inlinealgorithm}

\subsection{Abstract End-to-End Guarantee}

Use distinct failure budgets
\(\zeta_{\mathrm{peel}},\zeta_{\mathrm{rank}},
\zeta_{\mathrm{grp}},\zeta_{\mathrm{sgn}},\zeta_{\mathrm{tom}}\).
Set \(\zeta_{\mathrm{BS}}:=\zeta_{\mathrm{peel}}\).
Write \(\widehat K=|\widehat{\mathfrak C}|\) for the empirical-group
count and \(k_C=|J_C|\) for each empirical-register dimension.
Use the structural event $\mathcal E_{\mathrm{str}}$ from
Sec.~\ref{sec:supp-localization-interface} and define the ledgers
\begin{align}
\zeta_{\mathrm{str}}
&:=\zeta_{\mathrm{peel}}+\zeta_{\mathrm{rank}}+\zeta_{\mathrm{grp}},
&N_{\mathrm{str}}&:=2M_1+2M_2+N_{\mathrm{grp}},
\label{eq:supp-end-to-end-structural-ledger}\\
\zeta_{\mathrm{cond}}
&:=\zeta_{\mathrm{sgn}}+\zeta_{\mathrm{tom}},
&N_{\mathrm{cond}}^{\mathrm{real}}&:=tM_{\mathrm{sgn}}+N_{\mathrm{bp}},
\label{eq:supp-end-to-end-conditional-ledger}\\
N_{\mathrm{tot}}^{\mathrm{real}}
&:=N_{\mathrm{str}}+N_{\mathrm{cond}}^{\mathrm{real}}
=2M_1+2M_2+N_{\mathrm{grp}}+tM_{\mathrm{sgn}}+N_{\mathrm{bp}}.
\label{eq:supp-block-product-main-copy-count}
\end{align}
Here \(M_1\) and \(M_2\) count Bell-measurement rounds, while \(2M_1\)
and \(2M_2\) count state copies.
When \(\widehat K\ge1\), the abstract theorem uses the uniform local budgets
\(\varepsilon_C=\varepsilon_{\mathrm{tom}}/\widehat K\) and
\(\zeta_C=\zeta_{\mathrm{tom}}/\widehat K\).

\begin{theorem}[Abstract block-product tomography guarantee]
\label{thm:supp-abstract-block-product-tomography}
Let $\rho$ be an $n$-qubit CEBP state with known valid hidden-block-size
cap $2\le d\le n$.  Set $\ell_{\mathrm{grp}}=d$ and run
Algorithm~\ref{alg:supp-block-product-main-tomography}.

Assume the following stage conditions.
\begin{enumerate}
\item The initial Bell schedule satisfies
Proposition~\ref{prop:supp-certified-blockwise-peeling-span} with budget
$\zeta_{\mathrm{peel}}$.

\item The fresh residual Bell schedule satisfies
Proposition~\ref{prop:supp-uniform-concentration} with budget
$\zeta_{\mathrm{rank}}$ and tolerance
$\tau_{\mathrm{rank}}=\tau(M_2,n,\zeta_{\mathrm{rank}})$.  For every
$h\in\mathsf H_{\mathrm{peel}}$, with $\lambda_h:=1-h$,
\[
\theta-\tau_{\mathrm{rank}}>\lambda_h,
\qquad
\lambda_h(\theta-\tau_{\mathrm{rank}})
-2\tau_{\mathrm{rank}}>0.
\]
Thus Proposition~\ref{prop:supp-block-triple-single-block} applies on that
uniform event.

\item The applicable uniform grouping event has budget
$\zeta_{\mathrm{grp}}$ and satisfies the hypotheses of
Proposition~\ref{prop:supp-hierarchical-cumulant-correctness} or
Proposition~\ref{prop:supp-guessed-scale-overrefined-grouping} for the
selected exact or guessed-scale branch, respectively.

\item For every resulting structural transcript, the remaining schedules
satisfy Propositions~\ref{prop:supp-peeled-syndrome-sign-recovery}
and~\ref{prop:supp-block-product-bruteforce-tomography}, the budgets
\eqref{eq:supp-empirical-tomography-local-budgets}, and the numerical
allowance \eqref{eq:supp-empirical-projection-numerical-precision}.
\end{enumerate}

Then, with probability at least
$1-\zeta_{\mathrm{str}}-\zeta_{\mathrm{cond}}$, the algorithm does not
reject and outputs the compact physical description
\eqref{eq:supp-abstract-output-tuple}, which specifies a physical estimator
$\widehat\rho_{\mathrm{bp}}$
satisfying
\begin{equation}
\left\|\widehat\rho_{\mathrm{bp}}-\rho\right\|_1
\le
(\widehat K+1)E_{\mathrm{struct}}^{\mathrm{cert}}
+\varepsilon_{\mathrm{tom}}.
\label{eq:supp-block-product-total-tomography-error}
\end{equation}
Its realized copy count is Eq.~\eqref{eq:supp-block-product-main-copy-count}.
\end{theorem}

\begin{proof}
Apply the end-to-end sampling convention sequentially.  The first condition
gives $\mathcal E_{\mathrm{BS}}$ with failure
$\zeta_{\mathrm{peel}}$ and a certified peeling transcript.  Conditional on
that transcript, the second gives $\mathcal E_{\mathrm{rank}}$ with failure
$\zeta_{\mathrm{rank}}$ and, by
Proposition~\ref{prop:supp-block-triple-single-block}, complete block-pure
visible sectors.  Conditional once more, the third gives
$\mathcal E_{\mathrm{grp}}$ with failure $\zeta_{\mathrm{grp}}$ and the
applicable grouping conclusion.  Proposition
\ref{prop:supp-simultaneous-empirical-localization},
Proposition~\ref{prop:supp-missed-pauli-threshold-certificate}, and
Theorem~\ref{thm:supp-empirical-block-product-approximation} then produce
the structural handoff and
$\|\rho_{\mathrm{loc}}-\Omega_{\mathrm{emp}}\|_1
\le E_{\mathrm{struct}}^{\mathrm{cert}}$.
Partial-trace contractivity and tensor-product telescoping imply
\begin{equation}
\|\nu_C-\omega_C\|_1\le E_{\mathrm{struct}}^{\mathrm{cert}},
\qquad
\left\|\bigotimes_C\nu_C-\bigotimes_C\omega_C\right\|_1
\le\widehat K E_{\mathrm{struct}}^{\mathrm{cert}}.
\end{equation}
Thus comparison with $\rho_{\mathrm{loc}}$ costs
$(\widehat K+1)E_{\mathrm{struct}}^{\mathrm{cert}}$; this generic
multiplier is necessary because omitted hidden-parent directions and
radical-completion coordinates need not be recovered directions.

Conditional on the fixed handoff, the fresh syndrome and tomography
schedules fail with probability at most
$\zeta_{\mathrm{cond}}=\zeta_{\mathrm{sgn}}+\zeta_{\mathrm{tom}}$
(with no syndrome charge when $t=0$).  On their success events,
$\widehat b=b$ and
Eq.~\eqref{eq:supp-block-product-tomography-error} adds at most
$\varepsilon_{\mathrm{tom}}$, proving the localized error in the theorem.
The localized target and returned estimator satisfy
\begin{equation}
\rho
=
U_{\mathrm{stab}}\bar U_{\mathrm{rec}}
\rho_{\mathrm{loc}}
\bar U_{\mathrm{rec}}^\dagger U_{\mathrm{stab}}^\dagger,
\qquad
\widehat\rho_{\mathrm{bp}}
:=
U_{\mathrm{stab}}\bar U_{\mathrm{rec}}
\widehat\rho_{\mathrm{loc}}
\bar U_{\mathrm{rec}}^\dagger U_{\mathrm{stab}}^\dagger.
\end{equation}
Unitary invariance therefore transfers the conditional bound directly to
Eq.~\eqref{eq:supp-block-product-total-tomography-error}.

The sequential union bound charges the five budgets once, giving
$1-\zeta_{\mathrm{str}}-\zeta_{\mathrm{cond}}$.  The Bell-round convention
and the three remaining stage counts give
Eq.~\eqref{eq:supp-block-product-main-copy-count}.  For $t=0$, the sign term
and its failure charge are absent.
\end{proof}

A literal implementation includes exhaustive all-Pauli Bell-record work and
the $n^{O(d)}$ grouping and localized postprocessing quantified above; its
classical runtime is not optimized.  The output remains factorized, and the
state-copy cost of disjoint-register tomography is the maximum local schedule,
not the sum.

\subsection{Calibrated Parameter Schedule}
\label{sec:supp-calibrated-schedule}

\subsubsection{Learner-Known Calibration}
\label{sec:supp-learner-known-calibration}

We now give a baseline schedule using only \(n,\varepsilon,\delta\), and
the known cap \(d\).  Replacing a supplied cap by \(\min\{d,n\}\) if
necessary, this generic schedule assumes \(2\le d\le n\); the complete
\(d=1\) parameter range is handled separately in
Sec.~\ref{sec:supp-single-qubit-hidden-block-specialization}.  The grouping order is
\(\ell_{\mathrm{grp}}:=d\).  Thus the SM symbol \(d\) is the main-text
prior \(d\).  The dependency order below is: a learner-known
structural multiplier, then the target certificate allocation, grouping
scale, peeling floor, and finally the Bell tolerances and copy counts.

Recall that \(\Gamma_d\) is the ordered-Bell cumulant perturbation constant
defined in Eq.~\eqref{eq:supp-cumulant-gamma-s}.  Recall also from
Eq.~\eqref{eq:supp-split-cumulant-factor} that
\(F_d(x)=(1+x)^{2^d-d-1}-1\).  The upstream geometry gives
\begin{equation}
\widehat K\le L\le n,\qquad
k_C\le d,\qquad
\sum_{C\subseteq\mathcal B_a^\star}k_C\le L_a\le d.
\end{equation}
Thus the learner-known bounds
\begin{equation}
R_{\mathrm{ub}}:=n+1,
\qquad
A_d:=n2^d,
\qquad
q_d:=2^d-d-1
\label{eq:supp-calibrated-known-bounds}
\end{equation}
have the following separate roles.  The quantity \(R_{\mathrm{ub}}\) bounds
\(\widehat K+1\).  The quantity \(A_d\) bounds the missed-mass and
computable splitting prefactors; in particular,
\(\sum_a2^{L_a}\le n2^d=A_d\).  The positive quantity \(q_d\) is the
exponent in \(F_d\).  Thus no hidden partition statistic is an input.

Set
\begin{align}
\varepsilon_{\mathrm{tom}}
&:=\frac{\varepsilon}{4},
&
\theta_0
&:=
\frac{\varepsilon^2}
{128R_{\mathrm{ub}}^2A_d^2},
\label{eq:supp-calibrated-accuracy-scales}\\
\eta_s
&:=
\left(
1+\frac{\varepsilon}{8R_{\mathrm{ub}}A_d}
\right)^{1/q_d}-1,
&
\lambda_0
&:=
\min\left\{
\frac{\theta_0}{4},
\frac{\varepsilon^2}{128nR_{\mathrm{ub}}^2},
\frac{\eta_s^2}{512n\Gamma_d^2}
\right\}.
\label{eq:supp-calibrated-grouping-peeling-scales}
\end{align}
Choose
\begin{equation}
\theta:=\theta_0,\qquad
\eta_{\mathrm{test}}:=\frac{\eta_s}{2},\qquad
\tau_\kappa:=\frac{\eta_s}{4},
\label{eq:supp-calibrated-thresholds}
\end{equation}
and set the signed subset-moment tolerance to
$\tau_\mu
:=
\frac{\tau_\kappa}{\Gamma_d}
=
\frac{\eta_s}{4\Gamma_d}$.  Also set
$
\mathsf I_h
:=
\left[1-\lambda_0,1-\frac{\lambda_0}{2}\right],
\qquad
\Delta_h:=\frac{\lambda_0}{2}
$,
with the grid of
Eq.~\eqref{eq:supp-peeling-threshold-grid} constructed at target mesh
$\eta_{\mathrm{grid}}^{\mathrm{target}}
:=
\frac{\Delta_h}{4(2n+1)}$.
Allocate
\begin{equation}
\zeta_{\mathrm{peel}}
=\zeta_{\mathrm{rank}}
=\zeta_{\mathrm{grp}}
=\zeta_{\mathrm{sgn}}
=\zeta_{\mathrm{tom}}
:=
\frac{\delta}{5}.
\label{eq:supp-calibrated-failure-allocation}
\end{equation}
The numerical projection allowance is already included in
\(\varepsilon_{\mathrm{tom}}\) through the local
\(\varepsilon_C\)'s; it is not added again.

\subsubsection{Stagewise Sample Schedules}
\label{sec:supp-calibrated-stagewise-schedules}

The two Bell-round counts may be chosen as
\begin{align}
M_1
&:=
\left\lceil
\frac{512(2n+1)^2}{\lambda_0^2}
\log\frac{2\cdot4^n}{\zeta_{\mathrm{peel}}}
\right\rceil,
\label{eq:supp-calibrated-first-bell-rounds}\\
M_2
&:=
\left\lceil
\frac{128}{\lambda_0^2\theta_0^2}
\log\frac{2\cdot4^n}{\zeta_{\mathrm{rank}}}
\right\rceil.
\label{eq:supp-calibrated-second-bell-rounds}
\end{align}
Their tolerances are
\(\tau_1=\tau(M_1,n,\zeta_{\mathrm{peel}})\) and
\(\tau_{\mathrm{rank}}=\tau(M_2,n,\zeta_{\mathrm{rank}})\).
For \(h_{\min}=1-\lambda_0\), choose
\begin{equation}
M_{\mathrm{sgn}}
:=
\left\lceil
\frac{2}{h_{\min}}
\log\frac{2n}{\zeta_{\mathrm{sgn}}}
\right\rceil.
\label{eq:supp-calibrated-sign-count}
\end{equation}
This number is used only when \(t\ge1\).

Reserve the learner-known test bound
\begin{equation}
N_{\mathrm{test}}^{\mathrm{ub}}
:=
n4^d
\sum_{q=2}^{d}\binom nq
\label{eq:supp-calibrated-test-bound}
\end{equation}
and take
\begin{equation}
N_{\mathrm{grp}}
:=
N_{\mathrm{test}}^{\mathrm{ub}}
\left\lceil
\frac{32\Gamma_d^2}{\eta_s^2}
\log\!\left(
\frac{2(2^d-1)N_{\mathrm{test}}^{\mathrm{ub}}}
{\zeta_{\mathrm{grp}}}
\right)
\right\rceil.
\label{eq:supp-calibrated-grouping-copies}
\end{equation}
For every realized tuple, the ordinary estimator targets signed subset
moments at tolerance \(\tau_\mu\); the perturbation bound
Eq.~\eqref{eq:supp-bounded-moment-cumulant-perturbation} then gives
cumulant error at most \(\tau_\kappa\).  Under the end-to-end sampling
convention, Eq.~\eqref{eq:supp-ordinary-grouping-generic-copies} gives
\(
\Pr(\mathcal E_{\mathrm{grp}}\mid\text{fixed preceding successful
transcript})\ge1-\zeta_{\mathrm{grp}}
\).
After \(\widehat{\mathfrak C}\) is known, if
\(\widehat K\ge1\), set
$
\varepsilon_C:=\frac{\varepsilon_{\mathrm{tom}}}{\widehat K},
\qquad
\zeta_C:=\frac{\zeta_{\mathrm{tom}}}{\widehat K}
$
and use Proposition~\ref{prop:supp-block-product-bruteforce-tomography},
whose disjoint-register schedule defines $N_{\mathrm{bp}}$ by the maximum in
Eq.~\eqref{eq:supp-block-product-tomography-general-cost}.  If
\(\widehat K=0\), set \(N_{\mathrm{bp}}:=0\).
For deterministic pre-reservation, define
\begin{align}
N_{\mathrm P}^{\mathrm{wc}}
&:=4^d-1,
&
\tau_{\mathrm{bp}}^{\mathrm{wc}}
&:=\frac{\varepsilon_{\mathrm{tom}}}
{4n\sqrt{N_{\mathrm P}^{\mathrm{wc}}}},
\nonumber\\
M_{\mathrm P}^{\mathrm{wc}}
&:=\left\lceil
\frac{2}{(\tau_{\mathrm{bp}}^{\mathrm{wc}})^2}
\log\frac{2nN_{\mathrm P}^{\mathrm{wc}}}{\zeta_{\mathrm{tom}}}
\right\rceil,
&
N_{\mathrm{bp}}^{\mathrm{wc}}
&:=N_{\mathrm P}^{\mathrm{wc}}M_{\mathrm P}^{\mathrm{wc}}.
\end{align}
Then $N_{\mathrm{bp}}\le N_{\mathrm{bp}}^{\mathrm{wc}}$ by
$\widehat K\le n$ and $k_C\le d$.  The learner reserves
\begin{equation}
N_{\mathrm{tot}}^{\mathrm{bp}}
:=N_{\mathrm{tot}}^{\mathrm{wc}}
:=2M_1+2M_2+N_{\mathrm{grp}}+nM_{\mathrm{sgn}}
+N_{\mathrm{bp}}^{\mathrm{wc}}
\label{eq:supp-block-product-main-copy-count-worst-case}
\end{equation}
copies before sampling and consumes only the transcript-specific
$N_{\mathrm{tot}}^{\mathrm{real}}\le N_{\mathrm{tot}}^{\mathrm{wc}}$.

\subsubsection{Correctness of the Calibrated Schedule}
\label{sec:supp-calibrated-correctness}

\begin{proposition}[Calibrated end-to-end block-product tomography]
\label{prop:supp-calibrated-block-product-tomography}
Let $2\le d\le n$.  For every \(n\)-qubit CEBP state satisfying
$L_a\le d$ for every hidden block $a$,
the preceding schedule with \(\ell_{\mathrm{grp}}=d\) is executable from
the learner-known inputs \(n,\varepsilon,\delta,d\).  With probability at
least \(1-\delta\), it outputs the compact physical tuple in
Eq.~\eqref{eq:supp-abstract-output-tuple}, which specifies
\(\widehat\rho_{\mathrm{bp}}\) without a dense global matrix, and
\begin{equation}
\left\|\widehat\rho_{\mathrm{bp}}-\rho\right\|_1
\le\varepsilon.
\label{eq:supp-calibrated-block-product-accuracy}
\end{equation}
For the ordinary adaptive grouping implementation,
\begin{align}
N_{\mathrm{tot}}^{\mathrm{bp}}
=\widetilde O\!\Bigg(
\frac{
n^3\Gamma_d^4R_{\mathrm{ub}}^8A_d^8
q_d^{\,4}
}{\varepsilon^8}
+
\frac{
d^2R_{\mathrm{ub}}^2n^{d+3}4^d
q_d^{\,2}
}{\varepsilon^2}
\left(
\frac{4d^2}{e^2(\log 2)^2}
\right)^d
+
\frac{n^2 16^d}{\varepsilon^2}
+n
\Bigg),
\label{eq:supp-calibrated-block-product-copy-count}
\end{align}
Here \(\widetilde O\) suppresses logarithms in the displayed parameters and
\(\delta^{-1}\).  In particular,
\begin{equation}
N_{\mathrm{tot}}^{\mathrm{bp}}
=
\widetilde O\!\left[
2^{O(d\log d)}
\left(
\frac{n^{19}}{\varepsilon^8}
+
\frac{n^{d+5}}{\varepsilon^2}
\right)
\right],
\label{eq:supp-calibrated-block-product-copy-count-simplified}
\end{equation}
which for fixed $d\ge2$ is
$\widetilde O(n^{19}\varepsilon^{-8}+n^{d+5}\varepsilon^{-2})$.
\end{proposition}

\begin{proof}
\leavevmode
\paragraph{Feasibility and ranking/grouping margins.}
The choices of $M_1$ and the target mesh give
$\tau_1\le\lambda_0/[16(2n+1)]
=\Delta_h/[8(2n+1)]$, so
Proposition~\ref{prop:supp-certified-blockwise-peeling-span} applies.
Moreover, $\lambda_0\le\theta_0/4<1/2$, and its accepted value obeys
\begin{equation}
\frac{\lambda_0}{2}\le\lambda\le\lambda_0.
\label{eq:supp-calibrated-accepted-lambda-range}
\end{equation}

The choice of $M_2$ gives
\begin{equation}
\tau_{\mathrm{rank}}\le\frac{\lambda_0\theta_0}{8},
\qquad
\theta_0-\tau_{\mathrm{rank}}
\ge\frac{31}{32}\theta_0
>\frac{\theta_0}{4}\ge\lambda.
\end{equation}
Together with the lower bound in
Eq.~\eqref{eq:supp-calibrated-accepted-lambda-range},
\begin{align*}
\lambda(\theta_0-\tau_{\mathrm{rank}})
-2\tau_{\mathrm{rank}}
&\ge
\frac{\lambda_0}{2}(\theta_0-\tau_{\mathrm{rank}})
-\frac{\lambda_0\theta_0}{4}
=\frac{\lambda_0\theta_0}{4}
-\frac{\lambda_0\tau_{\mathrm{rank}}}{2}
>0.
\end{align*}
Thus both rank-recovery margins hold.

Certified peeling and the third branch of $\lambda_0$ give
\begin{align}
\varepsilon_{\mathrm{peel}}
&=\frac{t\lambda}{2}
\le\frac{n\lambda_0}{2}
\le
\frac{\varepsilon^2}{256R_{\mathrm{ub}}^2}
<1,
\nonumber\\
\beta_{\mathrm{peel}}
&=4\Gamma_d\sqrt{\varepsilon_{\mathrm{peel}}}
\le\frac{\eta_s}{8}
<\frac{\eta_s}{4}.
\end{align}
Because \(\ell_{\mathrm{grp}}=d\) and every true recovered-sector block has
at most \(d\) sectors, the grouping-order promise holds; the last display
and
Eq.~\eqref{eq:supp-calibrated-thresholds} verify the guessed-scale
no-false-merge conditions.

\paragraph{Accuracy allocation.}
Set $\theta_{\mathrm{rec}}:=\theta_0+\tau_{\mathrm{rank}}$.  The peeling
bound, $\tau_{\mathrm{rank}}\le\theta_0$, and
Eq.~\eqref{eq:supp-global-missed-mass-computable-certificate} give
\begin{align}
2\sqrt{\varepsilon_{\mathrm{peel}}}
&\le\frac{\varepsilon}{8R_{\mathrm{ub}}},
\label{eq:supp-calibrated-peeling-allocation}\\
n2^d\sqrt{\theta_{\mathrm{rec}}}
&\le A_d\sqrt{2\theta_0}
=\frac{\varepsilon}{8R_{\mathrm{ub}}}.
\label{eq:supp-calibrated-missed-allocation}
\end{align}
In the guessed-scale branch,
\(\xi_{\mathrm{eff}}<\eta_s\), \(\widehat K\le n\), and \(k_C\le d\).
Therefore
\begin{align}
E_{\mathrm{split}}^{\mathrm{cert}}
&\le F_d(\eta_s)\sum_{C\in\widehat{\mathfrak C}}2^d
\le A_dF_d(\eta_s)
\notag\\
&= A_d\left[(1+\eta_s)^{q_d}-1\right]
=\frac{\varepsilon}{8R_{\mathrm{ub}}}.
\label{eq:supp-calibrated-split-allocation}
\end{align}
Equations
\eqref{eq:supp-calibrated-peeling-allocation}--\eqref{eq:supp-calibrated-split-allocation}
and
Eq.~\eqref{eq:supp-localization-structural-certificate} yield
$
E_{\mathrm{struct}}^{\mathrm{cert}}
\le
\frac{3\varepsilon}{8R_{\mathrm{ub}}}
$.
The empirical-register size bound \(k_C\le d\) verifies the tomography
dimension premise.  Its local numerical projection error is already
included in each \(\varepsilon_C\).  Since
\(\widehat K+1\le R_{\mathrm{ub}}\), the abstract theorem gives
\begin{equation}
(\widehat K+1)E_{\mathrm{struct}}^{\mathrm{cert}}
+
\varepsilon_{\mathrm{tom}}
\le
\frac{3\varepsilon}{8}+\frac{\varepsilon}{4}
=\frac{5\varepsilon}{8}
<\varepsilon,
\end{equation}
leaving explicit accuracy slack.

\paragraph{Five-stage failure accounting.}
Under the sampling convention of Sec.~\ref{sec:supp-end-to-end}, the peeling,
rank, sign, and tomography stages invoke
Propositions~\ref{prop:supp-certified-blockwise-peeling-span},
\ref{prop:supp-block-triple-single-block},
\ref{prop:supp-peeled-syndrome-sign-recovery}, and
\ref{prop:supp-block-product-bruteforce-tomography}; grouping invokes
Proposition~\ref{prop:supp-guessed-scale-overrefined-grouping}.  The
conditional failure budgets are the five
quantities in
Eq.~\eqref{eq:supp-calibrated-failure-allocation}; a sequential union bound
sums them to $\delta$.  For $t=0$, the sign event and its copies are absent.

\paragraph{Resource substitution and simplification.}
For
\(x:=\varepsilon/(8R_{\mathrm{ub}}A_d)\le1\), the definition of
\(\eta_s\) gives
\(\eta_s^{-p}=O((R_{\mathrm{ub}}A_dq_d/\varepsilon)^p)\)
for \(p=2,4\), while
\(\theta_0^{-1}=O(R_{\mathrm{ub}}^2A_d^2/\varepsilon^2)\).
Substitution in the three defining branches of $\lambda_0$ gives
\begin{equation}
\lambda_0^{-2}
=
O\!\left(
\frac{
n^2\Gamma_d^4R_{\mathrm{ub}}^4A_d^4
q_d^{\,4}
}{\varepsilon^4}
\right).
\label{eq:supp-calibrated-lambda-inverse-bound}
\end{equation}
Keeping the uniform-all-Pauli logarithms visible, the two Bell pools obey
\begin{align}
M_1
&=
O\!\left(
n^2\lambda_0^{-2}
\left[n+\log\frac1{\zeta_{\mathrm{peel}}}\right]
\right),
\label{eq:supp-calibrated-first-bell-audit}\\
M_2
&=
O\!\left(
\lambda_0^{-2}\theta_0^{-2}
\left[n+\log\frac1{\zeta_{\mathrm{rank}}}\right]
\right).
\label{eq:supp-calibrated-second-bell-audit}
\end{align}
Equations~\eqref{eq:supp-calibrated-lambda-inverse-bound} and the two audits
give
\begin{equation}
M_2
=
\widetilde O\!\left(
\frac{
n^3\Gamma_d^4R_{\mathrm{ub}}^8A_d^8
q_d^{\,4}
}{\varepsilon^8}
\right).
\end{equation}
This term dominates $M_1$ for the stated parameter range and gives the first
line of
Eq.~\eqref{eq:supp-calibrated-block-product-copy-count}.  Substituting the bound on
\(\eta_s^{-2}\) into the ordinary generic grouping count gives the second
line.  More explicitly, for $\ell_{\mathrm{grp}}=d$,
$\tau_\kappa=\eta_s/4$, and failure budget
$\delta_{\mathrm{grp}}^{\mathrm{ordinary}}$, the bounds
$N_{\mathrm{test}}^{\max}=O(dn^{d+1}4^d)$ and
$\Gamma_d=O(\sqrt d(d/(e\log2))^d)$ give
\begin{equation}
N_{\mathrm{copy}}^{\mathrm{ordinary}}
=O\!\left[
\frac{d^2n^{d+1}}{\eta_s^2}
\left(\frac{4d^2}{e^2(\log2)^2}\right)^d
\log\!\left(
\frac{d\,n^{d+1}8^d}{\delta_{\mathrm{grp}}^{\mathrm{ordinary}}}
\right)
\right].
\label{eq:supp-copy-count-emp-cumulant-asymptotic}
\end{equation}
This gives the second line directly.  For tomography, substitute
$\widehat K\le n$, $k_C\le d$, and the
uniform local budgets into
Eq.~\eqref{eq:supp-block-product-tomography-max-block-cost}.  The joint
maximum schedule gives
\begin{equation}
N_{\mathrm{bp}}
\le N_{\mathrm{bp}}^{\mathrm{wc}}
=
O\!\left(
\frac{n^2 16^d}{\varepsilon_{\mathrm{tom}}^2}
\log\frac{n4^d}{\zeta_{\mathrm{tom}}}
\right).
\end{equation}
Together with \(tM_{\mathrm{sgn}}\le nM_{\mathrm{sgn}}\) and
Eq.~\eqref{eq:supp-calibrated-sign-count}, this gives the last line.
Adding the five disjoint pools proves both
\(N_{\mathrm{tot}}^{\mathrm{real}}\le N_{\mathrm{tot}}^{\mathrm{wc}}\)
and the displayed deterministic complexity.
Finally,
$R_{\mathrm{ub}}=O(n)$, $A_d=n2^d$, $q_d\le2^d$, and
$\Gamma_d=2^{O(d\log d)}$.  The Bell and ordinary-grouping terms in
Eq.~\eqref{eq:supp-calibrated-block-product-copy-count} therefore become
\begin{equation}
\frac{n^{19}}{\varepsilon^8}2^{O(d\log d)}
\qquad\text{and}\qquad
\frac{n^{d+5}}{\varepsilon^2}2^{O(d\log d)}.
\end{equation}
For $d\ge2$, $n\ge d$, and $0<\varepsilon<1$, the tomography and syndrome
terms are lower order and are absorbed by these two contributions, proving
Eq.~\eqref{eq:supp-calibrated-block-product-copy-count-simplified} with the
same logarithmic convention.
\end{proof}

\subsection{Single-Qubit Hidden-Block Specialization}
\label{sec:supp-single-qubit-hidden-block-specialization}

For \(d=1\), take every recovered sector as a singleton group.
Proposition~\ref{prop:supp-block-triple-single-block} makes each recovered
sector block-pure; distinct recovered sectors commute elementwise and their
recovered Pauli directions are linearly independent.  If two distinct sectors belonged
to one hidden qubit, they would supply two independent nonidentity one-qubit
Paulis, which anticommute, whereas Clifford conjugation preserves their
commutator---a contradiction.  Hence at most one sector belongs to each
hidden qubit, simultaneous localization assigns every sector a one-qubit
register \(J_C\), and
\[
\widehat K\le n,
\qquad
N_{\mathrm{grp}}=0,
\qquad
E_{\mathrm{split}}^{\mathrm{cert}}=0.
\]
Thus the only algorithmic distinction from the generic pipeline is that
no grouping is needed.  The same unconditional empirical-register
tomography is used after localization; the analysis below gives a sharper
\(d=1\) error bound and calibration.

Fix a successful nontrivial peeling-and-recovery transcript with
\(\varepsilon_{\mathrm{peel}}<1\), and define
$
p_b:=\Tr\!\left[\rho_{\mathrm{loc}}\left(
|b\rangle\!\langle b|_{[t]}\otimes I_{[m]}\right)\right]
$.
Equation~\eqref{eq:supp-positive-compressed-syndrome-weight} gives
\(p_b\ge1-\varepsilon_{\mathrm{peel}}>0\) on this branch.
For every represented block \(J_a=J_C\), while unrepresented directions
lie in \(J_{\mathrm{aux}}\).  Hence
Proposition~\ref{prop:supp-localized-true-block-product} gives
\begin{equation}
\Omega_{\mathrm{true}}=|b\rangle\!\langle b|_{[t]}\otimes
\bigotimes_{C\in\widehat{\mathfrak C}}\omega_C\otimes
\frac{I_{J_{\mathrm{aux}}}}{2^{|J_{\mathrm{aux}}|}},
\qquad \omega_C\in\D(\mathbb C^2).
\end{equation}
Let \(\nu_C\) be the unconditional localized marginal from
Eq.~\eqref{eq:supp-actual-empirical-marginals}, and set
\(\theta_{\mathrm{rec}}:=\theta+\tau_{\mathrm{rank}}\).
For each nonidentity one-qubit Pauli \(P\) on \(J_C\), define
\begin{equation}
m_{C,P}:=\Tr\!\left[\rho_{\mathrm{loc}}\left(
I_{[t]}\otimes P_{J_C}\otimes I_{[m]\setminus J_C}\right)\right]
=\Tr(\nu_C P).
\label{eq:supp-single-qubit-unconditional-moment}
\end{equation}
There are two cases, since a recovered sector can be a completed triple
or a single retained axis.

First suppose that
\(Q:=U_{\mathrm{rec}}(P_{J_C}\otimes I)U_{\mathrm{rec}}^\dagger\)
belongs to \(\mathcal P(C)\).  Then \(Q\) has a retained lift
\(Z^{c_0}\otimes Q\).  In the Fourier expansion
\eqref{eq:supp-phase-safe-syndrome-projector}, every component factors as
\(Z^c\otimes Q=(Z^{c+c_0}\otimes I)(Z^{c_0}\otimes Q)\).
The complete peeling subgroup is retained blockwise and the retained
hidden-local subgroup is closed under its multiplication.  Truncation
therefore preserves the syndrome-projected moment and, by
Eq.~\eqref{eq:supp-full-peeling-subgroup-moment-preservation}, its weight
\(p_b\).  The construction of \(\Omega_{\mathrm{true}}\) consequently gives
\begin{equation}
\begin{aligned}
m_{C,P}&=p_b\Tr(\omega_C P)+r_{C,P},\\
|r_{C,P}|&\le1-p_b,\\
|m_{C,P}-\Tr(\omega_C P)|&\le2(1-p_b)
\le2\varepsilon_{\mathrm{peel}}.
\end{aligned}
\end{equation}
Here \(r_{C,P}\) is the expectation of \(P\) on the complementary
syndrome subspace: the corresponding positive, subnormalized marginal
has trace \(1-p_b\), and \(\|P\|_\infty=1\).

If \(Q\notin\mathcal P(C)\), it is outside the full recovered span
\(V(\mathfrak G)\): the other groups are localized on disjoint registers.
Threshold-span completeness
\eqref{eq:supp-threshold-span-completeness} and the recovery concentration
event \eqref{eq:supp-block-uniform-score-event}, applied to the
prefix-identity lift of \(Q\), imply
\begin{equation}
|m_{C,P}|^2\le\theta+\tau_{\mathrm{rank}}
=\theta_{\mathrm{rec}},
\qquad
\Tr(\omega_C P)=0.
\end{equation}
The second equality follows because \(\omega_C\) has Pauli support only
in \(U_{\mathrm{rec}}^\dagger\mathcal P(C)U_{\mathrm{rec}}\).
These unretained coordinates must still be measured in unconditional
tomography; their bias need not vanish when
\(\varepsilon_{\mathrm{peel}}=0\).

For one-qubit states, the trace norm of their difference equals the
Euclidean distance between their Bloch vectors.  Splitting the at most
three coordinates into the two cases above gives
\begin{equation}
\|\nu_C-\omega_C\|_1
\le2\sqrt3\,\varepsilon_{\mathrm{peel}}
+\sqrt{3\theta_{\mathrm{rec}}}.
\end{equation}
Tensor-product telescoping over \(\widehat K\le n\) therefore yields
\begin{equation}
\left\|\bigotimes_C\nu_C-\bigotimes_C\omega_C\right\|_1
\le2\sqrt3\,n\varepsilon_{\mathrm{peel}}
+n\sqrt{3\theta_{\mathrm{rec}}}.
\end{equation}

Operationally, the learner applies the two known inverse Cliffords and
measures all three Pauli coordinates on every \(J_C\), using the same
unconditional disjoint-register schedule as
Proposition~\ref{prop:supp-block-product-bruteforce-tomography} with
\(k_C=1\).  Every tomography outcome is used, independently of syndrome
outcomes.  For \(\widehat K\ge1\), use
\(\varepsilon_C=\varepsilon_{\mathrm{tom}}/\widehat K\) and
\(\zeta_C=\zeta_{\mathrm{tom}}/\widehat K\), and project the estimated
Bloch vectors onto the Bloch ball with the numerical allowance
\eqref{eq:supp-empirical-projection-numerical-precision}.  This gives
physical estimates \(\widehat\nu_C\) satisfying
\(\sum_C\|\widehat\nu_C-\nu_C\|_1\le\varepsilon_{\mathrm{tom}}\)
with failure at most \(\zeta_{\mathrm{tom}}\).  Keep the separate syndrome
estimate \(\widehat b\), and assemble
\begin{equation}
\widehat\Omega_{1\mathrm q}
:=|\widehat b\rangle\!\langle\widehat b|_{[t]}
\otimes\bigotimes_C\widehat\nu_C
\otimes\frac{I_{J_{\mathrm{aux}}}}{2^{|J_{\mathrm{aux}}|}}.
\end{equation}
On the successful syndrome and tomography events,
\begin{equation}
\|\widehat\Omega_{1\mathrm q}-\Omega_{\mathrm{true}}\|_1
\le2\sqrt3\,n\varepsilon_{\mathrm{peel}}
+n\sqrt{3\theta_{\mathrm{rec}}}+\varepsilon_{\mathrm{tom}}.
\end{equation}
Write \(N_{1\mathrm q}\) for this \(d=1\) specialization of the same
unconditional local tomography copy count.  The three global Pauli
settings measure disjoint registers jointly, so copies are the maximum
of their local schedules, not their sum.  For \(\widehat K=0\), set
\(N_{1\mathrm q}=0\) and use the empty product with the maximally mixed
auxiliary output.

A subgroup of the one-qubit Pauli group has size \(1\), \(2\), or \(4\), so at most three nonidentity Pauli directions are omitted on each hidden qubit; Proposition~\ref{prop:supp-missed-pauli-threshold-certificate} therefore gives
$
E_{\mathrm{miss}}\le n\sqrt{3\theta_{\mathrm{rec}}}
$.
Define \(\widehat\rho_{1\mathrm q}:=U_{\mathrm{stab}}\bar U_{\mathrm{rec}}\widehat\Omega_{1\mathrm q}\bar U_{\mathrm{rec}}^\dagger U_{\mathrm{stab}}^\dagger\).  The preceding reconstruction, Proposition~\ref{prop:supp-localized-true-block-product}, and unitary invariance yield
\begin{equation}
\|\widehat\rho_{1\mathrm q}-\rho\|_1
\le2\sqrt{\varepsilon_{\mathrm{peel}}}
+2n\sqrt{3\theta_{\mathrm{rec}}}
+2\sqrt3\,n\varepsilon_{\mathrm{peel}}+\varepsilon_{\mathrm{tom}}.
\label{eq:supp-single-qubit-error-bound}
\end{equation}
This direct comparison avoids the generic \((\widehat K+1)\) marginal-product multiplier.

\begin{corollary}[Single-qubit hidden-block tomography]
\label{cor:supp-single-qubit-hidden-block-tomography}
Let \(n\ge1\), \(\varepsilon,\delta\in(0,1)\), and \(d=1\).  With probability at least \(1-\delta\), the specialized pipeline returns a physical compact factorized estimator \(\widehat\rho_{1\mathrm q}\) satisfying \(\|\widehat\rho_{1\mathrm q}-\rho\|_1\le\varepsilon\).  It uses no grouping and
\begin{equation}
N_{\mathrm{tot}}^{(d=1)}=2M_1+2M_2+tM_{\mathrm{sgn}}+N_{1\mathrm q}
=\widetilde O\!\left(\frac{n^9}{\varepsilon^8}\right).
\label{eq:supp-single-qubit-copy-count}
\end{equation}
\end{corollary}

\begin{proof}
Give peeling, rank recovery, syndrome recovery, and unconditional one-qubit
tomography failure budget $\delta/4$ each.  Use
Eqs.~\eqref{eq:supp-calibrated-first-bell-rounds},
\eqref{eq:supp-calibrated-second-bell-rounds}, and
\eqref{eq:supp-calibrated-sign-count} with $(\lambda_0,\theta_0)$ replaced
by $(\lambda_1,\theta_1)$.  The complete specialized calibration is
\begin{equation}
\begin{aligned}
\varepsilon_{\mathrm{tom}}&:=\frac{\varepsilon}{4},
&
\theta=\theta_1&:=\frac{\varepsilon^2}{128n^2},
&
\lambda_1&:=\frac{\theta_1}{4},
\\
h_{\min}&:=1-\lambda_1,
&
\frac{\lambda_1}{2}&\le\lambda\le\lambda_1,
&
\tau_{\mathrm{rank}}&\le\frac{\lambda_1\theta_1}{8}
=\frac{\theta_1^2}{32},
\\
\theta_1-\tau_{\mathrm{rank}}&>\lambda,
&
\lambda(\theta_1-\tau_{\mathrm{rank}})-2\tau_{\mathrm{rank}}&>0,
&
\theta_{\mathrm{rec}}&:=\theta_1+\tau_{\mathrm{rank}}
\le\frac{33}{32}\theta_1,
\\
\varepsilon_{\mathrm{peel}}
&\le\frac{n\lambda_1}{2}
=\frac{\varepsilon^2}{1024n},
&
p_b&\ge1-\varepsilon_{\mathrm{peel}}>\frac12,
\\
2\sqrt{\varepsilon_{\mathrm{peel}}}&\le\frac{\varepsilon}{16},
&
n\sqrt{3\theta_{\mathrm{rec}}}
&\le\frac{\sqrt{99}}{64}\varepsilon<\frac{\varepsilon}{6}.
\end{aligned}
\end{equation}
The first two strict inequalities are the two rank-recovery margins: the
first follows from $\lambda_1=\theta_1/4$ and the displayed tolerance, and
the second follows by the same lower-bound substitution
$\lambda\ge\lambda_1/2$.  The remaining bounds are certified peeling,
the resulting syndrome weight, and the missed-mass certificate.

With \(\varepsilon_{\mathrm{tom}}=\varepsilon/4\) and tomography failure
budget \(\zeta_{\mathrm{tom}}=\delta/4\), specializing
Eq.~\eqref{eq:supp-block-product-tomography-general-cost} to
\(k_C=1\) and \(\widehat K\le n\) gives
\begin{equation}
N_{1\mathrm q}=O\!\left(
\frac{n^2}{\varepsilon^2}\log\frac{16n}{\delta}\right)
=
\widetilde O(n^2\varepsilon^{-2}).
\end{equation}
The retained-coordinate bias obeys
\begin{equation}
2\sqrt3\,n\varepsilon_{\mathrm{peel}}
\le\frac{\sqrt3}{512}\varepsilon^2
\le\frac{\sqrt3}{512}\varepsilon.
\end{equation}
Substitution into Eq.~\eqref{eq:supp-single-qubit-error-bound} gives
\begin{equation}
\|\widehat\rho_{1\mathrm q}-\rho\|_1
\le\left(\frac1{16}+\frac{\sqrt{99}}{32}
+\frac{\sqrt3}{512}+\frac14\right)\varepsilon
<0.627\varepsilon<\varepsilon.
\end{equation}
Sequential conditioning over these four stages gives failure at most
\(\delta\); the syndrome pool and its otherwise-vacuous event are omitted
when $t=0$.

Eqs.~\eqref{eq:supp-calibrated-first-bell-audit} and~\eqref{eq:supp-calibrated-second-bell-audit} give the resource ledger
\begin{equation}
\begin{aligned}
M_1&=\widetilde O\!\left(\frac{n^7}{\varepsilon^4}\right),
&
M_2&=\widetilde O\!\left(\frac{n^9}{\varepsilon^8}\right),
\\
N_{1\mathrm q}&=\widetilde O(n^2\varepsilon^{-2}),
&
tM_{\mathrm{sgn}}&=\widetilde O(n),
&
N_{\mathrm{grp}}&=0.
\end{aligned}
\end{equation}
For \(n\ge1\) and \(\varepsilon\in(0,1)\), \(M_2\) dominates.
\end{proof}
\section{Numerical Simulation}
\label{sec:supp-numerical}

Figure~\ref{fig:cebp_vs_lpli} compares CEBP and LP-LI in physical-copy
complexity.  Bell and local product-Pauli measurements have different experimental
requirements; circuit depth, entangling-gate count, measurement-setting complexity,
classical runtime, and memory are not normalized.

\subsection{Target-State Ensemble and Simulation Settings}

Panel (a) uses $n=10$ and 20 independent targets for each $d=1,2,3$,
optimizing CEBP error separately at the 15 provisioned budgets
\begin{equation}
N\in\{10^3\}\cup\{10^k,5\times10^k:\ k=4,\ldots,10\}.
\label{eq:supp-numerical-cebp-budget-grid}
\end{equation}
Panel (b) seeks a statewise empirical threshold for
$D_{\mathrm{tr}}(\widehat\rho,\rho)=\frac12\|\widehat\rho-\rho\|_1\le0.05$,
using 20 frozen targets per $(n,d)$, with $n=1,\ldots,9$ for $d=1$,
$n=2,\ldots,9$ for $d=2$, and $n=3,\ldots,9$ for $d=3$.

Targets follow Eq.~\eqref{eq:supp-clifford-encoded-block-product}.
As a state-generation convention, not a restriction of the CEBP model, we use
contiguous blocks: all single-qubit blocks for $d=1$, and as many size-$d$
blocks as possible for $d=2,3$, followed by the remainder block when needed.
Independently for each $L_a$-qubit block, a $2^{L_a}\times2^{L_a}$ complex
Ginibre matrix gives $\rho_a=G_aG_a^\dagger/\Tr(G_aG_a^\dagger)$, the
Hilbert--Schmidt ensemble~\cite{ZyczkowskiSommers2001}.  The encoder
$U_{\mathrm c}$ is a nonuniform $5n$-step random walk over Hadamard, phase,
and CNOT gates.  Both methods use the same targets, frozen before tuning or
measurement simulation.  Target generation, structural-candidate sampling,
and measurement randomness are distinct.  Per target and budget in panel (a),
CEBP uses 16 tuning seeds and 17 independent holdout seeds, none used in tuning;
LP-LI uses six measurement replicates.  Panel (b) uses 16 prescribed measurement
seeds per target without additional holdouts.  Batched multinomial counts and
sufficient statistics reproduce shot-experiment statistics and copy accounting.

\subsection{Numerical Implementation and Optimization of the CEBP Protocol}

The learner receives only $n$, the known cap $d$, and measurement access, not the
hidden partition, latent states, or encoder.  It performs Bell peeling,
independent residual Bell sampling and rank-guided recovery, cumulant grouping
for $d\ge2$, simultaneous Clifford localization, syndrome-sign estimation,
localized tomography, physical projection, and global decoding
(Secs.~\ref{sec:supp-stabilizer-peeling}--\ref{sec:supp-end-to-end}).
For $d=1$, grouping is skipped and recovered sectors are treated individually.
Practical execution permits reconstruction under relaxed thresholds and
uncertified analytical margins, including the no-false-merge margin, while
requiring operational and budget feasibility; these runs are not theorem-certified.
Only after reconstruction is the exact simulated target used to evaluate error
and select hyperparameters: reconstruction is oracle-free, tuning oracle-assisted.

Both objectives optimize the practical candidate
\begin{equation}
\boldsymbol{\xi}=(h_{\min},h_{\max},c_\theta,\eta_{\mathrm{test}},
w_{\mathrm{peel}},w_{\mathrm{rec}},w_{\mathrm{grp}},w_{\mathrm{sgn}},w_{\mathrm{tom}}),
\end{equation}
where $c_\theta$ is a dimensionless threshold multiplier and the positive $w$'s
are unnormalized stage  weights.  The search ranges are:
\begin{center}
\begin{tabular}{lccc}
\hline
Parameter & Panel (a) & Panel (b) & Sampling \\
\hline
$h_{\min}$ & $[0.55,0.95]$ & $[0.55,0.95]$ & Uniform \\
$h_{\max}$ & $[0.96,0.999]$ & $[0.60,0.99]$ & Uniform \\
$c_\theta$ & $[1.5,128]$ & $[1.5,128]$ & Log-uniform \\
$\eta_{\mathrm{test}}$ & $[0.02,0.30]$ & $[0.02,0.30]$ & Log-uniform \\
$w_{\mathrm{peel}},w_{\mathrm{rec}},w_{\mathrm{tom}}$
& $[0.20,5.0]$ & $[0.20,5.0]$ & Log-uniform \\
$w_{\mathrm{grp}}$ & $[0.20,20.0]$ & $[0.20,5.0]$ & Log-uniform \\
$w_{\mathrm{sgn}}$ & $[0.05,2.0]$ & $[0.05,2.0]$ & Log-uniform \\
\hline
\end{tabular}
\end{center}
Sampling obeys $\frac12<h_{\min}<h_{\max}<1$; $w_{\mathrm{grp}}$ is inactive
for $d=1$.  Theorem-calibration parameters are not independently optimized.
Five peeling-grid intervals and the allocated second Bell pool determine
\begin{equation}
\eta_{\mathrm{peel}}=\frac{h_{\max}-h_{\min}}{5},\qquad
\theta_{\mathrm{raw}}=c_\theta\tau_{\mathrm{rank}},\qquad
\theta=\min\!\left\{\frac12,\theta_{\mathrm{raw}}\right\},
\end{equation}
with $\tau_{\mathrm{rank}}$ as in Sec.~\ref{sec:supp-rank-guided-recovery}
and empirical grouping threshold $\eta_{\mathrm{test}}$.
For both objectives, $\zeta_{\mathrm{peel}}=\zeta_{\mathrm{rank}}
=\zeta_{\mathrm{sgn}}=0.05$, $\delta_{\mathrm{grp}}=\zeta_{\mathrm{tom}}=0.10$,
and the diagnostic values are $\tau_\kappa=0.05$ and
$\varepsilon_{\mathrm{tom}}=1.0$; these do not constitute the rigorous theorem
calibration of Sec.~\ref{sec:supp-calibrated-schedule}.

The stage weights allocate the provisioned physical-copy budget $N$.
Normalize all five weights for $d\ge2$, or the four active weights with zero
grouping allocation for $d=1$; floor all caps except tomography, which receives
the integer remainder.  Unused capacity is not borrowed downstream.  Since
each Bell round consumes two copies,
\begin{equation}
N_{\mathrm{peel}}=2M_1,\qquad N_{\mathrm{rec}}=2M_2,\qquad
N_{\mathrm{tot}}^{\mathrm{real}}\le N.
\label{eq:supp-numerical-copy-ledgers}
\end{equation}
The realized cost sums the stage consumptions as in
Eq.~\eqref{eq:supp-block-product-main-copy-count}; disjoint recovered registers
$J_C$ share physical rounds for simultaneous local Pauli tomography, rather
than adding their schedule lengths.  Provisioned $N$ is the CEBP $x$-coordinate
in panel (a) and the optimized resource in panel (b).

Both searches use candidate pools $16,32,64,128,256$, seed fidelities
$1,2,4,8,16$, retention fraction $1/2$, at least three progressive rounds, and
improvement patience two.  Panel (a) minimizes the tuning-seed mean trace distance
among fully operational candidates, then freezes the choice before holdout evaluation.
For each structural candidate $\boldsymbol{\xi}$ in panel (b), we search for the smallest
empirically feasible physical-copy budget. Starting from the smallest budget
that assigns nonzero resources to every required stage, we increase $N$
 by a factor of two until the fixed-error criterion is satisfied
or the scientific cap $N=10^9$ is reached, using at most 32 upward
expansions. Once a failing budget $N_{\mathrm{low}}$ and a feasible budget
$N_{\mathrm{high}}$ are identified, we refine the bracket until
\begin{equation}
\frac{N_{\mathrm{high}}-N_{\mathrm{low}}}{N_{\mathrm{high}}}
\le 0.01,
\end{equation}
and take the feasible endpoint $N_{\mathrm{high}}$ as that candidate's
empirical threshold, which is not a proven global minimum. A target is declared
search-censored only if no tested structural candidate satisfies the
fixed-error criterion within $N\le 10^9$.
A final budget passes only if all 16 runs are operational, budget-feasible, and
finite and, with $D_j=D_{\mathrm{tr}}(\widehat\rho_j,\rho)$,
\begin{equation}
\frac1{16}\sum_{j=1}^{16}D_j\le0.05+10^{-12},\qquad
\#\{j:D_j\le0.05+10^{-12}\}\ge8.
\label{eq:supp-numerical-fixed-error-criterion}
\end{equation}

\subsection{Local-Pauli Linear-Inversion Baseline}

LP-LI measures all $3^n$ local settings $b\in\{X,Y,Z\}^n$, drawing a
$2^n$-outcome multinomial count vector per setting, without a structural prior.
In panel (a), write $N=q3^n+r$, $0\le r<3^n$: each setting receives $q$ shots,
and $r$ settings selected by a deterministic seeded permutation receive one extra,
using exactly $N$ copies.  The requested grid is
\begin{equation}
N\in\{10^k,5\times10^k:\ k=1,\ldots,10\}.
\label{eq:supp-numerical-lpli-budget-grid}
\end{equation}
Complete-design feasibility requires $N\ge3^n$; at $n=10$, $3^{10}=59049$,
so the first eight budgets through $5\times10^4$ are structurally infeasible,
not stochastic failures.  Pool Pauli outcome products from every setting compatible
with $P\in\Pc_n$, weighting by actual shot counts; for uniform $M$ shots per
setting the denominator is $M3^{n-\mathrm{wt}(P)}$.  Fix $\widehat r(I)=1$ and
form the linear estimate~\cite{AltepeterJamesKwiat2004}
\begin{equation}
\widehat\rho_{\mathrm{lin}}=2^{-n}\sum_{P\in\Pc_n}\widehat r(P)P.
\end{equation}
Hermitize it and project its eigenvalues in Euclidean distance onto the probability
simplex.  This Frobenius (Hilbert--Schmidt) projection onto positive-semidefinite
trace-one matrices supplies the physical estimate used for all errors.

For panel (b), uniform shots give $N_{\mathrm{LP\text{-}LI}}=M3^n$.
Use the same 16 prescribed seeds and criterion
\eqref{eq:supp-numerical-fixed-error-criterion}, including operational, finite,
budget-feasible runs.  Starting at integer $M=1$, bracket by factor-two growth.
An \(M^{-1/2}\) error-scaling estimate may be used
only as a heuristic to suggest a promising next value of \(M\); the suggested
value is still explicitly simulated before being classified as passing or
failing.  Once a bracket is found, we refine it by integer binary search.  We
then explicitly test the radius-two integer neighborhood around the detected
crossing, i.e., nearby values \(M_c-2,\ldots,M_c+2\) when they lie in the
allowed search domain.  Because finite-sample pass/fail outcomes need not be
strictly monotone in \(M\), we do not replace the measured outcomes by a
monotone envelope.  Instead, we report the smallest empirically passing
\(M\) among the tested values in the crossing neighborhood.  This quantity is
therefore an empirical threshold estimate, not a proof of the globally minimal
passing integer \(M\).
If none is found within
$M3^n\le5\times10^{10}$, the target is search-censored.

\subsection{Data Aggregation, Incomplete Runs, and Power-Law Fits}

For positive observations, use
\begin{equation}
\overline x_{\mathrm{geo}}=\exp\!\left(\frac1m\sum_{j=1}^m\log x_j\right),
\qquad
[\overline x_{\mathrm{geo}}e^{-\sigma_{\log x}},\,
\overline x_{\mathrm{geo}}e^{+\sigma_{\log x}}],
\label{eq:supp-numerical-geometric-error-bar}
\end{equation}
where $\sigma_{\log x}$ is the sample standard deviation of the logarithms,
with denominator $m-1$ for $m>1$.  Bars show empirical dispersion, not standard
errors or confidence intervals.  Panel (a) aggregates individual errors directly,
without first averaging within targets: $20\times17=340$ CEBP target--holdout
errors or $20\times6=120$ LP-LI target--replicate errors per budget and $d$.
All 15 CEBP budgets are complete, giving $15\times20\times17=5100$ evaluations
per $d$; all 12 feasible LP-LI budgets from $10^5$ through $5\times10^{10}$
have 120 reconstructions each.  Each smaller budget has 120 deterministic
infeasibility cases and no error point.  Panel (b) first determines one
$N_{\mathrm{emp},s}$ per target by its 16-seed decision, then aggregates complete
20-target ensembles as
$\overline N_{\mathrm{emp}}=\exp(\frac1{20}\sum_{s=1}^{20}\log N_{\mathrm{emp},s})$
with the same log-space dispersion; measurement seeds are not independent targets.

CEBP $d=1,2,3$ series contain all 20 thresholds at each $n$. 
LP-LI $d=1,2$ series are complete over $n=1,\ldots,9$ and $n=2,\ldots,9$;
$d=3$ is complete over $n=3,\ldots,8$, while $n=9$ has six validated thresholds
and 14 search-censored targets within the $5\times10^{10}$-copy domain.

Only CEBP data are fit, using unweighted ordinary least squares in base-10
logarithms.  Panel (a) fits $D_{\mathrm{tr}}=AN^p$ to geometric-mean errors
above $N\ge10^4,10^5,10^6$ for $d=1,2,3$, respectively (14, 12, and 10 points).
Panel (b) fits $\overline N_{\mathrm{emp}}=An^{p_n}$ over all complete,
uncensored sizes: $n=1,\ldots,9$ for $d=1$  ,$n=2,\ldots,9$ for $d=2$ and $n=3,\ldots,9$ for $d=3$.
\begin{center}
\begin{tabular}{ccrrr}
\hline
Panel & $d$ & Exponent ($p$ or $p_n$) & Standard error & $R^2_{\log_{10}}$ \\
\hline
(a) & 1 & $-0.58358$ & $0.01482$ & $0.99232$ \\
(a) & 2 & $-0.62451$ & $0.02835$ & $0.97980$ \\
(a) & 3 & $-0.62215$ & $0.02919$ & $0.98270$ \\
(b) & 1 & $1.79474$ & $0.08484$ & $0.98460$ \\
(b) & 2 & $2.15539$ & $0.10647$ & $0.98557$ \\
(b) & 3 & $1.97374$ & $0.21983$ & $0.94160$ \\
\hline
\end{tabular}
\end{center}
For panel (a), $N\propto D_{\mathrm{tr}}^{-p_D}$ with
$p_D=-1/p\simeq1.71,1.60,1.61$ for $d=1,2,3$, respectively, compared
qualitatively with the analytical $D_{\mathrm{tr}}^{-8}$ dependence.
These fits are finite-range empirical summaries, not asymptotic complexity
exponents; the analytical bounds are uniform worst-case guarantees.

\end{document}